\documentclass[runningheads]{llncs}

\usepackage[bookmarksdepth=3]{hyperref}

\usepackage[T1]{fontenc}
\usepackage{graphicx}

\usepackage{framed}
\usepackage{syntax}
\usepackage{mathtools}
\usepackage{amsmath}
\usepackage{amssymb}
\usepackage{amsfonts}
\usepackage{galois}
\usepackage{multirow}
\usepackage{stmaryrd}
\usepackage{listings}
\usepackage{cleveref}	
\usepackage[inline]{enumitem}
\newlist{inlinelist}{enumerate*}{1}
\setlist[inlinelist]{label=(\roman*)}

\usepackage{bussproofs}

\newcounter{csemantics}[section]
\makeatletter
\newcommand{\inflabel}[1]{
	\stepcounter{csemantics}
	\def\@currentlabel{(\arabic{section}.\arabic{csemantics})}
	\label{#1}
	\LeftLabel{(\arabic{section}.\arabic{csemantics}) }
}
\makeatother

\DeclarePairedDelimiterX\Set[1]{\lbrace}{\rbrace}{#1}
\newcommand{\ints}{\mathbb{Z}}
\newcommand{\nats}{\mathbb{N}}

\newcommand{\defn}{\triangleq}

\newcommand{\st}{\mathrel{.}}
\newcommand{\dom}[1]{dom(#1)}
\newcommand{\codom}[1]{codom(#1)}
\newcommand{\powerset}[1]{\wp(#1)}

\newcommand{\rarg}[1]{{\scalebox{0.6}{$#1$}}}
\DeclareRobustCommand{\dblupharpoon}{\mathrel{\ooalign{$\restriction$\cr\hidewidth\raisebox{0.5ex}{$\restriction$}\hidewidth\cr}}}
\DeclareRobustCommand{\doublerestriction}{\dblupharpoon}
\newcommand{\restrdom}[1]{\restriction_{\rarg{#1}}}               
\newcommand{\restrcodom}[1]{\mathord{\restriction^{\rarg{#1}}}}  
\newcommand{\restr}[2]{\restriction_{\rarg{#1}}^{\rarg{#2}}}      
\newcommand{\drestrdom}[1]{\doublerestriction_{\rarg{#1}}}         
\newcommand{\drestrcodom}[1]{\doublerestriction^{\rarg{#1}}}       
\newcommand{\drestr}[2]{\doublerestriction_{\rarg{#1}}^{\rarg{#2}}} 

\newcommand{\cupdot}{\mathbin{\dot{\cup}}\allowbreak}
\newcommand{\cupddot}{\mathbin{\ddot{\cup}}\allowbreak}
\newcommand{\bigcupdot}{\mathop{\dot{\bigcup}}}

\newcommand{\idfun}{\mathit{id}}

\newcommand{\code}[1]{\mbox{\lstinline@#1@}}
\newcommand{\freshCheap}{\functionWithSup{fresh}{\Cstate}{\Cheap}}
\newcommand{\new}[1][~]{\texttt{new}#1}
\usepackage{xcolor}

\newcommand{\Cpp}{{C\nolinebreak[4]\hspace{-.05em}\raisebox{.4ex}{\tiny\bf ++}}} 

\newcommand{\Cstate}{\mbox{\ensuremath{\mathcal{C}}}}
\newcommand{\Cheap}{\mbox{\ensuremath{\mathcal{H}}}}
\newcommand{\Cstack}{\mbox{\ensuremath{\mathcal{K}}}}
\newcommand{\cstate}{\mbox{\ensuremath{\chi}}}
\newcommand{\cheap}{\mbox{\ensuremath{\eta}}}
\newcommand{\cstack}{\mbox{\ensuremath{\kappa}}}
\newcommand{\Cset}{\mbox{\ensuremath{X}}} 

\newcommand{\Sval}{\mbox{\ensuremath{\mathcal{V}}}}    
\newcommand{\Smem}{\mbox{\ensuremath{\mathcal{M}}}}
\newcommand{\Sstate}{\mbox{\ensuremath{\mathcal{S}}}}
\newcommand{\sval}{\mbox{\ensuremath{\nu}}}     
\newcommand{\smem}{\mbox{\ensuremath{\mu}}}
\newcommand{\sstate}{\mbox{\ensuremath{\sigma}}}
\newcommand{\Sset}{\mbox{\ensuremath{\Sigma}}}

\newcommand{\decomposefun}{\mathsf{decompose}}
\newcommand{\composefun}{\mathsf{compose}}
\newcommand{\explodefun}{\mathsf{explode}}
\newcommand{\implodefun}{\mathsf{implode}}

\newcommand{\alphasplit}{\alpha_{{\scalebox{0.6}{$\Sstate$}}}}
\newcommand{\gammasplit}{\gamma_{{\scalebox{0.6}{$\Sstate$}}}}
\newcommand{\alphasplitdot}{\dot{\alphasplit}}
\newcommand{\gammasplitdot}{\dot{\gammasplit}}

\newcommand{\absdom}[1]{\ensuremath{#1}^{\sharp}}   

\newcommand{\Aval}{\absdom{\Sval}}            
\newcommand{\Amem}{\absdom{\Smem}}
\newcommand{\Astate}{\absdom{\Sstate}}
\newcommand{\aval}{\absdom{\sval}}         
\newcommand{\amem}{\absdom{\smem}}
\newcommand{\astate}{\absdom{\sstate}}

\newcommand{\MemID}{\mbox{\ensuremath{\mathsf{MemID}}}}
\newcommand{\MId}{\mbox{\ensuremath{\mathsf{I}}}}
\newcommand{\mId}{\iota}

\newcommand{\Sub}{\mbox{\ensuremath{\mathsf{Sub}}}}
\newcommand{\sub}{\mbox{\ensuremath{\mathsf{sub}}}}

\newcommand{\applysubfun}{\mathsf{applySub}}

\newcommand{\function}[2]{\mathsf{#1}_{\scalebox{0.6}{$#2$}}}   
\newcommand{\functionWithSup}[3]{\function{#1}{#2}^{\scalebox{0.6}{$#3$}}}   

\newcommand{\assumefunAV}{\function{assume}{\Aval}}
\newcommand{\assumefunAM}{\function{assume}{\Amem}}
\newcommand{\assignfunAV}{\function{assign}{\Aval}}
\newcommand{\assignfunAM}{\function{assign}{\Amem}}
\newcommand{\memidfunAM}{\function{memId}{\Amem}}
\newcommand{\isSumfun}{\function{isSum}{\Amem}}
\newcommand{\evalfun}{\function{eval}{\Aval}}
\newcommand{\pointedbyfun}{\function{addrOf}{\Amem}}
\newcommand{\pointstofun}{\function{pointsTo}{\Amem}}
\newcommand{\accessfun}{\function{access}{\Amem}}
\newcommand{\Rfun}{\function{R}{\Amem}}
\newcommand{\freshSsplit}{\functionWithSup{fresh}{\Amem}{\Cstack}}
\newcommand{\freshHsplit}{\functionWithSup{fresh}{\Amem}{\Cheap}}
\newcommand{\gammaA}{\function{\gamma}{\Astate}}
\newcommand{\gammaAV}{\function{\gamma}{\Aval}}
\newcommand{\gammaAM}{\function{\gamma}{\Amem}}
\newcommand{\gammaID}{\function{\gamma}{\MemID}}
\newcommand{\gammaIDF}{\function{\gamma}{\MemID f}}
\newcommand{\rhobase}{\rho_b}   
\newcommand{\rhofield}{\rho_f}  
\newcommand{\Avalmemid}{\Sval^{\scalebox{0.6}{\MemID}}}
\newcommand{\Avalcid}{\Sval^{\scalebox{0.6}{\cid}}}

\newcommand{\csem}[1]{\llbracket \hspace{1pt} #1 \hspace{1pt} \rrbracket}
\newcommand{\ccsem}[1]{\llbracket \hspace{1pt} #1 \hspace{1pt} \rrbracket^{\wp}}
\newcommand{\ssem}[1]{\llparenthesis \hspace{1pt} #1 \hspace{1pt} \rrparenthesis}
\newcommand{\cssem}[1]{\llparenthesis \hspace{1pt} #1 \hspace{1pt} \rrparenthesis^{\wp}}
\newcommand{\asem}[1]{\llparenthesis \hspace{1pt} #1 \hspace{1pt} \rrparenthesis^{\sharp}}

\newcommand{\sctext}[1]{\mbox{\normalfont\textsc{#1}}}

\newcommand{\mull}{\mu\mathsf{LL}}
\newcommand{\pexps}{\sctext{Pe}}
\newcommand{\pexp}{\mathsf{p}}
\newcommand{\vexps}{\sctext{Ve}}
\newcommand{\vexp}{\mathsf{v}}
\newcommand{\fexps}{\sctext{Fe}}
\newcommand{\fexp}{\mathsf{f}}
\newcommand{\hexps}{\sctext{He}}
\newcommand{\hexp}{\mathsf{h}}
\newcommand{\lhss}{\sctext{Lhs}}
\newcommand{\lhs}{\mathsf{lhs}}
\newcommand{\rhss}{\sctext{Rhs}}
\newcommand{\rhs}{\mathsf{rhs}}

\newcommand{\deref}{\texttt{*}}
\newcommand{\addrof}{\texttt{\&}}
\newcommand{\objupdate}[3]{#1[#2 \mapsto #3]}

\newcommand{\mujs}{\mu\mathsf{JS}}
\newcommand{\exps}{\sctext{e}}
\newcommand{\ski}{\texttt{skip}}
\newcommand{\oexps}{\sctext{Oe}}
\newcommand{\aexps}{\sctext{Ae}}
\newcommand{\bexps}{\sctext{Be}}
\newcommand{\sexps}{\sctext{Se}}

\newcommand{\stmts}{\sctext{Stmt}}

\renewcommand{\exp}{\mathsf{e}}
\newcommand{\oexp}{\mathsf{o}}
\newcommand{\aexp}{\mathsf{a}}
\newcommand{\bexp}{\mathsf{b}}
\newcommand{\sexp}{\mathsf{s}}

\newcommand{\stmt}{\mathsf{st}}

\newcommand{\true}{\texttt{true}}
\newcommand{\false}{\texttt{false}}

\newcommand{\str}[1]{\texttt{{"}#1\texttt{"}}}

\newcommand{\concat}[2]{\texttt{concat(}#1,#2\texttt{)}}

\newcommand{\ifc}[3]{\mbox{\texttt{if} #1 }\{\,#2\,\} \mbox{\texttt{else}} \{\,#3\,\}}
\newcommand{\while}[2]{\mbox{\texttt{while #1}} \{\,#2\,\}}

\newcommand{\cval}{\sctext{Val}\ifmmode\else\scriptsize\fi}
\newcommand{\cvalbar}{\overline{\sctext{Val}}}
\newcommand{\cobj}{\sctext{Obj}}
\newcommand{\caddr}{\sctext{Addr}}
\newcommand{\caddrs}{\sctext{Addr}\textsubscript{S}}
\newcommand{\caddrh}{\sctext{Addr}\textsubscript{H}}
\newcommand{\caddrr}{\sctext{Addr}\textsubscript{R}}
\newcommand{\caddrp}{\sctext{Addr}\textsuperscript{+}}
\newcommand{\cid}{\sctext{Id}}

\newcommand{\objV}{o}

\newcommand{\emptyobj}{\{\ \}}

\usepackage[dvipsnames]{xcolor}

\newif\ifdraft\drafttrue

\newcommand{\rremark}[2]{{\color{red}(#1: #2)}}
\newcommand{\gremark}[2]{{\color{green}(#1: #2)}}
\newcommand{\cremark}[2]{{\color{cyan}(#1: #2)}}

\newcommand{\oremark}[2]{{\color{orange}(#1: #2)}}
\newcommand{\bremark}[2]{{\color{blue}(#1: #2)}}

\ifdraft
	\newcommand{\gb}[1]{\cremark{GB}{#1}}
	\newcommand{\luc}[1]{\rremark{LN}{#1}}
	
	\newcommand{\todo}[1]{\oremark{TODO}{#1}}
	\newcommand{\future}[1]{\bremark{FUTURE}{#1}}
	\newcommand{\pf}[1]{\gremark{PF}{#1}}
\else
	\newcommand{\gb}[1]{}
	\newcommand{\va}[1]{}
	\newcommand{\luc}[1]{}
	\newcommand{\todo}[1]{}
	\newcommand{\future}[1]{}
	\newcommand{\pf}[1]{}
\fi

\newcommand{\tuple}[1]{\langle #1 \rangle}

\usepackage[mathlines]{lineno} 
\allowdisplaybreaks

\begin{document}

\title{A Modular Framework for Stack-Heap and Value Abstractions (Extended Version)}
\titlerunning{A Modular Framework for Stack-Heap and Value Abstractions (Ext. Ver.)}

\author{Giacomo Boldini\orcidID{0009-0006-4741-0033} \and
	Luca Negrini$^{*}$\orcidID{0000-0001-9930-8854} \and
	Luca Olivieri\orcidID{0000-0001-8074-8980} \and
	Pietro Ferrara\orcidID{0000-0002-4678-933X}}

\authorrunning{G. Boldini et al.}

\institute{
	Ca' Foscari University of Venice, Venice, Italy\\
	\email{\{giacomo.boldini, luca.negrini, luca.olivieri, pietro.ferrara\}@unive.it}\\
	$^{*}$ Corresponding author
}

\maketitle

\begin{abstract}
	Advanced static program analysis requires reasoning on the semantics of
	non-trivial program behaviors (e.g., pointers and complex data structures
	such as lists and sets, functions, and objects) and how they affect the memory. In
	most programming languages, static and dynamic allocations are typically
	managed by the stack and the heap, respectively. However, how allocations behave
	and how the memory is managed at runtime can vary significantly depending on the
	programming language being analyzed. Proper handling of these aspects is
	essential, as an accurate memory model enables the detection of critical
	issues such as buffer overflows and underflows, use-after-free errors, and
	null pointer exceptions prior to execution, that is, before such erroneous
	behaviors occur.

	In this paper, we propose and formalize a generic memory framework to
	handle stack and heap memory during the analysis, that is able to support
	various behaviors from different programming languages (e.g., C, \Cpp, Java,
	and Python), while remaining parametric, allowing different memory and
	value analyses to be independently chosen and combined. It relies on the
	Abstract Interpretation theory and enables sound approximation of different
	memory models and program behaviors. We introduce a \emph{split state}
	abstraction that separates value and memory analyses into two modular
	abstract domains. These domains interact through a set of \emph{memory identifiers},
	along with a set of operations defined by the domains to
	manipulate them, allowing the framework to capture both value information
	and structural memory relationships.

\end{abstract}

\section{Introduction}

During execution, a program uses memory to store and retrieve data. This memory
is typically split into the \emph{stack} and \emph{heap}. Stack memory is
used for \emph{static} allocation with a bounded size, such as local
variables within functions. In contrast, heap memory handles \emph{dynamic}
allocation, making it suitable for data and objects whose size can change at
runtime or that need to persist beyond the lifetime of a single function call.
Both kinds of memory can contain several types of values: from so-called
primitive types such as integers or floating point numbers, to arbitrarily
complex data structures such as lists, sets, functions, and objects.
Depending on the programming language, memory management can be completely
different. In languages like C or \Cpp, developers have complete control over
memory: allocations and deallocations are manual, and it is the developer's
responsibility to ensure that memory is used correctly (e.g., that a buffer is
only accessed within its allocated bounds, or that memory is deallocated when
it is no longer needed). In contrast, languages like Java or Python use
automatic memory management, where the runtime environment handles memory
allocation and deallocation, often through a garbage collection mechanism.

In static program analysis, precisely approximating the memory is crucial. An
accurate memory model enables detection of issues such as buffer overflows or
underflows, use-after-free, or null pointer dereferences. Some of these issues
also require a precise analysis of the values stored in memory, such as the
size of a buffer or the address of a pointer. Conversely, properties targeting
the values handled by the program, such as the absence of overflows or
arithmetic errors, require a precise analysis of the memory layout of the
program to retrieve the necessary information about values stored in memory.

Despite this close relation, the fields of memory analyses and value
analyses have developed orthogonally to one another, with several options to
pick for both~\cite{Andersen94,shape,Kanvar16HeapSurvey,MineLN,octagons,Cousot78}.
Any value analysis can be combined with any memory analysis, and
it is well-known that such a combination is required for any meaningful program
analysis~\cite{valueheap}.

Over the past decade, two main approaches have dominated the static analysis of programs:
\begin{inlinelist}
  \item analyzers that focus on value information, which first preprocess the program using a dedicated heap analysis and then replace heap accesses with symbolic variables for example Clousot~\cite{fahndrichclousot}, and
  \item heap abstraction techniques for example TVLA~\cite{lev2000tvla} that either do not track value information or require manual extensions such as adding specific predicates to capture particular value properties.
\end{inlinelist}
While recent work (see \Cref{sec:related}) has advanced the integration of heap and value analyses, truly generic and fully automatic analyzers, that is, approaches not bound to a specific heap or value domain and not requiring user provided annotations, are still limited.

\paragraph{Motivating Example.}
Before introducing the formal framework, we present a motivating example that highlights common challenges in heap memory analysis, including dynamic allocations, pointer aliasing, and complex structures.
Consider the code snippet written in C of \Cref{lst:dancing-c}, coming from the SV-COMP benchmark\footnote{\url{https://github.com/sosy-lab/sv-benchmarks/blob/master/c/heap-manipulation/dancing.c} Accessed (05/2026)}.
The code repeatedly allocates heap nodes around a stack-allocated object \texttt{list}. Two pointers, \texttt{x} and \texttt{tail}, are both initialised to \texttt{\&list} (lines $7$ and $8$), suggesting that both are intended to track the evolving structure.
One would expect both to be updated as the data structure grows. In practice, however, only \texttt{x} is ever updated (line $21$), while \texttt{tail} is never reassigned, \texttt{tail->R = n} always overwrites
\texttt{list.R} with the address of the latest node (line $18$).
Specifically,
after $k$ iterations \texttt{list.R} points to the $k$-th node,
while all earlier nodes (except possibly the one pointed to by \texttt{x}) are unreachable.
Similarly, \texttt{n->L = tail} sets every node's \texttt{L} field to \texttt{\&list} (line $16$),
so each heap node independently points back to \texttt{list}.
The pointer \texttt{x} may therefore alias \texttt{list} or any heap node
allocated during any previous iteration.

Verifying this kind of issues is non-trivial in program static analysis. Capturing such behaviours demands accurate stack and heap management, and without it an analysis can easily miss lost references or unintended aliasing. At the same time, it requires abstractions, but excessive over-approximation may render the analysis impractically imprecise, obscuring the very behaviours it aims to capture and thereby weakening its usefulness in practice.

\begin{lstlisting}[
  language=C,
  keywordstyle=\color{blue}\bfseries,
  numbers=left,
  caption={C slice from \texttt{dancing.c}.},
  label={lst:dancing-c}]
struct node { struct node *L; struct node *R; };

struct node list;
list.L=0;
list.R=0;

struct node *x=&list;
struct node *tail=&list;

while(__VERIFIER_nondet_bool())
  {
    struct node *n=malloc(sizeof(struct node));
    if(n==0)
      break;

    n->L=tail;
    n->R=0;
    tail->R=n;

    if(__VERIFIER_nondet_bool())
      x=n;
  }
\end{lstlisting}

\paragraph{Contributions.}
This work presents a modular framework for the combined static analysis of stack
and heap memory, formalised over $\mull$,
a small imperative language extending $\mujs$~\cite{Arceri2019Static} with
explicit pointers, address-of, dereference, field access, and heap allocation.
Our main technical contributions are as follows:
\begin{itemize}
	\item A formal operational semantics for $\mull$ modelling both
	      stack and heap, with support to novel language features introduced.
	\item A split state adapting the construction of whole value analysis proposed in Negrini~\cite{whole-value}. We reformulate the concrete state as a
	      Galois isomorphism separating value and address information with no loss of precision.
	\item A parametric abstract framework in which a value domain and a memory domain cooperate via a lightweight
	      interface and a substitution mechanism, inspired by generic object-oriented memory frameworks~\cite{Ferrara16Framework,ferrara2014generic}. Whereas those works target object-oriented languages with a single heap, our framework supports a broader range of memory models, including languages such as C, where diverse forms of memory manipulation are possible and execution is not constrained by object-oriented abstractions.
	\item A framework instantiation with non-relational numerical
	      domains~\cite{MineLN} and Andersen-style points-to analysis~\cite{Andersen94}.
\end{itemize}

\paragraph{Paper Structure.}
The structure of the paper is as follows. \Cref{sec:related} deals with related work. \Cref{sec:background} provides
the necessary background information. \Cref{sec:mull-language} introduces
the $\mull$ language and its semantics.
\Cref{sec:split-state} and \Cref{sec:abstract-state} present our main contribution defining and formalizing the split state and abstract state. \Cref{sec:instantiations} demonstrates the modularity and practical applicability of the proposed framework by instantiating it with concrete abstract domains. \Cref{sec:conclusions} concludes the paper.

\section{Related Work}\label{sec:related}

In this section, we provide a brief overview of existing frameworks and tools
for memory management. The literature can be divided into approaches tailored
to specific memory problems and more general-purpose frameworks. For instance,
Gopan et al.~\cite{Gopan1,Gopan2} propose a solution specialized for numerical
analysis and array operations, which inherently require reasoning about both
heap and value domains, since arrays combine memory structure with indexes and
content-dependent properties. However, these kinds of problems are orthogonal
to our study, since we focused on how to combine heap and value analyses for
more general purposes, data structures, and memory models.

The broader difficulty of reasoning about memory and pointers has long been
recognized. As early as 2001, Hind~\cite{Hind01} observed that, despite the
substantial body of research available at the time, memory management and
pointer-related issues remained complex and challenging problems. This
observation continues to motivate subsequent work in the area. In particular,
Kanvar et al.~\cite{Kanvar16HeapSurvey} provide a systematic taxonomy of heap
memory models used in static analysis, classifying them as storeless,
store-based, and hybrid, and surveying summarization techniques such as
k-limiting, allocation-site based, pattern-based, variable-based, and other
instrumentation predicates, as well as higher-order logic approaches.

About abstract interpretation approaches, to the best of our knowledge, only a
limited number of prior works have addressed a general integration of heap and
value management. Chang and Rival~\cite{Chang1} proposed an abstract domain
that integrates data abstraction with user-provided specifications of data
structures. In contrast, our approach is based on heap identifiers and does not
require explicit specifications of data structures. They also introduced a
modular integration of shape and numerical abstract domains~\cite{Chang2},
where shape analysis is based on points-to predicates, and the numerical domain
tracks information using a symbolic representation of data stored in the heap.
This differs slightly from our notion of memory identifiers, which are designed
to abstract memory locations. Moreover, their approach focuses on shape
analyses based on summarization and materialization of nodes. As a result, when
a node is materialized, the analysis requires to maintain a disjunctive
abstraction represented by a set of possible shapes. Blazy et al.~\cite{Blazy}
present an \emph{abstract memory functor} inspired by the memory model of C
programs. Basically, they design an abstract memory domain that is able to
handle points-to information and it is parameterized by a numerical abstract
domain. Then, they implemented it into a formally verified static analyzer that
is not dealing precisely with memory~\cite{blazy2013formal}. Their evaluation
considers both non-relational (intervals~\cite{CousotCousot77}) and relational
(polyhedra~\cite{Cousot78}) numerical abstract domains. Miné~\cite{celldomain}
proposes a memory abstraction that translates the semantics of all memory
accesses into a semantics over a set of scalar variables. In this abstraction,
a cell represents a memory region holding a scalar value together with its
type. The approach is parameterized by a numerical domain, but it does not
support summary nodes or dynamic allocation. It was initially implemented in
ASTRÉE~\cite{astree}, and has more recently been adopted and extended in
MOPSA~\cite{mopsa}. To handle pointer dereferences, MOPSA builds on the cell
domain, and relies on recency abstractions~\cite{Balakrishnan06} to model
dynamic memory allocation. For each allocation site, recency abstraction
distinguishes the most recently allocated block from all previous ones, which
are summarized into a single weak memory block. Allocation sites are
configurable~\cite{Monat20} and are typically defined by program locations.
However, this summarization can reduce precision, then alternative abstractions
are required to preserve separation between memory blocks during loop
unrolling, as in the analysis of C programs~\cite{mopsac24}. Milanese and
Miné~\cite{milanese2024under} also proposed an under-approximation backward
analysis for the memory model of MOPSA to enhance both precision and efficiency
in the analysis of C programs. Recency abstraction has also been applied for
other languages and specific contexts. For instance, Calzavara et
al.~\cite{Calzavara17} leverage ideas from recency abstraction to design a
sound, flow-sensitive heap abstraction for the static analysis of Android
applications written in Java language.

Finally, our work builds upon Ferrara~\cite{Ferrara16Framework,ferrara2014generic}
and Negrini~\cite{whole-value}. The main idea of~\cite{Ferrara16Framework,ferrara2014generic} is to
introduce an intermediate abstraction for heap and value analyses, by defining
a split state, which rewrites the concrete state (i.e., program memory and
concrete semantics) to separate values of different data types into distinct
maps. The heap analysis aims to abstract the heap structure, while the value
domain abstracts the values of non-reference variables, as well as heap
locations. Intuitively, the heap analysis approximates concrete locations
through heap identifiers, while the value analysis tracks information on these
identifiers. This provides several advantages: \begin{inlinelist} \item a more
	convenient state representation that simplifies definitions and proofs, \item
	a modular design where existing abstract domains can be applied independently
	to each sub-memory, reusing standard operations with minimal changes, and
	\item semantic preservation, since proving soundness/completeness is equivalent
	between the split and concrete states. \end{inlinelist} We extend this
setting by widening the set of supported operations to arbitrary references
and dereferences, and to dynamic field accesses. This broadens the scope
of the framework beyond object-oriented languages.
Moreover, we improve on its formalization and soundness proofs
borrowing the isomorphism approach from~\cite{whole-value}. In particular,
instead of building a Galois Connection between the concrete and split states,
followed by a second Galois Connection between the split and abstract states,
we enforce a Galois Isomorphism between the concrete and split states, ensuring
that no loss in precision occurs when moving from the concrete to the split state.
This is a key weakness of the original framework, that instead has two possible
sources of imprecision.

\section{Background}
\label{sec:background}

\paragraph{Sets.}
A set $X$ is a (possibly infinite) collection of elements.
We write $x \in X$ for membership, $\emptyset$ for the empty set,
$|X|$ for cardinality, and $\wp(X)$ for the powerset.
Given sets $X$ and $Y$, we use standard notation:
$X \subseteq Y$ for subset, $X \cup Y$ for union, $X \cap Y$ for intersection,
$X \setminus Y$ for set difference, and
$X \times Y = \{(x, y) \mid x \in X \land y \in Y\}$ for the Cartesian product.
Sets can be defined extensionally, e.g., $\{x_0, x_1, \dots\}$,
or in terms of a predicate, e.g., $\{x \mid \phi(x)\}$, meaning all $x$ such
that a predicate $\phi(x)$ holds.

\paragraph{Functions.}
A function $f: X \to Y$ is a subset of the Cartesian product $X \times Y$
such that $\nexists (x, y), (z, w) \in f : x = z \land y \neq w$, meaning
that each $x \in X$ is associated with at most one $y \in Y$.
The set $X$ is called \emph{domain}, denoted as $\dom{f}$, the set $Y$ is
called \emph{co-domain}, denoted as $\codom{f}$. The identity function
is denoted by $\idfun$.
Similarly to sets, functions can be defined extensionally, e.g.,
$\{(x_0, y_0), (x_1, y_1), \dots\}$, or in terms of a predicate,
e.g., $\{(x, y) \mid \phi(x, y)\}$, meaning all pairs $(x, y)$ such
that a predicate $\phi(x, y)$ holds.
Given a function $f$ and two elements $x$ and $y$, the update operator
$f[x \mapsto y]$ denotes $f \setminus \{(x, f(x))\} \cup \{(x,y)\}$
when $x \in \dom{f}$, and $f \cup \{(x,y)\}$ otherwise.
Abusing notation, we also use chained updates
$f[x_1 \mapsto y_1, \dots, x_n \mapsto y_n]$, which denotes successive
application of the function update, i.e.,
$f^0 = f, f^i = f^{i-1}[x_i \mapsto y_i]$, whose result is $f^n$.

\paragraph{Ordered Structures.}
A set $X$ with a partial ordering relation $\sqsubseteq_X\, \subseteq X \times X$
is a poset, denoted by $\tuple{X,\sqsubseteq_X}$. If a poset has a
bottom element $\bot_X$ and is closed under finitary applications of the least
upper bound (lub, $\sqcup_X$) operator of $X$, it is called a complete partial
order (cpo), denoted as $\tuple{X,\sqsubseteq_X, \sqcup_X , \bot_X}$. Moreover,
a lattice $\tuple{X, \sqsubseteq_X, \sqcup_X, \sqcap_X}$ is a poset having a
minimum element (bottom, $\bot_X \in X$), a maximum element (top, $\top_X \in X$)
and closed under finitary applications of the least upper bound (lub,
$\sqcup_X$) and the greatest lower bound (glb, $\sqcap_X$) operators. A
\emph{complete} lattice is closed under arbitrary lub and glb, so that
$\bigsqcup\,Y \in X$ and $\bigsqcap\,Y \in X$, for all $Y \subseteq X$, and it
is denoted as $\tuple{X,\sqsubseteq_X, \sqcup_X, \sqcap_X, \top_X, \bot_X}$.
Provided there is no ambiguity, we will omit subscripts of each operator for
clarity. Complete lattices can be derived from other structures. For instance,
given a set $X$, $\tuple{\wp(X), \subseteq, \cup, \cap, X, \emptyset}$ is a
complete lattice since $\subseteq$ is a partial ordering relation, $\cup$ and
$\cap$ are closed w.r.t. $\wp(X)$, and
$\forall Y \in \wp(X) \;:\; \emptyset \subseteq Y \subseteq X$.
By duality, $\tuple{\wp(X), \supseteq, \cap, \cup, \emptyset, X}$ is also
complete. Moreover, given $\tuple{X, \sqsubseteq, \sqcup, \sqcap, \top, \bot}$
and a set $Y$, the \emph{functional lift}~\cite{CousotCousot79-1} of $X$
w.r.t. $Y$ is the complete lattice
$\tuple{Y \to X, \dot{\sqsubseteq}, \dot{\sqcup}, \dot{\sqcap}, \dot{\bot}, \dot{\top}}$
of total functions $Y \to X$, that is, of functions defined on all
elements of $Y$. Lattice operators are defined as point-wise applications of
operators over $X$ on all $y \in Y$. Lastly, given a finite set of complete
lattices $\tuple{Y_i, \sqsubseteq_{Y_i}, \sqcup_{Y_i}, \sqcap_{Y_i}, \bot_{Y_i}, \top_{Y_i}}, i \in \Delta \subset \nats$,
their Cartesian product~\cite{Cousot21} is the complete lattice
$\tuple{\bigtimes_{i \in \Delta} Y_i, \sqsubseteq, \sqcup, \sqcap, \bot, \top}$,
where lattice operators
are component-wise applications of the operators over each $Y_i$. Given
a poset $\tuple{X,\sqsubseteq_X}$, an increasing chain $C \subseteq X$ is a
possibly infinite sequence of elements $x_0, x_1,\dots$ of $X$ such that
$x_0 \sqsubseteq_X x_1 \sqsubseteq_X \dots$.

\paragraph{Abstract Interpretation.}
Abstract Interpretation~\cite{CousotCousot77,Cousot21} is
a theoretical framework for sound reasoning on semantic properties of a
program, establishing a correspondence between the semantics of a program,
called concrete semantics, and an approximation of it, called abstract
semantics. Let $C$ and $A$ be complete lattices, a pair of functions
$\alpha: C \rightarrow A$ and $\gamma: A \rightarrow C$ forms a \emph{Galois Connection}
between $C$ and $A$, written
$\tuple{C, \sqsubseteq_C} \galois{\alpha}{\gamma} \tuple{A, \sqsubseteq_A}$,
if $\forall\, c \in C, a \in A : \alpha(c) \sqsubseteq_A a \Leftrightarrow c \sqsubseteq_C \gamma(a)$.
Equivalently, the Galois Connection exists if $\alpha \circ \gamma$ is reductive (i.e., if
$\forall\, a \in A \;:\; \alpha \circ \gamma(a) \sqsubseteq_A a$), and
$\gamma \circ \alpha$ is extensive (i.e., if
$\forall\, c \in C \;:\; c \sqsubseteq_C \gamma \circ \alpha(c)$).
In addition, if $\alpha \circ \gamma = id$, then $\alpha$ and $\gamma$ form a
\emph{Galois Embedding}, written
$\tuple{C, \sqsubseteq_C} \galoiS{\alpha}{\gamma} \tuple{A, \sqsubseteq_A}$,
where no two abstract elements have the same concretization. Furthermore, if
also $\gamma \circ \alpha = id$, then $\alpha$ and $\gamma$ form a
\emph{Galois Isomorphism}, written
$\tuple{C, \sqsubseteq_C} \GaloiS{\alpha}{\gamma} \tuple{A, \sqsubseteq_A}$,
where $A$ is simply a
reshaping of $C$ and no abstraction (i.e., loss of precision) happens. Note
that Abstract Interpretation can be employed also when a Galois Connection does
not exist: in fact, it is sufficient that $C$ and $A$ are complete partial
orders, and that a monotone concretization $\gamma$ exists.

\paragraph{Soundness.}
Given $\tuple{C, \sqsubseteq_C} \galois{\alpha}{\gamma} \tuple{A, \sqsubseteq_A}$,
a concrete function $f: C \rightarrow C$ is, in general, not
computable. Hence, a function $f^\sharp : A \rightarrow A$ that must
\emph{correctly} approximate the function $f$ is needed. If so, we say that
the function
$f^\sharp$ is \emph{sound}. Given $\tuple{C, \sqsubseteq_C} \galois{\alpha}{\gamma} \tuple{A, \sqsubseteq_A}$
and a concrete function $f : C \rightarrow C$, an abstract function
$f^\sharp : A \rightarrow A$ is sound w.r.t. $f$ if
$\forall c \in C\;:\;\alpha(f(c)) \sqsubseteq_A f^\sharp(\alpha(c))$,
or equivalently $\forall a \in A\;:\; f(\gamma(a)) \sqsubseteq_C \gamma(f^\sharp(a))$.
Note that the latter relation can be used
to prove soundness even when a Galois Connection does not exist.

\paragraph{Abstract Domains.}
In the Abstract Interpretation framework, abstractions are defined through
so-called \emph{abstract domains}. These are composed by a partial
order $\tuple{X,\sqsubseteq_X}$, possibly extended to a complete partial order,
a lattice, or a complete lattice, an upper bound operator $\sqcup_X$ on $X$, a
bottom element $\bot_X \in X$, a widening operator $\nabla_X$ that
over-approximates $\sqcup_X$ and ensures the convergence on increasing chains,
an \emph{abstract transformer} $\asem{\stmt} \colon X \to X$ for evaluating
statements, and an \emph{abstract transformer} $\asem{\texttt{ASM}(\bexp)} \colon X \to X$
for traversing conditions. The purpose of the transformers is to evolve an
instance of the domain according to the semantics of the statement $\stmt$ to
be executed, and to refine an instance of the domain assuming that a condition
$\bexp$ holds.

\paragraph{Disjoint-Domain Unions.}
For functions $f\colon A\to Y$ and $g\colon B\to Z$ with
$A\cap B=\emptyset$, we write $f\cupdot g: A\cup B \to Y\cup Z$ for
their disjoint-domain union, defined by
$(f\cupdot g)(x)=f(x)$ if $x\in A$, and $(f\cupdot g)(x)=g(x)$ if $x\in B$.

For second-order maps $f\colon A \to (B\to C)$ and
$g\colon D \to (E\to F)$, we write
$f\cupddot g: A \cup D \to (B \cup E) \to (C \cup F)$ for
their second-order disjoint-domain union, defined when inner maps are
domain-disjoint at overlapping outer keys:
$\forall x\in A\cap D:\ f(x) \cap g(x)=\emptyset$.
For each outer key $x\in A\cup D$,
$(f\cupddot g)(x)=f(x)$ if $x\in A\setminus D$,
$(f\cupddot g)(x)=g(x)$ if $x\in D\setminus A$,
and $(f\cupddot g)(x)=f(x)\cupdot g(x)$ if $x\in A\cap D$.

\paragraph{Restrictions on Maps.}
\label{par:restriction-properties}
For $f: K \to V$, $X \subseteq K$, and $Y \subseteq V$, we
define the restriction on the domain as
$f\restrdom{X} = \{(x,f(x)) \mid x\in X\}$,
the restriction on the co-domain as
$f\restrcodom{Y} = \{(x,f(x)) \mid x\in K,\ f(x)\in Y\}$, and the general restriction as
$f\restr{X}{Y} = \{(x,f(x)) \mid x\in X,\ f(x)\in Y\}$.
For second-order maps $f: K \to (L \to V)$, with $X \subseteq K$ and $Y \subseteq V$,
we define the outer-domain restriction as
$f\drestrdom{X} = \{(k, f(k)) \mid k \in X\}$,
the inner-co-domain restriction as
$f\drestrcodom{Y} = \{(k, f(k)\restrcodom{Y}) \mid k \in K\}$,
and the combined double restriction as
$f\drestr{X}{Y} = \{(k, f(k)\restrcodom{Y}) \mid k \in X\}$.
In both cases, the combined form equals any composition of the two:
$f\restr{X}{Y} = f\restrdom{X}\restrcodom{Y} = f\restrcodom{Y}\restrdom{X}$,
and $f\drestr{X}{Y} = f\drestrdom{X}\drestrcodom{Y} = f\drestrcodom{Y}\drestrdom{X}$,
which follow trivially from the definitions.

Let $f\colon K\to V$, $g\colon A\to B$ with $K\cap A=\emptyset$,
$X_1,X_2\subseteq K$, $Y_1,Y_2\subseteq V$ with $Y_1\cap Y_2=\emptyset$, and
second-order maps $M_1\colon D_1\to(L_1\to V_1)$, $M_2\colon D_2\to(L_2\to V_2)$
with $\forall k\in D_1\cap D_2: \dom{M_1(k)}\cap\dom{M_2(k)}=\emptyset$.
In this setting, the following properties hold:
\begin{inlinelist}
	\item $f\restrdom{X_1}\cup f\restrdom{X_2} = f\restrdom{X_1\cup X_2}$ (\emph{domain additivity});
	\item $f\restrcodom{Y_1}\cup f\restrcodom{Y_2} = f\restrcodom{Y_1\cup Y_2}$ (\emph{co-domain additivity});
	\item $(f\cupdot g)\restr{X}{Z} = (f\restr{X\cap K}{Z})\cupdot(g\restr{X\cap A}{Z})$
	for any $X\subseteq K\cup A$, $Z\subseteq V\cup B$ (\emph{distribution of $\restriction$ over $\cupdot$});
	\item $(M_1\cupddot M_2)\drestrcodom{Y} = M_1\drestrcodom{Y}\cupddot M_2\drestrcodom{Y}$,
	for any $Y \subseteq V_1 \cup V_2$ (\emph{distribution of $\doublerestriction$ over $\cupddot$}).
\end{inlinelist}
Properties \emph{(i)} and \emph{(iii)} also hold with $\doublerestriction$ in place of $\restriction$.
Property \emph{(ii)} holds for $\doublerestriction$ as
$f\drestrcodom{Y_1}\cupddot f\drestrcodom{Y_2} = f\drestrcodom{Y_1\cup Y_2}$
(requiring $Y_1\cap Y_2=\emptyset$).
Proofs in \Cref{sec:proof-restriction-properties}.

\section{The \texorpdfstring{$\mull$}{muLL} Language}
\label{sec:mull-language}

The $\mull$ language is a small, imperative language for studying memory
and pointer abstractions. It extends the core of $\mujs$~\cite{Arceri2019Static}
with a new set of features for explicit memory management and pointer
manipulation, yielding a simple model that is closer to a \Cpp-like memory
model (chosen because it includes pointers, clearly defines when memory is
created and accessed, and remains expressive enough to cover a wide range of
programming languages), while remaining compact for formal treatment.

The syntax of $\mull$ is defined in \Cref{fig:mull-syntax}.
The core expressions of $\mujs$ for arithmetic ($\aexps$), boolean
($\bexps$), string ($\sexps$), and object ($\oexps$) expressions are
mostly retained, and the language is extended with memory- and
pointer-specific constructs. In particular, we introduce
\begin{inlinelist}
	\item
	dynamic memory allocation through the \new construct, used to allocate
	memory on the heap for primitive values and objects,
	\item
	the address-of operation $\addrof x$, which retrieves the memory
	address of a variable $x$,
	\item
	the dereference operation $\deref\pexp$, which accesses
	the value stored at the memory address denoted by $\pexp$,
	\item
	static field access $\pexp.f$, which allows direct access to a field
	of an object, and
	\item
	dynamic field access $\pexp[\sexp]$, which allows access to a field
	whose name is determined by evaluating a string expression.
\end{inlinelist}
The last four additions form the class of \emph{pointer expressions} ($\pexps$).

For presentation purposes, we introduce some syntactic subgroups of
expressions. When objects are first defined, their fields can only contain
field expressions ($\fexps$), which may include both primitive expressions and pointer
expressions. Nested objects can then be added through field assignments.
Heap allocation via $\new$ is restricted to
heap-allocable expressions ($\hexps$), namely primitive expressions
and objects. The non-terminal $\lhss$ denotes the set of assignable expressions.
Although pointer expressions are recursive, only a syntactically restricted
subset of them may occur as left-hand sides. In particular, $\lhss$ admits variables,
field accesses, and dereferences, but requires the base of any field access
to be an identifier, so write-mode targets are always directly addressable.
Read-mode field access, by contrast, admits arbitrary pointer expressions as base,
covering chains of intermediate dereferences.
Since $\mull$ has no \Cpp-style arrow operator (\texttt{->}), reading a field through a
pointer requires an explicit dereference, as in \texttt{(*p).f}.
Writing a field of a heap object is not directly expressible with dot notation; it requires
an explicit load, field update, and write-back via dereference.
This distinction will be formalized in the semantics.
The non-terminal $\vexp$ groups expressions of primitive types,
while $\exp$ denotes the set of all expressions of the language, namely
those that may occur on the right-hand side of a standard (non-\new) assignment.
We assume $\mull$ programs to be \emph{statically typed}.

As discussed in the assignment semantics (\Cref{sec:concrete-assign}),
the \new operation allocates memory on the heap, while standard
assignment stores values on the stack.

\begin{figure}[t]
	\small
	\begin{framed}
		\vbox{
			\setlength{\grammarparsep}{3pt plus 1pt minus 1pt}
			\setlength{\grammarindent}{4em}
			\renewcommand{\syntleft}{}  \renewcommand{\syntright}{}
			\begin{grammar}

				<$\aexp \in$ \aexps> $\Coloneqq$
				$x$
				~|~ $n$
				~|~ $\aexp$ \texttt{+} $\aexp$
				~|~ $\aexp$ \texttt{-} $\aexp$
				~|~ $\aexp$ \texttt{*} $\aexp$
				~|~ $\aexp$ \texttt{/} $\aexp$

				<$\bexp \in$ \bexps> $\Coloneqq$
				$x$
				~|~ \true
				~|~ \false
				~|~ $\bexp$ \texttt{\&\&} $\bexp$
				~|~ $\bexp$ \texttt{||} $\bexp$
				~|~ \texttt{!} $\bexp$
				~|~ $\aexp$ \texttt{\textless} $\aexp$
				~|~ $\aexp$ \texttt{==} $\aexp$
				~|~ $\sexp$ \texttt{==} $\sexp$

				<$\sexp \in$ \sexps> $\Coloneqq$
				$x$
				~|~ \str{$s$}
				~|~ \concat{$\sexp_1$}{$\sexp_2$}

				<$\oexp \in$ \oexps> $\Coloneqq$
				$x$
				~|~ \texttt{\{} \texttt{\}}
				~|~ \texttt{\{} $f_0$ \texttt{:} $\fexp_0$, $f_1$ \texttt{:} $\fexp_1$, \dots, $f_n$ \texttt{:} $\fexp_n$  \texttt{\}}

				<$\pexp \in$ \pexps> $\Coloneqq$
				$x$
				~|~ $\pexp[\sexp]$
				~|~ $\pexp.f$
				~|~ $\addrof x$
				~|~ $\deref\pexp$

				<$\vexp \in$ \vexps> $\Coloneqq$
				$\aexp$ ~|~ $\bexp$ ~|~ $\sexp$

				<$\fexp \in$ \fexps> $\Coloneqq$
				$\vexp$ ~|~ $\pexp$

				<$\hexp \in$ \hexps> $\Coloneqq$
				$\vexp$ ~|~ $\oexp$

				<$\exp \in$ \exps> $\Coloneqq$
				$\vexp$ ~|~ $\pexp$ ~|~ $\oexp$

				<$\lhs \in$ \lhss> $\Coloneqq$
				$x$
				~|~ $x[\sexp]$
				~|~ $x.f$
				~|~ $\deref\pexp$

				<$\stmt \in$ \stmts> $\Coloneqq$
				$\stmt$ $\stmt$
				~|~ \ski \texttt{;}
				~|~ $\lhs$ = $\exp$\texttt{;}	
				~|~ $\lhs$ = \new $\hexp$\texttt{;} 
				~|~ \ifc{$\bexp$}{$\stmt$}{$\stmt$}
				~|~ \while{$\bexp$}{$\stmt$}
			\end{grammar}
			where $x \in \cid$ (identifiers), $n \in \ints$,
			$s \in \Sigma^*$, $f_0, f_1, \dots, f_n \in \cid$.
		}
	\end{framed}
	\normalsize
	\vspace*{-10pt}
	\caption{$\mull$ syntax.}
	\label{fig:mull-syntax}
\end{figure}

\subsection{Concrete Semantics}
\label{sec:concrete-semantics}

In $\mull$, expressions evaluate to a value $v$ in the set
$\cvalbar \defn \cval \cup \caddr \cup \cobj \cup \{\bot\}$, containing
values ($\cval \defn \mathbb{Z} \cup \{\true, \false\} \cup \Sigma^*$),
memory addresses ($\caddr$), objects ($\cobj$), and a special value
$\bot$ to denote the result of invalid computations (e.g., an access to
a non-existing field of an object, the dereference of an invalid memory
address, or the division by zero). $\mathbb{Z}$ denotes the set of integers,
while $\Sigma^*$ denotes the set of strings over the alphabet $\Sigma$.
The $\bot$ value is propagated during expression
evaluation: if one of the arguments of an expression is $\bot$, then the
result of the whole expression is $\bot$.
Identifiers are elements of the set $\cid$ and serve as both variable
names and object field names; they follow the \Cpp~naming convention
(strings of letters, digits, and underscores, starting with a letter or underscore).
Addresses distinguish between memory locations in the stack and in
the heap, namely
$\caddr \defn \caddrs \cup \caddrh$, such that $\caddrs \cap \caddrh = \varnothing$.
$\caddrs$ denotes the set of stack memory addresses, where each
variable $x \in \cid$ is associated with a unique stack address
$a_x$, thus $\caddrs = \{ a_x \mid x \in \cid \}$.
$\caddrh$ are memory locations in the
heap, represented as $h_n$ for each $n \in \mathbb{N}$,
thus $\caddrh = \{ h_n \mid n \in \mathbb{N} \}$.
The objects $\cobj$ are structured values, represented as partial
maps from field names to values, addresses, or other objects,
$\cobj : \cid \rightarrow \cvalbar$. Thus, we can represent an object literal as
$\{f_0:v_0, f_1:v_1, \dots, f_n:v_n\}$, where for $i \in [0,n]$, $f_i \in \cid$ are field
names and $v_i \in \cvalbar$ are the values associated with the fields.
The notation $\emptyobj$ represents the empty object.
We write $\objupdate{\objV}{f}{v}$ for the standard partial map override of $\objV$
at field $f$ with value $v$, leaving all other fields unchanged.
We define the field access operator as $\objV(f) \defn v$ if $f \in \dom{\objV}$
and $\bot$ otherwise, that retrieves the value of field $f$ in object $\objV$.

Program memories $\cstate = \langle \cstack, \cheap \rangle \in \Cstate \defn \Cstack \times \Cheap$
are tuples composed of a stack $\cstack$ (i.e., a map from identifiers to values, memory
addresses, or objects, namely $\Cstack \defn \cid \rightarrow \cvalbar$)
and a heap $\cheap$ (i.e., a map from
heap memory addresses to values, addresses, or objects, namely
$\Cheap \defn \caddrh \rightarrow \cvalbar$).
The lattice structure is given by $\langle \powerset{\Cstate}, \subseteq \rangle$.

By static typing, each entry of $\cstack$ and $\cheap$ consistently holds values of
a single type throughout execution (a primitive value $\cval$, an address
$\caddr$, or an object $\cobj$) and never switches between them.
Abusing notation, we denote as
$\bot \in \Cstate$ an invalid state resulting from the assignment of
a variable to $\bot$, that is propagated from one statement to
the next one.

The concrete semantics of $\mull$ statements is formalized by
function $\csem{\stmt}: \Cstate \rightarrow \Cstate$.
Abusing notation, we use
$\csem{\exp}: \Cstate \rightarrow \cvalbar$ to capture
the semantics of expressions, which are evaluated without side effects on the
state (i.e., expression semantics does not modify the stack or heap).
Note that, for all $\stmt \in \stmts$, $\csem{\stmt}\bot = \bot$
and $\csem{\stmt}(\cstack, \cheap) = \bot$ if there exists
at least one sub-expression $\exp \in \exps$ in $\stmt$ such that
$\csem{\exp}(\cstack, \cheap) = \bot$. To avoid cluttering
the notation, we only define $\csem{~}$ on well-formed
expressions and statements, that is, ones whose evaluation does not produce $\bot$ (i.e.,
where none of the arguments are $\bot$ and that produce a valid result).

\begin{figure}[t]
	\begin{equation}
		\label{eq:c-obj}
		\frac{
			\forall i \in [0,n]\st \csem{\fexp_i} (\cstack, \cheap) = v_i
		}{
			\csem{\{f_0:\fexp_0, \dots, f_n:\fexp_n\}} (\cstack, \cheap) = \{f_0:v_0, \dots, f_n : v_n\}
		}
	\end{equation}

	\begin{multicols}{2}
		\begin{equation}
			\label{eq:c-empty-obj}
			\frac{}{
				\csem{\emptyobj} (\cstack, \cheap) = \emptyobj
			}
		\end{equation}

		\begin{equation}
			\label{eq:c-var}
			\frac{
				x \in \dom{\cstack}
			}{
				\csem{x} (\cstack, \cheap) = \cstack(x)
			}
		\end{equation}

		\begin{equation}
			\label{eq:c-addr}
			\frac{
				x \in \dom{\cstack}
			}{
				\csem{\addrof x} (\cstack, \cheap) = a_x
			}
		\end{equation}

		\begin{equation}
			\label{eq:c-stack-deref}
			\frac{
				\csem{\pexp} (\cstack, \cheap) = a_y \in \caddrs
			}{
				\csem{\deref\pexp} (\cstack, \cheap) = \cstack(y)
			}
		\end{equation}

		\begin{equation}
			\label{eq:c-heap-deref}
			\frac{
				\csem{\pexp} (\cstack, \cheap) = a \in \caddrh
			}{
				\csem{\deref \pexp} (\cstack, \cheap) = \cheap(a)
			}
		\end{equation}

		\begin{equation}
			\label{eq:c-static-access}
			\frac{
				\csem{\pexp} (\cstack, \cheap) = o \in \cobj
			}{
				\csem{\pexp.f} (\cstack, \cheap) = o(f)
			}
		\end{equation}

		\begin{equation}
			\label{eq:c-dynamic-access}
			\frac{
				\csem{\pexp} (\cstack, \cheap) = o \in \cobj
				\quad \csem{\sexp} (\cstack, \cheap) = f
			}{
				\csem{\pexp[\sexp]} (\cstack, \cheap) = o(f)
			}
		\end{equation}
	\end{multicols}

	\begin{equation}
		\label{eq:c-div}
		\frac{
			\csem{\aexp_1} (\cstack, \cheap) = n_1
			\quad \csem{\aexp_2} (\cstack, \cheap) = n_2
			\quad n_2 \neq 0
			\quad n_1\ \mathrm{mod}\ n_2 = 0
		}{
			\csem{\aexp_1 / \aexp_2} (\cstack, \cheap) = n_1 / n_2
		}
	\end{equation}
	\caption{Semantics of relevant expressions in $\Cstate$.}
	\label{fig:concrete-expr}
\end{figure}

\subsubsection{Relevant Expressions.}
\label{sec:concrete-expr}

The semantic rules for relevant expressions are summarized in \Cref{fig:concrete-expr}.
The evaluation of object expressions involves evaluating the fields of
the object, associating them with their corresponding string keys. The semantics of
a non-empty object \eqref{eq:c-obj} evaluates each expression
$\vexp_i$ to its concrete value $v_i$ and assigns it to their respective key
$f_i$ in the map, while the semantics of
an empty object \eqref{eq:c-empty-obj} evaluates to an empty map.
The semantics of an identifier $x$ \eqref{eq:c-var} retrieves the value stored in the
variable $x$ in the stack. The address-of operator \eqref{eq:c-addr}
returns the stack memory address of the variable $x$, which is $a_x$
by definition. The dereference operator \eqref{eq:c-stack-deref} and \eqref{eq:c-heap-deref}
retrieves the value stored at the memory address (both on
stack or heap) given by the evaluation of $\pexp$. If $\pexp$ does not
evaluate to a valid address, the result is $\bot$.
Field accesses, both static \eqref{eq:c-static-access} and dynamic \eqref{eq:c-dynamic-access},
retrieve the value of the field in the object, with dynamic access
first evaluating the expression $\sexp$ used to compute the field name.
If $\pexp$ does not evaluate to a valid object or if the field does not
exist, the result is $\bot$.
The arithmetic division computes the integer division of two
arithmetic expressions. If the division does not yield an integer
result or if the divisor is zero, the result is $\bot$.

\subsubsection{Assignments.}
\label{sec:concrete-assign}

Since most statements are standard, we omit their semantics from this paper, focusing
only on assignments.
\Cref{fig:concrete-assign-stack,fig:concrete-assign-heap}
summarize the semantics of assignment statements for two different cases:
the stack assignment $\lhs = \exp$, and
the heap assignment $\lhs = \new \hexp$,
with the latter dynamically allocating memory
(on the heap) for a variable using a fresh address. This address is
obtained from the function $\freshCheap:\Cheap \rightarrow \caddrh$, which
returns a new heap address $a_f \in \caddrh$ that is not already in use
in the heap.

Each figure outlines the possible left-hand side targets of
the assignment: to a variable \eqref{eq:c-sasg-var} and \eqref{eq:c-hasg-var}, to a pointer
dereference \eqref{eq:c-sasg-deref-stack}, \eqref{eq:c-sasg-deref-heap}, \eqref{eq:c-hasg-deref-stack} and \eqref{eq:c-hasg-deref-heap},
to a static field of an object \eqref{eq:c-sasg-static-field-stack} and \eqref{eq:c-hasg-static-field-stack},
and to a dynamic field of an object \eqref{eq:c-sasg-dynamic-field-stack} and \eqref{eq:c-hasg-dynamic-field-stack}.
In general, every assignment operation evaluates the
expression on the right-hand side, which results in a value $v \in \cvalbar$.
In the case of heap allocations, a fresh address $a_f \in \caddrh$
is computed, and the value produced by the evaluation of the right-hand side
is stored as image of that address in $\cheap$. In the following, we refer
to the result of the right-hand side evaluation in stack assignments and to the
freshly generated address in heap assignments as the \emph{assigned value}, to
unify the presentation of the semantics of the two types of assignments.
Variable assignments store the assigned value as the variable's image in the stack.
Dereference assignments store the assigned value at the variable (on the stack)
or address (on the heap) obtained by evaluating the left-hand side dereference.
Field assignments require the base identifier to hold an object directly on the stack:
the object is retrieved, updated to map the desired field to the assigned value, and stored back.
To assign a field of a heap-allocated object (or a stack-allocated object pointed to by a pointer),
an explicit load/store sequence is required: load the object via dereference,
update the field, and write it back.
Note that objects stored at heap addresses are always flat, i.e., their field
values lie in $\cval \cup \caddr$: this is enforced jointly by
rule~\eqref{eq:c-sasg-deref-heap}, which restricts the right-hand side of a
heap dereference assignment to $\cval \cup \caddr$, and by the grammar, which
limits the fields of object literals to $\fexps$ (primitive values and pointer
expressions).
\begin{figure}[t]
	\begin{equation}
		\label{eq:c-sasg-var}
		\frac{
			\csem{\exp} (\cstack, \cheap) = v
		}{
			\csem{x = \exp} (\cstack, \cheap) = (\cstack[x \mapsto v], \cheap)
		}
	\end{equation}

	\begin{equation}
		\label{eq:c-sasg-deref-stack}
		\frac{
			\csem{\pexp} (\cstack, \cheap) = a_y \in \caddrs
			\quad \csem{\exp} (\cstack, \cheap) = v
		}{
			\csem{\deref\pexp = \exp} (\cstack, \cheap) = (\cstack[y \mapsto v], \cheap)
		}
	\end{equation}

	\begin{equation}
		\label{eq:c-sasg-deref-heap}
		\frac{
			\csem{\pexp} (\cstack, \cheap) = a \in \caddrh
			\quad \csem{\exp} (\cstack, \cheap) = v \in \cval \cup \caddr
		}{
			\csem{\deref\pexp = \exp} (\cstack, \cheap) = (\cstack, \cheap[a \mapsto v])
		}
	\end{equation}

	\begin{equation}
		\label{eq:c-sasg-static-field-stack}
		\frac{
			\cstack(x) = o \in \cobj
			\quad \csem{\exp} (\cstack, \cheap) = v
		}{
			\csem{x.f = \exp} (\cstack, \cheap)
			= (\cstack[x \mapsto \objupdate{o}{f}{v}], \cheap)
		}
	\end{equation}

	\begin{equation}
		\label{eq:c-sasg-dynamic-field-stack}
		\frac{
			\cstack(x) = o \in \cobj
			\quad \csem{\exp} (\cstack, \cheap) = v
			\quad \csem{\sexp} (\cstack, \cheap) = f \in \cid
		}{
			\csem{x[\sexp] = \exp} (\cstack, \cheap)
			= (\cstack[x \mapsto \objupdate{o}{f}{v}], \cheap)
		}
	\end{equation}

	\caption{Semantics of expression assignments $\lhs = \exp$ in $\Cstate$.}
	\label{fig:concrete-assign-stack}
\end{figure}

\begin{figure}[t]
	\begin{equation}
		\label{eq:c-hasg-var}
		\frac{
			a_f = \freshCheap(\cheap)
			\quad \csem{\hexp} (\cstack, \cheap) = v
		}{
			\csem{x = \new \hexp} (\cstack, \cheap)
			= (\cstack[x \mapsto a_f], \cheap[a_f \mapsto v])
		}
	\end{equation}

	\begin{equation}
		\label{eq:c-hasg-deref-stack}
		\frac{
			\csem{\pexp} (\cstack, \cheap) = a_y \in \caddrs
			\quad a_f = \freshCheap(\cheap)
			\quad \csem{\hexp} (\cstack, \cheap) = v
		}{
			\csem{\deref\pexp= \new \hexp} (\cstack, \cheap)
			= (\cstack[y \mapsto a_f], \cheap[a_f \mapsto v])
		}
	\end{equation}

	\begin{equation}
		\label{eq:c-hasg-deref-heap}
		\frac{
			\csem{\pexp} (\cstack, \cheap) = a \in \caddrh
			\quad a_f = \freshCheap(\cheap)
			\quad \csem{\hexp} (\cstack, \cheap) = v
		}{
			\csem{\deref\pexp = \new \hexp} (\cstack, \cheap)
			= (\cstack, \cheap[a \mapsto a_f, a_f \mapsto v])
		}
	\end{equation}

	\begin{equation}
		\label{eq:c-hasg-static-field-stack}
		\frac{
			a_f = \freshCheap(\cheap)
			\quad \csem{\hexp} (\cstack, \cheap) = v
			\quad \cstack(x) = o \in \cobj
		}{
			\csem{x.f = \new \hexp} (\cstack, \cheap)
			= (\cstack[x \mapsto \objupdate{o}{f}{a_f}], \cheap[a_f \mapsto v])
		}
	\end{equation}

	\begin{equation}
		\label{eq:c-hasg-dynamic-field-stack}
		\frac{
			a_f = \freshCheap(\cheap)
			\quad \csem{\hexp} (\cstack, \cheap) = v
			\quad \cstack(x) = o \in \cobj
			\quad \csem{\sexp} (\cstack, \cheap) = f
		}{
			\csem{x[\sexp] = \new \hexp} (\cstack, \cheap)
			= (\cstack[x \mapsto \objupdate{o}{f}{a_f}], \cheap[a_f \mapsto v])
		}
	\end{equation}

	\caption{Semantics of \new assignments $\lhs=\new \hexp$ in $\Cstate$}
	\label{fig:concrete-assign-heap}
\end{figure}

\subsubsection{Assume Statement.}
\label{sec:assume}
The $\texttt{ASM}(\bexp)$ statement is an auxiliary construct, not part of the $\mull$ syntax,
used when analyzing conditionals and loops to restrict the program state to traces satisfying a guard.
Its concrete semantics is:
\begin{equation*} \text{(assume)} \quad
	\frac{
		\csem{\bexp} (\cstack, \cheap) = \true
	}{
		\csem{\texttt{ASM}(\bexp)} (\cstack, \cheap) = (\cstack, \cheap)
	}
\end{equation*}
States for which the guard evaluates to $\false$ are mapped to $\bot$.

\begin{example}[$\mull$ Running Example]
	We now encode the motivating C snippet from \Cref{lst:dancing-c} into $\mull$.
	The translation follows a few systematic choices.
	First, the C struct initialisation \texttt{list.L=0; list.R=0;} becomes simply
	\texttt{list = \{\}}: since $\mull$ is statically typed and \texttt{L} and \texttt{R}
	are address-typed fields, they cannot hold the integer \texttt{0} as a null
	placeholder without violating the typing invariant.
	Instead, a null-initialised pointer field is represented as absent from
	the object.
	Second, pointer declarations (\texttt{struct node *x = \&list}) map directly to
	\texttt{x = \&list}, using the address-of operator already in the language.
	Third, the \texttt{malloc} call combined with field initialisation
	(\texttt{n->L=tail; n->R=0}) is folded into a single heap-allocation statement
	\texttt{n = new \{ L:tail \}},\footnote{The $\mull$
	language does not support uninitialized data structures, so we perform
	explicit initialisation of all fields at allocation time. However, this is not a limitation of the
	framework, but rather a stylistic choice to simplify the presentation.} which atomically allocates a fresh heap address
	and stores the object; the \texttt{R} field is omitted for the same reason as above.
	Fourth, mutation through a pointer -- the C idiom \texttt{tail->R = n} --
	requires an explicit load/store sequence: \texttt{tmp = *tail} loads the object,
	\texttt{tmp.R = n} updates the field, and \texttt{*tail = tmp} writes it back.
	Finally, the two non-deterministic booleans \texttt{\_\_VERIFIER\_nondet\_bool()} are
	replaced by the predicates \texttt{nd\_loop()} and \texttt{nd\_pick()} that yield non-deterministic
	values.\footnote{While $\mull$ does not support predicates, we rely on an intuitive definition
	to faithfully replicate the running example. As these do not represent operations
	of interest for our framework, we omit their definition.}

	\begin{table}[htbp]
		\centering
		\caption{Concrete state annotations for \Cref{lst:dancing-mull}. $\cstack_l$ and $\cheap_l$ denote the stack and heap at line $l$. Inside the loop body, $a_f^{\,i}=\freshCheap(\cheap_6^{\,i})$ and $i\geq 1$.}
		\label{tab:dancing-states}
		\renewcommand{\arraystretch}{1.3}
		\begin{tabular}{|c|l|l|}
			\hline
			\textbf{Line} & \multicolumn{1}{c|}{$\cstack$}                                                                         & \multicolumn{1}{c|}{$\cheap$} \\
			\hline
			1             & $[\texttt{list}\mapsto \emptyobj]$                                                                     & \multirow{3}{*}{$\emptyset$}  \\
			\cline{1-2}
			2             & $\cstack_1[\texttt{x}\mapsto a_{\texttt{list}}]$                                                       &                               \\
			\cline{1-2}
			3             & $\cstack_2[\texttt{tail}\mapsto a_{\texttt{list}}]$                                                    &                               \\
			\hline
			6             & $\cstack_{17}^{\,i-1}$                                                                                 & $\cheap_{17}^{\,i-1}$         \\
			\hline
			7             & $\cstack_6^{\,i}[\texttt{n}\mapsto a_f^{\,i}]$
			              & \multirow{7}{*}{$\cheap_6^{\,i}[a_f^{\,i}\mapsto\{\texttt{L}\!:\!a_{\texttt{list}}\}]$}                                                \\
			\cline{1-2}
			8             & $\cstack_7^{\,i}[\texttt{tmp}\mapsto \cstack_7^{\,i}(\texttt{list})]$                                  &                               \\
			\cline{1-2}
			9             & $\cstack_8^{\,i}[\texttt{tmp}\mapsto\objupdate{\cstack_8^{\,i}(\texttt{tmp})}{\texttt{R}}{a_f^{\,i}}]$ &                               \\
			\cline{1-2}
			10            & $\cstack_9^{\,i}[\texttt{list}\mapsto \cstack_9^{\,i}(\texttt{tmp})]$                                  &                               \\
			\cline{1-2}
			13            & $\cstack_{10}^{\,i}[\texttt{x}\mapsto a_f^{\,i}]$                                                      &                               \\
			\cline{1-2}
			15            & $\cstack_{10}^{\,i}$                                                                                   &                               \\
			\cline{1-2}
			17            & $\cstack_{13}^{\,i} \cup \cstack_{15}^{\,i}$                                                           &                               \\
			\hline
		\end{tabular}
	\end{table}

	\begin{lstlisting}[
	language=c,
	numbers=left,
	tabsize=4,
	morekeywords={new},
	keywordstyle=\color{blue}\bfseries,
	caption={$\mull$ translation of \Cref{lst:dancing-c} with concrete state annotations.},
	label={lst:dancing-mull}]
list = {};
x = &list;
tail = &list;

/* nd_loop() and nd_pick() denote nondeterministic booleans */
while nd_loop() {
	n = new { L:tail };
	tmp = *tail;
	tmp.R = n;
	*tail = tmp;

	if nd_pick() {
		x = n;
	} else {
		skip;
	}
}
\end{lstlisting}

	The state annotations in \Cref{tab:dancing-states} track the concrete state at
	each program point.
	Before the loop (lines 1–3), the stack maps \texttt{list} to an empty object
	and both \texttt{x} and \texttt{tail} to $a_\texttt{list}$, with the heap empty.
	At each iteration $i$, a fresh heap address $a_f^i$ is allocated for the new node
	(line~7), \texttt{list.R} is updated to $a_f^i$ via the load/store sequence on
	lines~8--10, and \texttt{x} either advances to $a_f^i$ (line~13) or retains its
	previous value (line~15), so after the loop \texttt{x} may hold $a_\texttt{list}$
	or any of $a_f^1, \ldots, a_f^k$.
\end{example}

\subsection{Collecting Semantics}
The concrete semantics is lifted pointwise to sets of states $\Cset \subseteq \Cstate$.
For a statement $\stmt$,
$\ccsem{\stmt} : \powerset{\Cstate} \to \powerset{\Cstate}$
is defined as $\ccsem{\stmt}(\Cset) = \{ \csem{\stmt}(\cstate) \mid \cstate \in \Cset \}$;
for the auxiliary construct $\texttt{ASM}(\bexp)$,
$\ccsem{\texttt{ASM}(\bexp)} : \powerset{\Cstate} \to \powerset{\Cstate}$
is defined as $\ccsem{\texttt{ASM}(\bexp)}(\Cset) = \{ \cstate' \mid \cstate \in \Cset \land \csem{\texttt{ASM}(\bexp)}(\cstate) = \cstate' \land \cstate' \neq \bot \}$;
and for an expression $\exp$,
$\ccsem{\exp} : \powerset{\Cstate} \to \powerset{\cval}$
is defined as $\ccsem{\exp}(\Cset) = \{ \csem{\exp}(\cstate) \mid \cstate \in \Cset \}$.

\section{Split State}
\label{sec:split-state}

A state $\sstate \in \Sstate$ is a tuple $(\sval, \smem) \in \Sval \times \Smem$
where locations (\cid\ or \caddrh) are partitioned by the kind of
information they contain:
\begin{inlinelist}
	\item if the location contains a concrete value, it is stored in $\Sval$,
	\item otherwise, if it contains an address, it is stored in $\Smem$.
\end{inlinelist}
Moreover, if a location contains an object, the latter is \emph{exploded}
into entries identified by the location and the field name, each stored
in $\Sval$ or $\Smem$ using the same reasoning (each object is flattened
to per-field mappings via our $\explodefun$ function, see
\Cref{sec:explode-implode-obj}).
Here, $\caddrr = \{r_n \mid n \in \mathbb{N}\}$ is a set of addresses disjoint from
all concrete addresses and identifiers ($\caddrr \cap \caddr = \emptyset$,
$\caddrr \cap \cid = \emptyset$), generated during object decomposition to
represent the locations of subobjects within stack-allocated objects.
We define $\caddrp \defn \caddr \cup \caddrr$ as the extended address set of the split state.
The $\Sstate$ subcomponents are formalized as:
\[
	\begin{matrix*}
		\Sval = \Sval_{ih} \times \Sval_{f} \\
		\Sval_{ih} : (\cid \cup \caddrh) \rightarrow \cval \\
		\Sval_{f} : (\cid \cup \caddrh \cup \caddrr) \rightarrow (\cid \rightarrow \cval) \\
	\end{matrix*}
\]
\[
	\begin{matrix*}
		\Smem = \Smem_{ih} \times \Smem_{f} \\
		\Smem_{ih} : (\cid \cup \caddrh) \rightarrow \caddr \\
		\Smem_{f} : (\cid \cup \caddrh \cup \caddrr) \rightarrow (\cid \rightarrow \caddrp) \\
	\end{matrix*}
\]
where the subscript $ih$ represents storage for variables and heap locations,
and $f$ for object fields.
A split state can be equivalently referred to as $\sstate$, as the pair
$(\sval, \smem)$, or as the expanded form
$((\sval_{ih}, \sval_{f}), (\smem_{ih}, \smem_{f}))$.
The lattice structure is given by $\langle \powerset{\Sstate}, \subseteq \rangle$.

Throughout, all states are assumed to satisfy the following well-formedness condition.
\begin{definition}[Well-Formed Split State]
	Consider an arbitrary split state $\sstate = ((\sval_{ih}, \sval_{f}),$ $(\smem_{ih}, \smem_{f}))$.
	We say that $\sstate$ is \emph{well-formed} iff:
	\begin{itemize}
		\item
		      the pairs of functions $(\sval_{ih},\smem_{ih}), (\sval_{ih},\sval_{f}),
			      (\sval_{ih}, \smem_{f}), (\smem_{ih}, \sval_{f}), (\smem_{ih},\smem_{f})$
		      have disjoint domains, and
		\item
		      $\forall l \in \dom{\sval_{f}} \cup \dom{\smem_{f}},$ the pair of functions
		      $(\sval_{f}(l),\smem_{f}(l))$ have disjoint domains.
	\end{itemize}
\end{definition}
Note that this is a reasonable
assumption, since in a concrete state, a location can only be mapped to a value,
an address, or an object. Thus, when partitioning this state, each location can
only end up in one of the split components.
Since $\drestrcodom{}$ retains all keys (as used in $\alphasplitdot$,
\Cref{sec:concretization-abstraction}), every object-holding location $l$
appears in both $\dom{\sval_f}$ and $\dom{\smem_f}$, while each
\emph{object–field pair} $(l,f)$ occurs in exactly one of the two inner maps.
Concretely:
\begin{itemize}
	\item \emph{Empty object} $\emptyobj$:
	      $\sval_f(l) = \emptyset$ and $\smem_f(l) = \emptyset$.
	\item \emph{Only value-typed fields} (e.g., $\{f\!:\!v\}$, $v \in \cval$):
	      $\sval_f(l) \neq \emptyset$ and $\smem_f(l) = \emptyset$.
	\item \emph{Only reference-typed fields} (e.g., $\{f\!:\!a\}$, $a \in \caddrp$):
	      $\sval_f(l) = \emptyset$ and $\smem_f(l) \neq \emptyset$.
	\item \emph{Mixed fields}: both $\sval_f(l) \neq \emptyset$ and $\smem_f(l) \neq \emptyset$.
\end{itemize}

\begin{example}[Split State Benefit]
	\label{sec:split-state-example}
	Consider the program \texttt{p = new 5; tmp = *p; y = tmp + 1}.
	After all three statements, the concrete state is the pair $(\cstack, \cheap)$ with
	$\cstack = \{\texttt{p} \mapsto a,\; \texttt{tmp} \mapsto 5,\; \texttt{y} \mapsto 6\}$
	and $\cheap = \{a \mapsto 5\}$;
	within $\cstack$, an address and two integer values coexist in the same map.
	The split state separates these by type into two components:
	$\smem_{ih} = \{\texttt{p} \mapsto a\}$ captures the memory structure,
	while $\sval_{ih} = \{a \mapsto 5,\; \texttt{tmp} \mapsto 5,\; \texttt{y} \mapsto 6\}$
	captures the numerical values; the fields components $\smem_f$ and $\sval_f$ are empty here.
	This separation is not merely convenient but necessary from a formalization
	standpoint: each abstraction step can then focus exclusively on the relevant
	component of the state, i.e., memory structure~$\smem$ or numerical
	values~$\sval$, without needing to inspect the other.
\end{example}

\subsection{Concrete-Split Equivalence}
\label{sec:concrete-split-equivalence}
To ensure that the split state and semantics are just a reformulation of the
concrete state and semantics, it is necessary to show that there is no loss
of precision in the abstraction/concretization process, nor in the split semantics.
This reduces to showing:
\begin{inlinelist}
	\item there exist a pair of monotone functions $\alphasplit$ and $\gammasplit$
	such that
	$\langle \powerset{\Cstate}, \subseteq \rangle
		\GaloiS{\gammasplit}{\alphasplit}
		\langle \powerset{\Sstate}, \subseteq \rangle$, and
	\item for all statements $\stmt$ and split states $(\sval, \smem)$,
	$\csem{\stmt} (\gammasplit(\sval, \smem)) = \gammasplit(\ssem{\stmt}(\sval,\smem))$.
\end{inlinelist}

\subsubsection{Exploding and Imploding Objects.}
\label{sec:explode-implode-obj}
In order to define the relationship between concrete and split states, we introduce
two mutually inverse helper functions
$
	\explodefun: ((\cid \cup \caddrh) \to \cobj)
	\to ((\cid \cup \caddrh \cup \caddrr) \to \tilde{\cobj})
$ and
$
	\implodefun: ((\cid \cup \caddrh \cup \caddrr) \to \tilde{\cobj})
	\to ((\cid \cup \caddrh) \to \cobj)
$,
where
$\tilde{\cobj} \defn \cid \rightarrow (\cval \cup \caddrp \cup \bot)$
denotes an \emph{exploded object} -- a flat field map with no nested objects.
Intuitively, $\explodefun$ flattens nested objects replacing each sub-object with a fresh
address $r \in \caddrr$, while $\implodefun$ reassembles them back.
More precisely, $\explodefun$ introduces $\caddrr$ addresses only for \emph{stack} objects
(indexed by $\cid$), and acts as the identity on \emph{heap} objects (indexed by $\caddrh$):
the latter are already flat, since the $\new$ construct allocates each nested object
at a fresh heap address in $\caddrh$ (see \Cref{sec:concrete-assign}).
Formal definitions and proofs are in \Cref{sec:app-explode-implode}.

\subsubsection{Concretization and Abstraction.}
\label{sec:concretization-abstraction}
In the following definitions we use disjoint-domain union $f\cupdot g$,
second-order disjoint-domain union $f\cupddot g$, restriction
$f\restrdom{X}$, $f\restrcodom{Y}$, $f\restr{X}{Y}$, and double
restriction $f\drestr{X}{Y}$, as introduced in \Cref{sec:background}.

\begin{definition}[Concretization of Split States]
	The function $\gammasplit: \powerset{\Sstate} \to \powerset{\Cstate}$
	maps a set of split states $\Sset$ to a set of concrete states. It is defined as:
	\[
		\gammasplit(\Sset) = \bigcup_{\sstate \in \Sset} \{ \gammasplitdot(\sstate) \},
	\]
	where \(\gammasplitdot: \Sstate \to \Cstate\) is the concretization function for a single split state.
	Given a split state \(\sstate = (\sval, \smem)\), the corresponding concrete state is:
	\[
		\gammasplitdot(\sstate) =
		\left(
		\sval_{i} \cupdot \smem_{i} \cupdot \implodefun(\sval_{sf} \cupddot \smem_{sf}),
		\sval_{h} \cupdot \smem_{h} \cupdot \implodefun(\sval_{hf} \cupddot \smem_{hf})
		\right).
	\]
	where
	\[
		\sval_{i} = \sval_{ih}\restrdom{\cid}, \quad
		\sval_{h} = \sval_{ih}\restrdom{\caddrh}, \quad
		\sval_{sf} = \sval_{f}\restrdom{(\cid \cup \caddrr)}, \quad
		\sval_{hf} = \sval_{f}\restrdom{\caddrh},
	\]
	\[
		\smem_{i} = \smem_{ih}\restrdom{\cid}, \quad
		\smem_{h} = \smem_{ih}\restrdom{\caddrh}, \quad
		\smem_{sf} = \smem_{f}\restrdom{(\cid \cup \caddrr)}, \quad
		\smem_{hf} = \smem_{f}\restrdom{\caddrh}.
	\]
\end{definition}

\begin{definition}[Abstraction of Concrete States]
	The function $\alphasplit: \powerset{\Cstate} \to \powerset{\Sstate}$
	maps a set of concrete states $\Cset$ to a set of split states. It is defined as:
	\[
		\alphasplit(\Cset) = \bigcup_{\cstate \in \Cset} \{ \alphasplitdot(\cstate) \},
	\]
	where \(\alphasplitdot: \Cstate \to \Sstate\) is the
	abstraction function for a single concrete state.
	Given a concrete state \(\cstate = (\cstack, \cheap)\), the corresponding
	split state is:
	\[
		\alphasplitdot(\cstate) = ((\sval_{ih}, \sval_{f}), (\smem_{ih}, \smem_{f})),
	\]
	where, letting $\tilde{O} \defn \explodefun(\cstack\restrcodom{\cobj} \cupdot \cheap\restrcodom{\cobj})$:
	\[
		\sval_{ih} = \cstack\restrcodom{\cval} \cupdot \cheap\restrcodom{\cval},
		\quad
		\sval_f = \tilde{O}\drestrcodom{\cval},
	\]
	\[
		\smem_{ih} = \cstack\restrcodom{\caddr} \cupdot \cheap\restrcodom{\caddr},
		\quad
		\smem_f = \tilde{O}\drestrcodom{\caddrp}.
	\]
\end{definition}

\begin{theorem}[Concrete-Split Galois Isomorphism]
	\label{thm:concrete-split-galois-isomorphism}
	The functions $\alphasplit$ and $\gammasplit$ form a
	\emph{Galois isomorphism} between the concrete state lattice
	and the split state lattice:
	$\langle \powerset{\Cstate}, \subseteq \rangle \GaloiS{\alphasplit}{\gammasplit}
		\langle \powerset{\Sstate}, \subseteq \rangle$.
	The proof is provided in \Cref{sec:proof-concrete-split-galois-isomorphism}.
\end{theorem}

\subsubsection{Split Semantics.}

Rather than defining new semantics directly on the split state, we
follow~\cite{whole-value} and leverage
the fact that $\alphasplit$ and $\gammasplit$ form a Galois isomorphism
between concrete and split states (\Cref{thm:concrete-split-galois-isomorphism}).
This property ensures that applying $\gammasplit$ to a split state, computing
semantics in the concrete domain, and then lifting the result back via
$\alphasplit$ yields the most precise abstraction of the concrete
outcome~\cite{CousotCousot79-1}.
Note that this does not make our semantics computable, but allows direct reuse
of the concrete semantics as a specification for the split semantics without creating a new one.
Successive abstractions can thus compare their abstract semantics to the concrete one by applying
$\alphasplit$ and $\gammasplit$.
Formally, the split semantics of a statement $\stmt$, auxiliary construct
$\texttt{ASM}(\bexp)$, and expression $\exp$ are respectively:
$
	\ssem{\stmt} = \alphasplit \circ \csem{\stmt} \circ \gammasplit,
$
$
	\ssem{\texttt{ASM}(\bexp)} = \alphasplit \circ \csem{\texttt{ASM}(\bexp)} \circ \gammasplit,
$
and
$
	\ssem{\exp} = \csem{\exp} \circ \gammasplit\,.
$
Note that for expressions, $\alphasplit$ is unnecessary since
$\csem{\exp}$ already produces a concrete value.
Since $\gammasplit$ and $\alphasplit$ are inverses of each other
(as shown in \Cref{sec:proof-concrete-split-galois-isomorphism}),
the following equivalence holds trivially.
\begin{theorem}[Split and Concrete Semantics Equivalence]
	\label{thm:split-concrete-semantics-equivalence}
	For any split state $\sstate \in \Sstate$, the following holds:
	\begin{inlinelist}
		\item for any $\stmt \in \stmts$:
		$\csem{\stmt} \gammasplit(\sstate) = \gammasplit(\ssem{\stmt}(\sstate))$;
		\item for any $\bexp \in \bexps$:
		$\csem{\texttt{ASM}(\bexp)} \gammasplit(\sstate) = \gammasplit(\ssem{\texttt{ASM}(\bexp)}(\sstate))$;
		\item for any $\exp \in \exps$:
		$\csem{\exp}(\gammasplit(\sstate)) = \ssem{\exp}(\sstate)$.
	\end{inlinelist} (Proof in \Cref{sec:proof-split-concrete-semantics-equivalence}).
\end{theorem}

Thus, the split state is a rephrasing of the concrete state, with no loss of
precision in both the commutation between the representations and the semantics.
Abstractions over the concrete semantics are challenging: the program state is
composed by two functions that mix value information with points-to
information. Instead, the split state provides a setting in which value and
points-to information are well separated: this lays the ground for modular
abstractions of the split state, in which each component can focus on one kind
of information only, either points-to or values. Since the split state is a
simple rephrasing, being sound w.r.t. the split semantics entails being sound
w.r.t. the concrete semantics.

\subsection{Collecting Semantics}
The split semantics is lifted pointwise to sets of split states $\Sset \subseteq \Sstate$.
For a statement $\stmt$,
$\cssem{\stmt}: \powerset{\Sstate} \to \powerset{\Sstate}$
is defined as $\cssem{\stmt}(\Sset) = \{\ssem{\stmt}(\sstate) \mid \sstate \in \Sset\}$;
for the auxiliary construct $\texttt{ASM}(\bexp)$,
$\cssem{\texttt{ASM}(\bexp)}: \powerset{\Sstate} \to \powerset{\Sstate}$
is defined as $\cssem{\texttt{ASM}(\bexp)}(\Sset) = \{\sstate' \mid \sstate \in \Sset \land \ssem{\texttt{ASM}(\bexp)}(\sstate) = \sstate' \land \sstate' \neq \bot\}$;
and for an expression $\exp$,
$\cssem{\exp}: \powerset{\Sstate} \to \powerset{\cval}$
is defined as $\cssem{\exp}(\Sset) = \{\ssem{\exp}(\sstate) \mid \sstate \in \Sset\}$.
By pointwise application of \Cref{thm:split-concrete-semantics-equivalence} over all
$\sstate \in \Sset$, the split collecting semantics commutes with concretization:
$\gammasplit(\cssem{\stmt}(\Sset)) = \ccsem{\stmt}(\gammasplit(\Sset))$.

\section{Abstract State}
\label{sec:abstract-state}

In this section, we define a general framework in which two abstract domains,
one abstracting values and one abstracting memory, can cooperate modularly to
form a sound over-approximation of the split state and semantics.
Specifically, we consider
a \emph{value domain} $\Aval$, which abstracts the value component $\sval$ of
the split state $(\sval, \smem) \in \Sval \times \Smem$, and
a \emph{memory domain} $\Amem$, which abstracts the memory component $\smem$.
Intuitively, the value domain approximates the values stored in program variables and memory
locations, while the memory domain approximates the structure of memory,
including allocation, indirection, aliasing, and object fields.
Although $\Amem$ is often referred to in the literature as a
\emph{heap domain/state}, in our framework it abstracts all program
memory locations, both stack- and heap-allocated. For this reason, we
consistently use the terms \emph{memory domain/state} to refer to $\Amem$.

\paragraph{Memory Identifiers and Domain Interaction.}
The cooperation between the two domains is mediated through a set of
\emph{memory identifiers}, denoted by $\MemID$, along with a set of
operations defined by the domains to manipulate them.
These identifiers represent abstract memory locations corresponding to
one or more split locations in the memory of the analyzed program,
including both stack-allocated and heap-allocated entities.
They are treated as abstract names without internal structure.
Their concrete interpretation and allocation strategy are determined by
the chosen memory domain implementation.

During abstract execution, memory operations may reorganize memory identifiers,
for instance by materializing new locations or merging existing ones.
Such structural changes are communicated to the value domain through
\emph{substitutions}, ensuring that the combined abstract state remains
consistent and sound.

\paragraph{Domain Requirements.}
Since the framework is parametric in the choice of $\Aval$ and $\Amem$,
we require each abstract domain $\mathcal{X} \in \{\Aval, \Amem\}$
to satisfy the following minimal properties,
using the abstract interpretation terminology
recalled in \Cref{sec:background}:
\begin{itemize}
	\item it should form a partial order, denoted $\langle \mathcal{X}, \sqsubseteq_\mathcal{X} \rangle$,
	      with bottom element $\bot_\mathcal{X}$, and should define an upper bound operator $\sqcup_\mathcal{X}$;
	\item it should provide a widening operator $\triangledown_\mathcal{X}$
	      (if not needed by the domain itself, it can simply delegate to
	      $\sqcup_\mathcal{X}$).
\end{itemize}

In addition, we require the domains to provide monotone concretization
functions $\gammaAV$ and $\gammaAM$ whose purpose is not to directly abstract
parts of the split state, but to provide partial views of the program memory
that can be combined to produce an overall concretization of the abstract
state. Furthermore, instead of defining standard abstract transformers, we
require each domain to expose a small operational interface that the framework
can use to evolve the two domains consistently during the analysis. These
requirements are detailed in
\Cref{sec:concretization-function,sec:abstract-semantics}, respectively.

\subsection{Abstract State Construction}
The overall abstract state is obtained by combining the value and memory
domains into a single state. It is defined as $\Astate \defn \Aval \otimes \Amem$,
where $\otimes$ denotes the \emph{smash product} of abstract
domains~\cite{ArceriM17}, a form of reduced product~\cite{CortesiCF13}
that propagates the bottom element of one component to all others.
An abstract state is therefore a pair $\astate = (\aval, \amem) \in \Astate$,
where $\aval \in \Aval$ and $\amem \in \Amem$.
Since the smash product is a special form of Cartesian product, the ordering and lattice
operators on $\Astate$ are defined component-wise:
\begin{itemize}
	\item $(\aval_1, \amem_1) \sqsubseteq_{\Astate} (\aval_2, \amem_2)
		      \iff \aval_1 \sqsubseteq_{\Aval} \aval_2 \land \amem_1 \sqsubseteq_{\Amem} \amem_2$;
	\item $(\aval_1, \amem_1) \sqcup_{\Astate} (\aval_2, \amem_2)
		      = (\aval_1 \sqcup_{\Aval} \aval_2, \amem_1 \sqcup_{\Amem} \amem_2)$;
	\item $(\aval_1, \amem_1) \triangledown_{\Astate} (\aval_2, \amem_2)
		      = (\aval_1 \triangledown_{\Aval} \aval_2, \amem_1 \triangledown_{\Amem} \amem_2)$.
\end{itemize}
For the same reason, $\Astate$ inherits the lattice properties of its components:
$\sqsubseteq_{\Astate}$ is a partial order, $\sqcup_{\Astate}$ is an
upper bound operator, $\bot_{\Astate}$ is the bottom element, and
$\triangledown_{\Astate}$ is a widening operator.

\subsection{Concretization Function}
\label{sec:concretization-function}
The concretization function
$\gammaA:\Astate \to \powerset{\Sstate}$ maps an abstract state
to the set of all possible split states it represents.

$\gammaAV$ is a monotone function that produces split value environments in which
split addresses are abstracted by memory identifiers.
Concretely, we write:
$$
	\gammaAV : \Aval
	\to
	\powerset{\Avalcid \times \Avalmemid},
$$
where $\Avalmemid : \MemID \to \cval$ is a mapping that associates
split values to memory identifiers, while $\Avalcid$ denotes the
$\cid$-indexed part of the value split state, i.e. the restriction
of both $\Sval_{ih}$ and $\Sval_f$ to keys in $\cid$
(so $\dom{\Avalcid} \subseteq \cid$).
To relate abstract memory identifiers to split memory locations,
the concretization of the memory domain provides, together with
split memory structure, a mapping that connects abstract
identifiers to split addresses. Formally, we define:
$$
	\gammaAM : \Amem
	\to
	\powerset{\Smem \times \gammaID \times \gammaIDF},
$$
as a monotone function.
Intuitively, for each abstract memory state $\amem$, every element
$(\smem, \gammaID, \gammaIDF) \in \gammaAM(\amem)$ describes:
\begin{inlinelist}
	\item a possible split memory structure $\smem$, and
	\item how abstract memory identifiers are implemented by
	split memory locations in $\smem$.
\end{inlinelist}
Here, $\gammaID : \MemID \to \powerset{\caddrh}$ maps abstract
memory identifiers to sets of heap addresses, while
$\gammaIDF : \MemID \to \powerset{(\caddrh\cup\caddrr) \times \cid}$
maps abstract memory identifiers to sets of split
field locations (only heap- and stack-nested objects carry fields). $\gammaID$ and $\gammaIDF$ have disjoint domains, namely
$\dom{\gammaID} \cap \dom{\gammaIDF} = \emptyset$: by the static typing assumption
(\Cref{sec:mull-language}), the abstract memory identifiers for
base locations (domain of $\gammaID$) and for field locations (domain of $\gammaIDF$)
are always distinct.

To concretize the $\MemID$-indexed value component $\Avalmemid$,
we need to replace each memory identifier in the state with the split locations
it may denote. We thus define two helper concretizations that replace
memory identifiers in $\nu \in \Avalmemid$ by split bindings induced
by $\gammaID$ and $\gammaIDF$, defined as:
\begin{align*}
	 & \rhobase: (\Avalmemid \times \gammaID) \to (\caddrh \to \cval) \\
	 & \rhobase(\nu,\gammaID) =
	\{ (a,v) : (\mId,v)\in\nu \wedge a \in \gammaID(\mId) \}
\end{align*}
\begin{align*}
	 & \rhofield: (\Avalmemid \times \gammaIDF) \to ((\caddrh \cup \caddrr) \to \cid \to \cval) \\
	 & \rhofield(\nu,\gammaIDF) =
	\{ (a,\{(f_i,v_i)\}) : (\mId,v_i)\in\nu \wedge (a,f_i)\in\gammaIDF(\mId) \}
\end{align*}
Both functions are trivially monotone, since their definitions only add bindings
to the output as the inputs grow.

Finally, the concretization of the full abstract state $(\aval,\amem)$ is
defined as
\begin{align*}
	\gammaA(\aval,\amem) = & \{ (\sval,\smem): \;
	(\sval_{id},\sval_{mid}) \in \gammaAV(\aval) \;\wedge\;
	(\smem,\gammaID,\gammaIDF) \in \gammaAM(\amem)                                                              \\
	                       & \quad \wedge \sval_{mid}^{id} = \rhobase(\sval_{mid},\gammaID) \;\wedge\;
	\sval_{mid}^{idf} = \rhofield(\sval_{mid},\gammaIDF)                                                        \\
	                       & \quad \wedge \sval = \sval_{id} \cupdot \sval_{mid}^{id} \cupdot \sval_{mid}^{idf}
	\}
\end{align*}
Here, $(\sval_{id}, \sval_{mid}) \in \gammaAV(\aval)$ is a pair produced by the
value concretization, where $\sval_{id} \in \Avalcid$ provides the
$\cid$-keyed bindings directly, and $\sval_{mid} \in \Avalmemid$ provides the
$\MemID$-keyed bindings that still need to be resolved to split addresses.
The helpers $\rhobase$ and $\rhofield$ perform that resolution:
$\sval_{mid}^{id}$ maps
memory identifiers to base-address bindings, while
$\sval_{mid}^{idf}$ maps them to
location-field bindings.
The final split value environment $\sval$ is the disjoint union of all
three components.

\begin{lemma}[Monotonicity of $\gammaA$]
	\label{lemma:abs-gamma-monotone}
	The concretization function $\gammaA$ is monotone, that is:
	\[
		\forall \astate_1, \astate_2 \in \Astate : \astate_1 \sqsubseteq_{\Astate} \astate_2
		\implies \gammaA(\astate_1) \subseteq \gammaA(\astate_2)
	\]
\end{lemma}
Intuitively, since $\Astate$ is a product domain, monotonicity follows
if $\gammaAV$ and $\gammaAM$ are monotone. Proof in \Cref{sec:proof-abs-gamma-monotone}.

\subsection{Abstract Semantics}
\label{sec:abstract-semantics}

To formalize the semantics of our framework, we require that the value and
memory domains provide specialized abstract transformers that the framework
exploits to provide a sound over-approximation of the split semantics. These
are either slight modifications of the standard transformers required by
abstract interpretation, or new operators that extract information from the domains.
To simplify the presentation of such transformers we introduce the following
auxiliary sets:
\begin{inlinelist}
	\item
	$\tilde{V} \defn \MemID \cup \vexps$, the set of value expressions that can also contain memory identifiers;
	\item
	$\rhss \defn \exp \mid \oexp \mid \new(\aexp \mid \bexp \mid \sexp \mid \oexp)$,
	the set of expressions that may occur on the right-hand side of assignments
	and that may contain memory operations; we let $\rhs \in \rhss$;
	\item
	$\tilde{S}$, a lattice structure that models string values that is produced by the value domain (used to resolve dynamic accesses).
\end{inlinelist}
Intuitively, $\tilde{S}$ is an abstraction of the set of possible string values
that over-approximate the string argument $\sexp$ in $\exp[\sexp]$.

Operations marked with ${}^{(\star)}$ are specific to this framework;
the remaining ones are standard abstract interpretation transformers,
slightly modified to accommodate information exchange between the domains.
The value domain must provide the following operations:
\begin{itemize}
	\item
	      $\assignfunAV : (\MemID \cup \cid) \times \tilde{V} \times \Aval \times \{\text{\ttfamily true} \mid \text{\ttfamily false}\} \to \Aval$,
	      which assigns the abstract value denoted by the second argument to the (memory)
	      identifier given as first argument in the abstract value state;
	      the boolean flag selects between a strong update (overwrite)
	      and a weak update (join with the existing value);
	\item
	      $\assumefunAV : \bexps \times \Aval \to \Aval$, which refines the abstract value
	      state according to a boolean guard.
	\item
	      ${}^{(\star)}\evalfun : \sexps \times \Aval \to \tilde{S}$, which evaluates string
	      expressions and returns a abstract set of values used to resolve
	      dynamic memory accesses;
\end{itemize}
Note that, while the existence of an evaluation function such as $\evalfun$ is typical of
non-relational domains, relational domains employed in whole-value analyses
must also produce an abstraction of an expression’s value to interact with
other domains. Therefore, the framework is not restricted to non-relational
value domains. For instance, a relational domain like the polyhedra domain can still produce
a non relational interval reporting the range of values an expression can evaluate to.
It follows that the value domain must be able to over-approximate string values
to properly support the $\exp[\sexp]$ instruction. If it cannot, the set $\tilde{S}$
can be instantiated with a trivial lattice (e.g., $\{\top\}$) that over-approximates all strings,
and the $\evalfun$ operator can be defined as a no-op that always returns $\top$.

Instead, the memory domain must provide the following operations:
\begin{itemize}
	\item
	      $\assignfunAM : \lhss \times \rhss \times \Amem \to \Sub \times \Amem$,
	      which updates the abstract memory structure by assigning the second
	      argument to the first, and returns any substitutions required to keep the
	      value domain consistent; the first argument may be either a syntactic
	      left-hand side (an element of $\lhss$) or an explicit memory identifier
	      (an element of $\MemID$);

	\item
	      $\assumefunAM : \bexps \times \Amem \to \Amem$, which refines the abstract memory
	      state according to a boolean guard;

	\item
	      ${}^{(\star)}\pointstofun : \wp(\cid \cup \MemID) \times \Amem \to \powerset{\cid \cup \MemID}$,
	      which yields the set of memory identifiers that a set of identifiers and
	      memory identifers may point to in the current memory state;

	\item
	      ${}^{(\star)}\pointedbyfun : \cid \times \Amem \to \MemID$,
	      which yields the memory identifier that abstracts the address of the given program variable
	      in the current memory state;

	\item
	      ${}^{(\star)}\accessfun : \powerset{\MemID} \times \Sigma^* \times \Amem
		      \to \powerset{\MemID} \times \Sub \times \Amem$,
	      where $\Sigma^*$ is the set of string literals,
	      which resolves field accesses starting from a set of base memory identifiers
	      and a statically known field name to a new set of memory identifiers;
	      since an access on the left-hand side of an assignment might define new fields,
	      $\accessfun$ has agency to update the memory structure and return a
	      new instance of $\Amem$, together with any substitutions required to
	      keep the value domain consistent;

	\item
	      ${}^{(\star)}\accessfun : \powerset{\MemID} \times \tilde{S} \times \Amem
		      \to \powerset{\MemID} \times \Sub \times \Amem$
	      that models a field access in the same way as its previous
	      overload, but starting from a dynamically evaluated field name that
	      has already been abstracted to an instance of $\tilde{S}$ by the
	      value domain;

	\item
	      ${}^{(\star)}\freshSsplit : \Amem \to \MemID \times \Amem$, which
	      allocates fresh memory identifiers corresponding to a
	      stack-allocated abstract location; since the memory abstraction might
	      need update to track the new location, $\freshSsplit$ returns an
	      updated instance of $\Amem$ together with the fresh identifier, but
	      no substitutions since the location is fresh and cannot be referenced
	      by the value domain yet;\footnote{Note that $\freshSsplit$ does not need to return an
	      unused memory identifier at every call, but rather an address for a
	      particular allocation: for instance, an allocation-site based
	      analysis will generate one identifier for each location in the
	      program, but will always generate the same identifier for the
	      same location at each loop iteration.}

	\item
	      ${}^{(\star)}\freshHsplit : \Amem \to \MemID \times \Amem$, that
	      performs the same role as $\freshSsplit$ but for heap-allocated
	      abstract locations;

	\item
	      ${}^{(\star)}\memidfunAM : \Amem \to \powerset{\MemID}$, which returns the set of
	      memory identifiers currently tracked in the given abstract memory state;
	      used to state the coverage condition on $\gammaID$ and $\gammaIDF$
	      (see~C1 below);

	\item
	      ${}^{(\star)}\isSumfun : \MemID \times \Amem \to \{\text{\ttfamily true} \mid \text{\ttfamily false}\}$,
	      which returns \texttt{true} if the memory identifier may represent more than one
	      concrete location (a summary node), and \texttt{false} otherwise (represent
	      exactly one, i.e., a singleton node).
\end{itemize}

\subsection{Substitutions}
\label{sec:substitutions}

The memory domain may reorganize abstract memory identifiers during abstract
computations, for instance by materializing (or splitting), merging and renaming
locations. These identifier changes must be communicated to the value domain
so the two components remain consistent.

We adopt the substitution concept
from~\cite{Ferrara16Framework}, and specialize it to our setting.
A substitution $\sub = [\MId_1 \to \MId'_1, \dots, \MId_n \to \MId'_n] \in \Sub$
is a finite sequence of replacements, where each $\MId_j, \MId'_j \subseteq \MemID$.
A binding $\MId \to \MId'$ states that the post-state identifiers in $\MId$
derive from the pre-state identifiers in $\MId'$; identifiers in $\MId'$
not appearing in any $\MId$ are killed.
The three canonical cases are:
\begin{inlinelist}
	\item \emph{rename}: $\MId$ and $\MId'$ are both singletons;
	\item \emph{merge}: $\MId$ is a singleton, $|\MId'| > 1$
	(e.g., $\{\mId_1\} \to \{\mId_2,\mId_3\}$ merges $\mId_2$ and $\mId_3$);
	\item \emph{split/materialization}: $|\MId| > 1$, $\MId'$ is a singleton
	(e.g., $\{\mId_1,\mId_2\} \to \{\mId_3\}$ splits $\mId_3$).
\end{inlinelist}

When a memory domain transformer returns a substitution $\sub$, the value domain
applies it via $\applysubfun(\sub, \aval, \amem)$, defined as:
\begin{align*}
	 & \applysubfun(\sub,\, \aval,\, \amem) = \aval_n \text{ where } \sub = [\MId_1 \to \MId'_1, \dots, \MId_n \to \MId'_n],\ \aval_0 = \aval \\
	 & \quad \wedge \forall\, j \in [1,n] : \aval_j =
	\sqcup_{\mId' \in \MId'_j} \{ {\aval}'_m \;:\; \MId_j = \{\mId_1, \dots, \mId_m\}, {\aval}'_0 = \aval_{j-1},                     \\
	 & \quad \forall\, k \in [1,m] : {\aval}'_k = \assignfunAV(\mId_k, \mId', {\aval}'_{k-1}, \lnot\isSumfun(\mId_k, \amem)) \}
\end{align*}

\subsection{Expression Rewriting}

During abstract computations, the right-hand side of an assignment may contain
memory-related sub-expressions (e.g., pointer dereferences, field accesses, object
allocations) that must be rewritten to memory identifiers since the value domain
cannot handle them.
The function $\Rfun : \rhss \times \Aval \times \Amem \to \wp(\tilde{V}) \times \Aval \times \Amem$
handles this rewriting: it traverses a right-hand-side
expression, resolves every memory-sensitive sub-expression to a set of abstract
value references, and threads the updated abstract state through the process.
To avoid cluttering the definition, we will pass sets in $\wp(\tilde{V})$ to
functions that expect as arguments sets of $\wp(\MemID)$, implictly assuming
that any element of the argument that is not a memory identifier is ignored by
the function.
Function $\Rfun$ is defined by structural induction on the syntax of the input expression as:

\begin{flalign*}
	\Rfun(\texttt{x}, \aval, \amem) =\;                                 & (\{\texttt{x}\}, \aval, \amem)                                                                                           \\
	\Rfun(\texttt{a} \mid \texttt{b} \mid \texttt{s}, \aval, \amem) =\; & (\{\texttt{a} \mid \texttt{b} \mid \texttt{s}\}, \aval, \amem)                                                           \\
	\Rfun(\addrof\texttt{x}, \aval, \amem) =\;                          & (\{\pointedbyfun(\texttt{x}, \amem)\}, \aval, \amem)                                                                     \\
	\Rfun(\deref\texttt{p}, \aval, \amem) =\;                           & \left( \MId_1,\ \aval_1,\ \amem_1 \right)                                                                                \\
	\text{ where }                                                      & (\MId, \aval_1, \amem_1) = \Rfun(\texttt{p}, \aval, \amem) \;\land                                                       \\
	                                                                    & \MId_1 = \pointstofun(\MId, \amem_{1})                                                                                   \\
	\Rfun(\texttt{p.f}, \aval, \amem) =\;                               & (ac, \applysubfun(\sub, \aval_1, \amem_2), \amem_2)                                                                       \\
	\text{ where }                                                      & (ids, \aval_1, \amem_1) = \Rfun(\texttt{p}, \aval, \amem) \;\land                                                        \\
	                                                                    & (ac, \sub, \amem_2) = \accessfun(ids, \texttt{f}, \amem_1)                                                               \\
	\Rfun(\texttt{p[s]}, \aval, \amem) =\;                              & (ac, \applysubfun(\sub, \aval_1, \amem_2), \amem_2)                                                                       \\
	\text{where }                                                       & (ids, \aval_1, \amem_1) = \Rfun(\texttt{p}, \aval, \amem) \;\land                                                        \\
	                                                                    & S = \evalfun(\texttt{s}, \aval_1) \;\land                                                                                \\
	                                                                    & (ac, \sub, \amem_2) = \accessfun(ids, S, \amem_1)                                                                        \\
	\Rfun(\texttt{\{\;\}}, \aval, \amem) =\;                            & (\{\ell\}, \aval, \amem_1)                                                                                               \\
	\text{where }                                                       & (\ell, \amem_1) = \freshSsplit(\amem)                                                                                    \\
	\Rfun(\texttt{\{f_i:v_i\}}_{i=1}^n, \aval_0, \amem) =\;             & (\{\ell\}, \aval_n, \amem_n)                                                                                             \\
	\text{ where }                                                      & (\ell, \amem_0) = \freshSsplit(\amem) \;\land                                                                            \\
	                                                                    & \text{for } i = 1..n:                                                                                                    \\
	                                                                    & \quad (\MId_i, \sub_i, \amem_i) = \accessfun(\{\ell\}, \texttt{f}_i, \amem_{i-1}) \;\land                                \\
	                                                                    & \quad \aval_{s_i} = \applysubfun(\sub_i, \aval_{i-1}, \amem_i) \;\land                                                     \\
	                                                                    & \quad \aval_i = \bigsqcup_{\mId \in \MId_i} \assignfunAV(\mId, \texttt{v}_i, \aval_{s_i}, \lnot\isSumfun(\mId, \amem_i)) \\
	\Rfun(\texttt{new a|b|s}, \aval, \amem) =\;                         & (\{\ell\}, \aval_1, \amem_1)                                                                                             \\
	\text{ where }                                                      & (\ell, \amem_1) = \freshHsplit(\amem) \;\land                                                                            \\
	                                                                    & \aval_1 = \assignfunAV(\ell, \texttt{a|b|s}, \aval, \lnot\isSumfun(\ell, \amem))                                         \\
	\Rfun(\texttt{new \{\;\}}, \aval, \amem) =\;                        & (\{\ell\}, \aval, \amem_1)                                                                                               \\
	\text{where }                                                       & (\ell, \amem_1) = \freshHsplit(\amem)                                                                                    \\
	\Rfun(\texttt{new \{f_i:v_i\}}_{i=1}^n, \aval_0, \amem) =\;         & (\{\ell\}, \aval_n, \amem_n)                                                                                             \\
	\text{ where }                                                      & (\ell, \amem_0) = \freshHsplit(\amem) \;\land                                                                            \\
	                                                                    & \text{for } i = 1..n:                                                                                                    \\
	                                                                    & \quad (\MId_i, \sub_i, \amem_i) = \accessfun(\{\ell\}, \texttt{f}_i, \amem_{i-1}) \;\land                                \\
	                                                                    & \quad \aval_{s_i} = \applysubfun(\sub_i, \aval_{i-1}, \amem_i) \;\land                                                     \\
	                                                                    & \quad \aval_i = \bigsqcup_{\mId \in \MId_i} \assignfunAV(\mId, \texttt{v}_i, \aval_{s_i}, \lnot\isSumfun(\mId, \amem_i)) \\
\end{flalign*}

\subsection{Statement Semantics}

The abstract semantics is defined to abstract the split semantics while operating over abstract states.
For each statement $\stmt$, we define an abstract transfer function $\asem{\stmt}: \astate \to \astate$.
The semantics of assignment statements is defined as follows:
\begin{equation*} \text{(assign)} \quad
	\frac{
		\begin{array}{c}
			(\sub_m, \amem_1) = \assignfunAM(\lhs, \rhs, \amem)	\quad
			\aval_m = \applysubfun(\sub_m,\aval,\amem_1)                     \\
			(ids_r,\aval_r,\amem_r) = \Rfun(\rhs,\aval_m,\amem_1) 		\quad
			(ids_l,\aval_l,\amem_l) = \Rfun(\lhs,\aval_r,\amem_r)            \\
			\aval_2 = \bigsqcup_{id_l\in ids_l, id_r\in ids_r}
			\assignfunAV(id_l, id_r, \aval_l, \lnot\isSumfun(id_l, \amem_l)) \\
		\end{array}
	}{
		\asem{\lhs = \rhs} (\aval, \amem) = (\aval_2, \amem_l)
	}
\end{equation*}

Note that we can call $\Rfun$ on $\lhs$ since $\lhss \subseteq \rhss$.
Also, we can call $\assignfunAV$ on $id_l$ since the function $\Rfun$
always returns a set of memory identifiers $\MemID$ when called on a $\lhs$ expression.
Finally, the assignment on $id_l$ is strong if it is not a summary node,
and weak otherwise.

The abstract counterpart of $\texttt{ASM}(\bexp)$
refines both domain components independently using the $\assumefunAV$ and $\assumefunAM$ operations
provided by each domain:
\begin{equation*} \text{(assume)} \quad
	\frac{
		\aval_1 = \assumefunAV(\bexp, \aval) \qquad
		\amem_1 = \assumefunAM(\bexp, \amem)
	}{
		\asem{\texttt{ASM}(\bexp)} (\aval, \amem) = (\aval_1, \amem_1)
	}
\end{equation*}
Note that, since $\bexp$ cannot contain memory expressions, function $\Rfun$
does not need to be applied in assume operations.

\begin{example}[Value--Memory Interaction]
\label{sec:value-memory-example}
Similarly to \Cref{sec:split-state-example}, consider the program
\texttt{p = new 5; tmp = *p; y = tmp + 1}.
We assume the memory domain is the Andersen-style points-to analysis
(\Cref{sec:inst-andersen}) and the value domain is constant propagation;
\Cref{tab:value-memory-trace} shows the resulting abstract state after each statement.
To show the interaction between the value and memory domains, consider the
second statement \texttt{tmp = *p}, where the abstract assignment rule
(\Cref{sec:abstract-state}) proceeds as follows.
\begin{enumerate}
  \item $\assignfunAM(\texttt{tmp}, \texttt{*p}, \amem)$: a load does not
        change pointer structure, so it returns an empty substitution
        and $\amem$ unchanged.
  \item $\Rfun(\texttt{*p}, \aval, \amem)$: the dereference rule
        evaluates $\Rfun(\texttt{p}, \aval, \amem) = (\{\texttt{p}\},
        \aval, \amem)$ first, then queries
        $\pointstofun(\{\texttt{p}\}, \amem) = \{\mId_\ell\}$.
        The memory domain supplies the identifier to read: $ids_r =
        \{\mId_\ell\}$.
  \item $\Rfun(\texttt{tmp}, \aval, \amem) = (\{\texttt{tmp}\}, \aval,
        \amem)$: the variable rule gives $ids_l = \{\texttt{tmp}\}$.
  \item $\assignfunAV(\texttt{tmp}, \mId_\ell, \aval, \lnot\isSumfun(\texttt{tmp}, \amem))$:
        since \texttt{tmp} is a concrete variable (not a summary node),
        $\isSumfun(\texttt{tmp}, \amem) = \mathtt{false}$ and the assignment is strong;
        constant propagation reads $\aval(\mId_\ell) = 5$ and sets
        $\aval(\texttt{tmp}) \leftarrow 5$.
\end{enumerate}
For \texttt{y = tmp + 1} all four steps follow the same structure, with the
memory domain playing no role:
\begin{enumerate}
  \item $\assignfunAM(\texttt{y}, \texttt{tmp+1}, \amem)$: an arithmetic
        assignment does not change pointer structure; returns an empty
        substitution and $\amem$ unchanged.
  \item $\Rfun(\texttt{tmp+1}, \aval, \amem)$: the scalar rule returns the
        expression as-is, so $ids_r = \{\texttt{tmp+1}\}$.
  \item $\Rfun(\texttt{y}, \aval, \amem) = (\{\texttt{y}\}, \aval, \amem)$:
        the variable rule gives $ids_l = \{\texttt{y}\}$.
  \item $\assignfunAV(\texttt{y}, \texttt{tmp+1}, \aval, \lnot\isSumfun(\texttt{y}, \amem))$:
		the operation $\isSumfun(\texttt{y}, \amem)$ returns \texttt{false} because \texttt{y}
		is a concrete variable, resulting in a strong assignment; constant propagation evaluates
        $\aval(\texttt{tmp}) + 1 = 6$ and sets $\aval(\texttt{y}) \leftarrow 6$.
\end{enumerate}
\begin{table}[t]
  \centering
  \caption{Abstract state at each program point.
           The memory domain determines \emph{which} location is read
           ($\mId_\ell$, step~2); the value domain determines \emph{what}
           is stored there~($5$, step~4).}
  \label{tab:value-memory-trace}
  \renewcommand{\arraystretch}{1.3}
  \begin{tabular}{|l|c|c|}
    \hline
    \textbf{After statement} & $\aval$ & $\amem$ \\
    \hline
    \texttt{p = new 5}
      & $[\mId_\ell \mapsto 5]$
      & $[\texttt{p} \mapsto \{\mId_\ell\}]$ \\
    \hline
    \texttt{tmp = *p}
      & $[\mId_\ell \mapsto 5,\ \texttt{tmp} \mapsto 5]$
      & $[\texttt{p} \mapsto \{\mId_\ell\}]$ \\
    \hline
    \texttt{y = tmp + 1}
      & $[\mId_\ell \mapsto 5,\ \texttt{tmp} \mapsto 5,\ \texttt{y} \mapsto 6]$
      & $[\texttt{p} \mapsto \{\mId_\ell\}]$ \\
    \hline
  \end{tabular}
\end{table}
\end{example}

\subsection{Soundness}
\label{sec:soundness}

The soundness of our abstract domain can be established by assuming some
conditions over the transformers provided by the value and memory domains.
Specifically, we assume two conditions on how memory identifiers are concretized
in order to prove the soundness of the analysis. These are adapted from C1 and
C3 of~\cite{Ferrara16Framework}.
\begin{enumerate}[label=C\arabic*., ref=C\arabic*]
	\item\label{cond:C1} $\gammaID$ and $\gammaIDF$ must together cover all memory identifiers
	      tracked in the concretized memory state. Formally,
	      $
		      \forall \amem \in \Amem,\,
		      \forall (\smem, \gammaID, \gammaIDF) \in \gammaAM(\amem):
		      \dom{\gammaID} \cup \dom{\gammaIDF} = \memidfunAM(\amem).
	      $

	\item\label{cond:C2} Different memory identifiers represent disjoint portions of memory, that is,
		  $\dom{\gammaID} \cap \dom{\gammaIDF} = \emptyset$ (since intuitively
		  one concretizes locations and the other concretizes object fields)
		  and for any $\iota_1, \iota_2 \in \dom{\gammaIDF}$,
		  $\gammaIDF(\iota_1) \cap \gammaIDF(\iota_2) = \emptyset$ (since
		  different identifiers must represent disjoint fields).
\end{enumerate}

For the value abstraction, we employ the usual requirements of Abstract Interpretation:
the $\assignfunAV$ must correctly over-approximate the effect of assignments, the $\assumefunAV$
must correctly refine the abstract state according to boolean guards, and
$\evalfun$ must over-approximate the possible values of string expressions.
The same requirements are also employed for $\assignfunAM$ and $\assumefunAM$ in the memory domain.

To state requirements on framework-specific operators uniformly, we define for any $\mId \in \dom{\gammaID} \cup \dom{\gammaIDF}$,
with $(\smem, \gammaID, \gammaIDF) \in \gammaAM(\amem)$:
\[
    \mathsf{locs}(\mId) \;\defn\;
    \begin{cases}
        \gammaID(\mId)                    & \text{if } \mId \in \dom{\gammaID},  \\
        \{a : (a,f) \in \gammaIDF(\mId)\} & \text{if } \mId \in \dom{\gammaIDF}.
    \end{cases}
\]

The framework-specific operators must satisfy the following:
\begin{enumerate}[label=M\arabic*., ref=M\arabic*]
	\item\label{cond:M3} $\pointstofun$ must correctly include all possible
	      points-to targets of the given set of identifiers:
	      $\forall I \subseteq \cid \cup \MemID, \amem \in \Amem : \bigcup_{i'\in\pointstofun(I, \amem)}\mathsf{locs}(i') \supseteq \{ a : \exists i\in I : a\in\mathsf{locs}(i) \}$;
	\item\label{cond:M4} $\pointedbyfun$ must correctly produce the memory
	      identifier that abstracts the address of the given program variable:
	      $\forall \texttt{x} \in \cid, \amem \in \Amem : \pointedbyfun(\texttt{x}, \amem) = i$
	      where $i$ is the unique memory identifier such that
	      $\mathsf{locs}(i) = \{a_x\}$;
	\item\label{cond:M5} both overloads of $\accessfun$ must correctly resolve field accesses starting from a set of base memory identifiers and a statically known field name:
	      $\forall I \subseteq \MemID, f \in \Sigma^*, \amem \in \Amem : \accessfun(I, f, \amem) = (I', \sub, {\amem}')$
	      where $I'$ is a set of memory identifiers such that
	      $\bigcup_{i\in I'}\mathsf{locs}(i) \supseteq \{ a' : (\smem, \gammaID, \gammaIDF) \in \gammaAM({\amem}') \land \exists i\in I, a\in\mathsf{locs}(i) : (a,f,a')\in\smem \}$;
	      a similar condition must hold for the second overload, where $f$ is taken from the concretization of $\tilde{S}$;
	      it is also required that both $\accessfun$ overloads do not change what was already tracked by the domain (i.e., $\amem \sqsubseteq_{\Amem} {\amem}'$).
	\item\label{cond:M6} both $\freshSsplit$ and $\freshHsplit$ must produce
	      fresh memory identifiers that do not overlap with any existing
	      identifier in the input memory state, and must return an updated memory
	      state that tracks the new identifier; formally,
	      $\forall \amem \in \Amem : \freshSsplit(\amem) = (\mId, {\amem}')$ and
	      $\freshHsplit(\amem) = (\mId, {\amem}'')$ where
	      $\mId \not\in \memidfunAM(\amem)$, $\memidfunAM({\amem}') = \memidfunAM({\amem}'') = \memidfunAM(\amem) \cup \{\mId\}$,
	      and ${\amem}'$ and ${\amem}''$ are
	      consistent with $\amem$ on all identifiers in $\memidfunAM(\amem)$ (i.e., $\amem \sqsubseteq_{\Amem} {\amem}'$ and $\amem \sqsubseteq_{\Amem} {\amem}''$).
	\item\label{cond:M7} predicate $\isSumfun$ must correctly identify summary nodes and singletons;
	      formally, $\forall \mId \in \dom{\gammaID}: \isSumfun(\mId, \amem) = \text{\ttfamily false} \implies \|\gammaID(\mId)\| = 1$,
	      i.e., if $\mId$ is not a summary node then it concretizes to exactly one heap address.
	\item\label{cond:M8} following~\cite{Ferrara16Framework}, we expect all substitutions
	      produced by the memory domain to be coherent with the modifications that happened
	      therein~\cite[Prop.~4.1]{Ferrara16Framework}.
\end{enumerate}

Our framework uses two auxiliary functions in its semantics: $\applysubfun$ and $\Rfun$. The
former is sound by~\cite[Lem.~C.1]{Ferrara16Framework}, since \ref{cond:M7} holds and $\assignfunAV$
is sound by assumption, while the latter is novel.
In Lemma~\ref{lemma:Rfun-sound} we show that $\Rfun$ is sound, that is, that
the values computed by the expression passed as its argument are
over-approximated by the set of values computed by the rewritten expressions.

In general, soundness can be established by showing that
$\forall \astate \in \Astate, \stmt \in \stmts : \ssem{\stmt}(\gammaA(\astate)) \subseteq \gammaA(\asem{\stmt}(\astate))$.
Intuitively, every concrete behavior admitted by a split state is also
admitted by its abstract counterpart after executing the same statement.

\begin{theorem}[Soundness of the Abstract Semantics]
	\label{thm:abstract-soundness}
	If conditions \ref{cond:C1}--\ref{cond:C2}
	and \ref{cond:M3}--\ref{cond:M8} hold, then the abstract semantics is sound with respect to the split semantics.
	Concretely, for any abstract state
	$\astate = (\aval, \amem) \in \Astate$, the following must hold:
	\begin{itemize}
		\item $\ssem{\texttt{lhs = rhs}}\bigl(\gammaA(\astate)\bigr) \subseteq \gammaA\bigl(\asem{\texttt{lhs = rhs}}(\astate)\bigr)$, and
		\item $\ssem{\texttt{ASM}(\bexp)}\bigl(\gammaA(\astate)\bigr) \subseteq \gammaA\bigl(\asem{\texttt{ASM}(\bexp)}(\astate)\bigr)$.
	\end{itemize}
	Proof in \Cref{sec:proof-abstract-soundness}.
\end{theorem}

\subsection{Practical Benefits of Modularity}
\label{sec:practical-benefits}

The modular design of the framework delivers three concrete practical benefits.
\begin{itemize}
	\item \emph{Reusability}: an existing value domain or memory domain can be
	      plugged into the framework at the cost of adapting it to the common
	      interface. As the instantiation of \Cref{sec:instantiations} shows,
	      this corresponds to implementing a modest set of operators
	      ($\assignfunAV$, $\assumefunAV$, and $\evalfun$ for the value side;
	      $\assignfunAM$, $\assumefunAM$, $\memidfunAM$, $\pointstofun$,
	      $\pointedbyfun$, $\accessfun$, $\freshSsplit$, and $\freshHsplit$ for
	      the memory side), each of which usually maps naturally to operations
	      the domain already performs.
	\item \emph{Tunable precision and efficiency}: both the value and memory
	      domains can be swapped for a cheaper or a more precise abstraction
	      independently of the other component; the example at the end of
	      this Section
	      demonstrates this for constants vs.\ intervals vs.\ constant sets on
	      the same program.
	\item \emph{Compositional soundness}: each domain only needs to satisfy a
	      set of local conditions (\Cref{sec:soundness},
	      \Cref{thm:abstract-soundness}); if both domains satisfy them, the
	      combined analysis is automatically sound, i.e., there is no need to
	      prove the correctness of the combination from scratch since the
	      framework guarantees it by construction.
\end{itemize}
Implementation-wise, a modular implementation following our framework
encourages reuse and testing, and removes the need for defining custom
combinations for each pair of domains.

Compared to running the memory and value analyses
in a Cartesian product (or in isolation), the framework presents a
clear tradeoff. On the one hand, the value domain tracks information
not only on the variables but also on the memory locations. Since
the complexity of value domains usually depends on the number of
variables, our approach increases this value by adding all possible
memory identifiers, thus producing a less efficient analysis.
However, any analysis that combines memory and value information
will incur in this additional cost.
On the
other hand, it allows tracking of value information on memory
locations, something that would not be possible in the Cartesian
product (or running the analyses in isolation), thus sensibly
improving the precision of the analysis.

\begin{example}[Modularity in Action]
The following example showcases the modularity of our approach
by plugging a different value domain into the
framework, while keeping the memory domain fixed, directly affects the
precision of the result, without any other change to the analysis.

\paragraph{Program.}
Consider the $\mull$ program below.
It allocates a single heap cell initialised to $0$, non-deterministically
overwrites it with either $1$ or $3$, and then loads the result into~\texttt{q}.

\begin{lstlisting}[language=c, numbers=left, morekeywords={new},
    keywordstyle=\color{blue}\bfseries,
    caption={Program for the three-domain modularity example.},
    label={lst:two-domains}]
p = new 0;
if nd {
    *p = 1;
} else {
    *p = 3;
}
q = *p;
\end{lstlisting}

\paragraph{Memory Domain.}
To analyze this program, we use the Andersen-style points-to domain of
\Cref{sec:inst-andersen}.
The single \new expression at line~1 gives rise to one abstract memory
identifier~$\mId_\ell \in \MemID$.
After line~1, the memory state is
$\amem = [\texttt{p} \mapsto \{\mId_\ell\}]$
and this state is preserved unchanged throughout the if-else block and
the final load: neither branch adds or removes pointer targets.
The value stored at $\mId_\ell$ is tracked entirely by the value domain.

\paragraph{Value Domain 1 -- Constant Propagation ($\mathsf{CP}$).}
The $\mathsf{CP}$ domain represents each identifier by either a known
integer constant $c$ or the top element $\top$ (value unknown).
The join of two distinct constants is $\top$.
\begin{itemize}
  \item \textbf{After line~1}: $\aval(\mId_\ell)=0$.
  \item \textbf{After line~3} (then-branch): $\aval(\mId_\ell)=1$.
  \item \textbf{After line~5} (else-branch): $\aval(\mId_\ell)=3$.
  \item \textbf{After the join}: $\aval(\mId_\ell)=1\sqcup_{\mathsf{CP}}3=\top$.
  \item \textbf{After line~7}: $\aval(\texttt{q})=\top$.
\end{itemize}
Constant propagation cannot bound the value of~\texttt{q} in any way.

\paragraph{Value Domain 2 -- Intervals ($\mathsf{Int}$).}
The interval domain represents each identifier by a range $[l,u]$.
The join of two singletons $[1,1]$ and $[3,3]$ is $[1,3]$.
\begin{itemize}
  \item \textbf{After line~1}: $\aval(\mId_\ell)=[0,0]$.
  \item \textbf{After line~3} (then-branch): $\aval(\mId_\ell)=[1,1]$.
  \item \textbf{After line~5} (else-branch): $\aval(\mId_\ell)=[3,3]$.
  \item \textbf{After the join}: $\aval(\mId_\ell)=[1,1]\sqcup_{\mathsf{Int}}[3,3]=[1,3]$.
  \item \textbf{After line~7}: $\aval(\texttt{q})=[1,3]$.
\end{itemize}
The interval analysis establishes $\texttt{q}\in[1,3]$, in particular
$\texttt{q}\geq 1$ and $\texttt{q}\leq 3$.

\paragraph{Value Domain 3 -- Constant Sets ($\mathsf{CSet}_k$).}
We use the bounded variant of the constant-set domain~\cite{MineLN},
which represents each identifier by a finite set of integer constants
of cardinality at most~$k$, or $\top$ otherwise.
Join and widening both collapse to $\top$ when $|A \cup B| > k$;
within the bound, join is set union.
In this example, having $k\geq 2$ is sufficient to capture the exact set
of possible values.
\begin{itemize}
  \item \textbf{After line~1}: $\aval(\mId_\ell)=\{0\}$.
  \item \textbf{After line~3} (then-branch): $\aval(\mId_\ell)=\{1\}$.
  \item \textbf{After line~5} (else-branch): $\aval(\mId_\ell)=\{3\}$.
  \item \textbf{After the join}: $\aval(\mId_\ell)=\{1\}\sqcup_{\mathsf{CSet}}\{3\}=\{1,3\}$.
  \item \textbf{After line~7}: $\aval(\texttt{q})=\{1,3\}$.
\end{itemize}
The constant-set analysis yields the exact set of possible values,
establishing $\texttt{q}\in\{1,3\}$ precisely.

\begin{table}[h]
  \centering
  \caption{Abstract states at key points for the three instantiations.
           The memory state $\amem$ is identical throughout;
           only the value domain changes the precision of the result.}
  \label{tab:two-domains}
  \renewcommand{\arraystretch}{1.3}
  \begin{tabular}{|l|c|c|c|c|}
    \hline
    \textbf{Program point} &
    $\amem$ &
    $\aval_{\mathsf{CP}}(\mId_\ell)$ &
    $\aval_{\mathsf{Int}}(\mId_\ell)$ &
    $\aval_{\mathsf{CSet}}(\mId_\ell)$ \\
    \hline
    After line~1 & $[\texttt{p}\mapsto\{\mId_\ell\}]$ & $0$ & $[0,0]$ & $\{0\}$ \\
    \hline
    After line~3 (then) & $[\texttt{p}\mapsto\{\mId_\ell\}]$ & $1$ & $[1,1]$ & $\{1\}$ \\
    After line~5 (else) & $[\texttt{p}\mapsto\{\mId_\ell\}]$ & $3$ & $[3,3]$ & $\{3\}$ \\
    \hline
    After join          & $[\texttt{p}\mapsto\{\mId_\ell\}]$ & $\top$ & $[1,3]$ & $\{1,3\}$ \\
    \hline
    After line~7 ($\aval(\texttt{q})$) &
      $[\texttt{p}\mapsto\{\mId_\ell\}]$ & $\top$ & $[1,3]$ & $\{1,3\}$ \\
    \hline
  \end{tabular}
\end{table}

The memory state $\amem$ is identical in all three analyses.
Swapping the value domain requires no change to the memory domain or to the
abstract semantics of the framework; it affects only the precision of the
inferred value information, as summarised in \Cref{tab:two-domains}.
This demonstrates concretely that the modular design of the framework allows
independent selection of value and memory abstractions.

\end{example}

\section{Instantiations}
\label{sec:instantiations}

This section instantiates the framework of \Cref{sec:abstract-state}
with concrete abstract domains. Note that, as this framework is an extension
of~\cite{Ferrara16Framework}, all instantiations supported by the latter
can also be employed in our framework. To showcase the differences in
the integration of abstract domains w.r.t.~\cite{Ferrara16Framework},
we present two instantiations:
classical \emph{non-relational numerical domains}~\cite{MineLN} as value
domains $\Aval$, and Andersen-style points-to
analysis~\cite{Andersen94} as the memory domain $\Amem$.
Together they yield a concrete, sound analysis for $\mull$.

\subsection{Non-relational Numerical Domains}
\label{sec:non-relational}

For each $x \in \cid$ and $\mId \in \MemID$, the domain maintains a non-relational abstract value $\absdom{v}$ over-approximating the set of numerical values that an identifier may hold at run-time.
Non-relational means that the domain tracks information on each identifier independently, that is, without tracking (e.g., symbolic) relations between different identifiers.

\paragraph{Abstract Values.}
\label{sec:abstractvalues}

We assume \cite[Sect. 4.1.1]{MineLN} that the domain defines a set of abstract values $\absdom{V}$ together with a partial order $\sqsubseteq_V$, least upper bound $\sqcup_V$, greatest lower bound $\sqcap_V$, and widening $\nabla_V$ operators, and concretization $\gamma_V$ and abstraction $\alpha_V$ functions. These operators form a complete lattice, and a Galois connection with a concrete domain representing the set of numerical values 
whose lattice structure is defined with set operators. In addition, we assume that sound abstract operators for standard arithmetic operations (such as evaluating constant values, addition, multiplication, etc.) are provided.

\paragraph{Abstract States.}
\label{sec:abstractstates}
The abstract state maps each identifier to an abstract value. The various operators are defined as the point-wise lifting of the operators over abstract values \cite[Sect. 4.1.2]{MineLN}. 

\paragraph{Abstract Semantics.}
First of all, the non-relational abstract domain defines an $\widehat{\mathsf{eval}}(expr,\aval)$ function that evaluates expressions. This relies on the sound abstract operators defined by abstract values. Assignments then link the value resulting from the evaluation of the expression to the assigned identifier; when the boolean parameter of the $\assignfunAV$ is \texttt{true}, the assignment is a strong update (overwrite), and when \texttt{false}, a weak update (join, $\sqcup$, with the previous abstract value) \cite[Sect. 4.1.3]{MineLN}.

\paragraph{Examples.}
In the literature, several non-relational numerical abstract domains have been formalized. The most popular ones are:
\begin{itemize}
	\item the Sign domain \cite[Sect. 4.2]{MineLN}, tracking the sign of each identifier (e.g., if it is positive, negative, or zero);
	\item the constant propagation domain \cite[Sect. 4.3]{MineLN}, tracking if an identifier always has a constant numerical value;
	\item the interval domain \cite[Sect. 4.5]{MineLN} tracking the minimum and maximum value of each identifier; and
	\item the congruence domain \cite{congruences} tracking if the value assigned to an identifier is always congruent to $n$ modulo $m$.
\end{itemize}
All these domains define the abstract values described in this Section and can therefore be used in our framework.

\subsection{Andersen-Style Points-To Analysis}
\label{sec:inst-andersen}

We instantiate $\Amem$ with an Andersen-style points-to
analysis~\cite{Andersen94}.
For each abstract memory identifier in $\MemID$, the domain tracks the set of
identifiers it may point to.
Since the framework is flow-sensitive, the points-to sets are threaded through
the abstract semantics, growing monotonically at each statement and being joined
at merge points; this subsumes the precision of the original flow-insensitive
Andersen analysis.

\paragraph{Identifier Naming Convention.}
\label{sec:inst-andersen-naming}

We fix a static naming scheme for $\MemID$ that makes the analysis
deterministic and finitely bounded:
\begin{itemize}
	\item For each program variable $x \in \cid$, a distinguished \emph{base
		      identifier} $\mId_x \in \MemID$ represents the abstract memory location of $x$.
	\item For each syntactic allocation site $\ell$ (a \new-expression or object
	      literal occurring in the program text), a distinguished base identifier
	      $\mId_\ell \in \MemID$ represents all objects born at that site
	      (\emph{allocation-site abstraction}): every execution of the same \new
	      collapses to the same abstract identifier.
	\item For each base identifier $\mId$ and field name $f \in \cid$, a
	      distinguished \emph{field identifier} $\mId_{\mId,f} \in \MemID$ represents
	      the abstract memory cell for field $f$ of the object $\mId$.
\end{itemize}
All such identifiers are pairwise distinct.  We write
$\MemID^{\mathsf{b}} \subset \MemID$ for the set of base identifiers and
$\MemID^{\mathsf{f}} \subset \MemID$ for the set of field identifiers;
these two sets are disjoint.  The total set of identifiers is finite
(bounded by $|\cid|$ plus the number of allocation sites times the number
of field names), so no widening is needed for the memory domain.

\paragraph{Domain Definition.}
\label{sec:inst-andersen-domain}

The abstract memory domain is the set of finite partial functions
\[
	{\Amem}^{\mathsf{PT}}
	\;=\;
	(\MemID \cup \cid) \;\rightharpoonup\; \powerset{\MemID}.
\]
An element $\amem \in {\Amem}^{\mathsf{PT}}$ maps each currently
\emph{tracked} identifier (either a program variable or a memory identifier)
to its \emph{points-to set}: the
over-approximation of the concrete addresses that the corresponding
identifier may refer to.
A non-tracked identifier $\mId \notin \dom{\amem}$ is treated as having
an empty points-to set (it does not hold a pointer, or it has not yet been
allocated).

The partial order and lattice operators are pointwise:
\begin{align*}
	\amem_1 \sqsubseteq \amem_2
	 & \iff
	\dom{\amem_1} \subseteq \dom{\amem_2}
	\;\wedge\;
	\forall \mId \in \dom{\amem_1}:\; \amem_1(\mId) \subseteq \amem_2(\mId), \\
	(\amem_1 \sqcup \amem_2)(\mId)
	 & =
	\amem_1(\mId) \cup \amem_2(\mId)
	\quad\text{(treating missing entries as } \emptyset\text{)},             \\
	\bot_{{\Amem}^{\mathsf{PT}}}
	 & = \emptyset \quad (\text{empty domain}),                              \\
	\top_{{\Amem}^{\mathsf{PT}}}(\mId)
	 & = \MemID \quad \text{for all } \mId \in \MemID.
\end{align*}
Since $\MemID$ is finite, every ascending chain stabilises; widening
reduces to join: $\amem_1 \triangledown \amem_2 = \amem_1 \sqcup \amem_2$.

\paragraph{Abstract Memory State and Concretization.}
\label{sec:inst-andersen-concretization}

The concretization
$\gammaAM : {\Amem}^{\mathsf{PT}} \to
	\powerset{\Smem \times \gammaID \times \gammaIDF}$
maps each abstract state to the set of concrete memory structures
consistent with it.

For a given $\amem \in {\Amem}^{\mathsf{PT}}$, a triple
$(\smem, \gammaID, \gammaIDF) \in \gammaAM(\amem)$ must satisfy:
\begin{enumerate}
	\item \emph{Separation}: $\gammaID$ is defined on
	      $\mathsf{base}(\amem) \triangleq \dom{\amem} \cap \MemID^{\mathsf{b}}$
	      and maps each base identifier to a non-empty set of concrete addresses:
	      $\gammaID(\mId) \subseteq \caddrp$.
	      $\gammaIDF$ is defined on
	      $\mathsf{field}(\amem) \triangleq \dom{\amem} \cap \MemID^{\mathsf{f}}$
	      and maps each field identifier to a set of concrete field locations:
	      $\gammaIDF(\mId) \subseteq (\caddrh \cup \caddrr) \times \cid$.
	\item \emph{Soundness}: for every tracked identifier $\mId \in \dom{\amem}$
	      and every concrete location $\ell \in \mathsf{locs}(\mId)$
	      (where $\mathsf{locs}$ is as in condition C2 of \Cref{sec:abstract-state}),
	      if the concrete memory holds a pointer at $\ell$ (i.e., $\ell$ appears in
	      $\dom{\smem_{ih}}$ or $\dom{\smem_f}$), then the pointed-to concrete address
	      belongs to $\bigcup_{\mId' \in \amem(\mId)} \mathsf{locs}(\mId')$:
	      \[
		      \smem_{ih}(\ell) \;\in\; \bigcup_{\mId' \in \amem(\mId)} \mathsf{locs}(\mId')
		      \quad\text{or}\quad
		      \smem_f(\ell, f) \;\in\; \bigcup_{\mId' \in \amem(\mId)} \mathsf{locs}(\mId')
	      \]
	      (whichever applies to $\ell$).
\end{enumerate}
We now verify the two structural conditions required by the framework.

\begin{lemma}[C1 for ${\Amem}^{\mathsf{PT}}$]
	\label{lem:pt-C1}
	For every $\amem \in {\Amem}^{\mathsf{PT}}$ and $(\smem, \gammaID, \gammaIDF) \in \gammaAM(\amem)$:
	$\dom{\gammaID} \cup \dom{\gammaIDF} = \memidfunAM(\amem)$.
	Proof in \Cref{sec:proof-pt-C1}.
\end{lemma}

\begin{lemma}[C2 for ${\Amem}^{\mathsf{PT}}$]
	\label{lem:pt-C2}
	For every $\amem \in {\Amem}^{\mathsf{PT}}$ and
	$(\smem, \gammaID, \gammaIDF) \in \gammaAM(\amem)$, and distinct
	$\mId_1, \mId_2 \in \dom{\gammaID} \cup \dom{\gammaIDF}$:
	$\mathsf{locs}(\mId_1) \cap \mathsf{locs}(\mId_2) = \emptyset$; moreover,
	for field identifiers with the same parent base,
	$\gammaIDF(\mId_1) \cap \gammaIDF(\mId_2) = \emptyset$ as sets of $(a,f)$ pairs.
	Proof in \Cref{sec:proof-pt-C2}.
\end{lemma}

\begin{lemma}[Monotonicity of $\gammaAM$]
	\label{lem:gamma-am-pt-monotone}
	If $\amem_1 \sqsubseteq \amem_2$ then
	$\gammaAM(\amem_1) \subseteq \gammaAM(\amem_2)$.
	Proof in \Cref{sec:proof-pt-gamma-am-monotone}.
\end{lemma}

\subsubsection{Required Operations.}
\label{sec:inst-andersen-ops}

\paragraph{$\memidfunAM$.} The set of memory identifiers known to the domain
can be formalized as
$\memidfunAM(\amem) = \bigl(\dom{\amem} \cup \bigcup\codom{\amem}\bigr) \cap \MemID$.

\paragraph{$\freshSsplit$ and $\freshHsplit$.}
For an allocation site $\ell$ (determined by the program point of the call),
both functions return the canonical identifier $\mId_\ell$:
\[
	\freshSsplit(\amem) = \freshHsplit(\amem) = \mId_\ell,
\]
extending the domain of $\amem$ with $\mId_\ell \mapsto \emptyset$ if not
already present.

\paragraph{$\isSumfun$.} Every node in the Andersen points-to graph is a
summary node, as the analysis does not distinguish between different concrete
objects allocated at the same site. Thus, $\isSumfun(\mId) = \texttt{true}$ for
all $\mId \in \MemID$. Note that the analysis could easily be improved by
\begin{inlinelist}
	\item generating strong identifiers when a location is first created (i.e., it is not already part of the points-to graph), and
	\item turning the strong identifier into a summary one (by updating the internal mapping and producing the appropriate substitution) when a second object is allocated at the same site.
\end{inlinelist}

\paragraph{$\pointedbyfun$.} Computing the set of identifiers that may point to
a given identifier $\mId$ is straightforward:
\[
	\pointedbyfun(x, \amem) = \begin{cases}
		\{\mId_x\} \cup \{\mId \in \dom{\amem} \mid x \in \amem(\mId)\} & \text{if } x \in \cid; \\
		\{\mId \in \dom{\amem} \mid x \in \amem(\mId)\}                 & \text{otherwise}.
	\end{cases}
\]
That is, if $x$ is a program variable, we include its base identifier $\mId_x$.
Then, we also include any identifier $\mId$ in the domain of $\amem$ whose
points-to set contains $x$.

\paragraph{$\pointstofun$.}
As the Andersen points-to graph is represented by $\amem$ itself, this function simply
returns the points-to set of the given identifier:
\[
	\pointstofun(\mId, \amem) = \amem(\mId).
\]

\paragraph{$\accessfun$.}
For a set of base identifiers $\MId$ and a statically known field name $f$:
\[
	\accessfun(\MId, f, \amem)
	= \Bigl(\bigcup_{\mId \in \MId} \{\mId_{\mId,f}\},\;
	[],\;
	\amem_{\mathsf{ext}}\Bigr),
\]
where $\amem_{\mathsf{ext}}$ extends $\amem$ with $\mId_{\mId,f} \mapsto \emptyset$ for
each $\mId$ whose field identifier is not yet tracked.
For a lattice element $S \in \tilde{S}$:
\begin{itemize}
	\item If $S$ concretizes to a finite set of field names, apply the static case to each;
	\item If $S = \top_{\tilde{S}}$ (unknown selector), return all field
	      identifiers of all objects in $\MId$ currently tracked in $\amem$.
\end{itemize}
No substitution is returned since field identifiers are pre-determined by the
naming convention; no memory reorganization occurs. Note that this case could
be further refined with techniques such as~\cite{vinceobjects} to operate
directly with the lattice element $S$ without concretizing it.

\paragraph{$\assumefunAM$.}
Boolean guards constrain values but do not affect the pointer structure:
\[
	\assumefunAM(\bexp, \amem) = \amem.
\]

\paragraph{$\assignfunAM$.}
This is the core operation.  It updates the points-to relation based on the
syntactic form of the assignment, implementing the Andersen inclusion
constraints in a flow-sensitive manner.  All updates are
\emph{accumulative} (we only add to points-to sets, never remove), ensuring
soundness at the cost of possible imprecision.
The substitution returned is always $[]$ (no identifier renaming is needed),
so the value domain is not disturbed.

Because the naming convention makes $\freshSsplit/\freshHsplit$ deterministic,
$\assignfunAM$ can pre-compute the identifier $\mId_\ell$ that will be
allocated for any allocation site $\ell$ in the $\rhs$, and update the
points-to set of the $\lhs$ accordingly before $\Rfun$ is invoked.

The main cases are:
\begin{itemize}
	\item \emph{Address-of} ($\lhs = x \in \cid$, $\rhs = \addrof y$ with $y \in \cid$)
	      records the Andersen address-of constraint, that is, that $x$ may point to $y$'s abstract
	      location $\mId_y$:
	      \[
		      \assignfunAM(x, \addrof y, \amem)
		      = \bigl([],\; \amem[x \mapsto \amem(x) \cup \{\mId_y\}]\bigr).
	      \]

	\item \emph{Pointer copy} ($\lhs = x \in \cid$, $\rhs = y \in \cid$)
	      propagates $y$'s points-to set into $x$'s; if $y$ is not a pointer,
	      $\amem(\mId_y) = \emptyset$ and the update is a no-op:
	      \[
		      \assignfunAM(x, y, \amem)
		      = \bigl([],\; \amem[x \mapsto \amem(x) \cup \amem(y)]\bigr).
	      \]

	\item \emph{Pointer load} ($\lhs = x \in \cid$, $\rhs = \deref p$,
	      with $p \in \cid$)
	      records that $x$ may point to anything that the targets of $p$ itself point to:
	      \[
		      \assignfunAM(x, \deref p, \amem)
		      = \Bigl([],\;
		      \amem\Bigl[x \mapsto \amem(x)
			      \cup \bigcup_{\mId \in \amem(p)} \amem(\mId)\Bigr]\Bigr).
	      \]

	\item \emph{Pointer store} ($\lhs = \deref p$, $\rhs = y$, with $p \in \cid$)
	      adds $y$'s points-to set to all possible targets, since
	      $p$ may alias multiple concrete locations:
	      \[
		      \assignfunAM(\deref p, y, \amem)
		      = \Bigl([],\; \bigsqcup_{\mId \in \amem(p)} \amem\Bigl[\mId \mapsto \amem(\mId) \cup \amem(y) \Bigr]\Bigr).
	      \]

	\item \emph{Object creation} ($\lhs = x \in \cid$,
	      $\rhs \in \{\texttt{\{f$_i$:v$_i$\}}, \new\texttt{\{f$_i$:v$_i$\}}\}$
	      at allocation site $\ell$)
	      records that $x$ may now point to the freshly allocated object $\mId_\ell$
	      and initialises its field identifiers:
	      \[
		      \assignfunAM(x, \rhs, \amem)
		      = \bigl([],\;
		      \amem[x \mapsto \amem(x) \cup \{\mId_\ell\},\;
			      \mId_\ell \mapsto \emptyset,\;
			      \mId_{\mId_\ell, f_i} \mapsto \emptyset \mid i]\bigr).
	      \]

	\item \emph{Field assignment} ($\lhs = x.f$ or $\lhs = x[\sexp]$,
	      $\rhs = y$) is analogous to the pointer-store case, propagating $y$'s
	      points-to set into the appropriate field identifiers of $x$'s targets.

	\item \emph{Non-pointer assignment} (rhs is a $\vexps$ or $\aexps$
	      expression that cannot yield an address) behaves like a no-op on the memory domain:
	      \[
		      \assignfunAM(\lhs, \rhs, \amem) = ([], \amem).
	      \]
\end{itemize}

\begin{example}[Application to the Running Example]
We trace the points-to analysis on the $\mull$ program of
\Cref{lst:dancing-mull}.

\paragraph{Identifiers.}
The naming convention introduces one base identifier per variable and
one per allocation site:
\begin{itemize}
	\item $\mId_{\texttt{list}}, \mId_{\texttt{x}},
		      \mId_{\texttt{tail}}, \mId_{\texttt{n}},
		      \mId_{\texttt{tmp}}$ — variable identifiers;
	\item $\mId_{\ell_1}$ — stack allocation site for
	      $\texttt{list} = \{\}$, an empty object: no field identifiers
	      exist until a field is first written;
	\item $\mId_{\ell_7}$ — heap allocation site for
	      $\texttt{new}\{\texttt{L}:\texttt{tail}\}$
	      (fixed across loop iterations by allocation-site abstraction).
\end{itemize}
Field identifiers are written $\mId_{\mId, f}$ for field $f$ of object $\mId$.

\paragraph{After Initialisation (Lines 1--3).}
Starting from $\amem_0 = \emptyset$:

\noindent\textbf{Line 1} ($\texttt{list} = \{\}$):
$\assignfunAM$ (object-creation case) extends the state with
${\texttt{list}} \mapsto \{\mId_{\ell_1}\}$ and $\mId_{\ell_1} \mapsto \emptyset$.
Since the object literal declares no fields, no field identifiers are
created yet.

\noindent\textbf{Line 2} ($\texttt{x} = \addrof\texttt{list}$):
$\assignfunAM$ (address-of case) sets
${\texttt{x}} \mapsto \{\mId_{\texttt{list}}\}$.
$\Rfun(\addrof\texttt{list})$ calls $\pointedbyfun(\texttt{list},\amem)$
and returns $\{\texttt{x}, \mId_{\texttt{list}}\}$ directly.

\noindent\textbf{Line 3} ($\texttt{tail} = \addrof\texttt{list}$): analogous,
giving ${\texttt{tail}} \mapsto \{\mId_{\texttt{list}}\}$.

The state after line 3 is:
\[
	\amem_3 = \left\{
		{\texttt{list}}  \mapsto \{\mId_{\ell_1}\},\;
		{\texttt{x}}     \mapsto \{\mId_{\texttt{list}}\},\;
		{\texttt{tail}}  \mapsto \{\mId_{\texttt{list}}\},\;
		\mId_{\ell_1}     \mapsto \emptyset
	\right\}
\]

\paragraph{Loop Body (Lines 7--10).}

\noindent\textbf{Line 7} ($\texttt{n} = \new\{\texttt{L}:\texttt{tail}\}$):
$\assignfunAM$ (object-creation) adds
${\texttt{n}} \mapsto \{\mId_{\ell_7}\}$ and $\mId_{\ell_7} \mapsto \emptyset$,
initialising the field identifier $\mId_{\mId_{\ell_7},\texttt{L}} \mapsto \emptyset$
(there is no $\texttt{R}$ field this time, since the literal only declares $\texttt{L}$).
Then $\accessfun(\{\mId_{\ell_7}\},\texttt{L},\amem)$ produces
$\mId_{\mId_{\ell_7},\texttt{L}}$, and the field-assignment case of
$\assignfunAM$ propagates \texttt{tail}'s points-to set into it:
\[
	\mId_{\mId_{\ell_7},\texttt{L}} \;\mapsto\; \amem({\texttt{tail}})
	\;=\; \{\mId_{\texttt{list}}\}.
\]

\noindent\textbf{Line 8} ($\texttt{tmp} = \deref\texttt{tail}$):
pointer-load case of $\assignfunAM$:
\[
	{\texttt{tmp}} \;\mapsto\;
	\bigcup_{\mId \in \amem({\texttt{tail}})} \amem(\mId)
	\;=\; \amem({\texttt{list}})
	\;=\; \{\mId_{\ell_1}\}.
\]
After this, ${\texttt{tmp}}$ and ${\texttt{list}}$ both point to
$\mId_{\ell_1}$: \texttt{tmp} is an alias of \texttt{list}'s object.

\noindent\textbf{Line 9} ($\texttt{tmp.R} = \texttt{n}$):
field-assignment (weak update).
$\accessfun$ resolves \texttt{tmp}'s target to $\mId_{\ell_1}$ and
selects the field identifier $\mId_{\mId_{\ell_1},\texttt{R}}$, not yet
tracked at this point since \texttt{list}'s object had no $\texttt{R}$
field until now.
By the domain's convention, a non-tracked identifier has an empty
points-to set, so the update introduces it directly with \texttt{n}'s
points-to set:
\[
	\mId_{\mId_{\ell_1},\texttt{R}}
	\;\mapsto\; \emptyset \cup \amem({\texttt{n}})
	\;=\; \{\mId_{\ell_7}\}.
\]
This is an \emph{over-approximation}: a strong update would replace the
old value, but since $|\gammaID(\mId_{\ell_1})| > 1$ is possible, we
conservatively join.

\noindent\textbf{Line 10} ($\deref\texttt{tail} = \texttt{tmp}$):
pointer-store case.
$\amem({\texttt{tail}}) = \{\mId_{\texttt{list}}\}$, so the weak
update adds $\amem({\texttt{tmp}}) = \{\mId_{\ell_1}\}$ to
$\amem({\texttt{list}})$, which already contains $\mId_{\ell_1}$:
no change.

\noindent\textbf{Lines 12--16} (if-branch $\texttt{x} = \texttt{n}$ at line~13,
else \texttt{skip} at line~15):
pointer-copy in the true branch gives
${\texttt{x}} \mapsto \{\mId_{\texttt{list}}\} \cup \{\mId_{\ell_7}\}$.
After joining the two branches at the loop head (line~17):
${\texttt{x}} \mapsto \{\mId_{\texttt{list}}, \mId_{\ell_7}\}$.

\paragraph{Fixpoint.}
Re-entering the loop with the end-of-body state, no entry grows further
(allocation-site abstraction collapses all heap nodes created at line 7 into
the single $\mId_{\ell_7}$; all points-to sets are already at their
join values).
The analysis reaches a fixpoint after one iteration.
The final abstract memory state is:
\[
	{\amem}^* = \left\{\begin{array}{c}
		{\texttt{list}}                \mapsto \{\mId_{\ell_1}\},\;
		\mId_{\mId_{\ell_1},\texttt{R}}   \mapsto \{\mId_{\ell_7}\},\;
		{\texttt{x}}                   \mapsto \{\mId_{\texttt{list}},\mId_{\ell_7}\},\;
		{\texttt{tail}}                \mapsto \{\mId_{\texttt{list}}\}, \\
		{\texttt{n}}                   \mapsto \{\mId_{\ell_7}\},\;
		\mId_{\mId_{\ell_7},\texttt{L}}   \mapsto \{\mId_{\texttt{list}}\},\;
		{\texttt{tmp}}                 \mapsto \{\mId_{\ell_1}\},\;
		\ldots
	\end{array}
	\right\}
\]
\end{example}

\subsubsection{Soundness.}
\label{sec:inst-andersen-soundness}

The operations above are sound with respect to $\gammaAM$ because each one
can only \emph{grow} points-to sets: it adds identifiers to sets but never
removes them.  Since $\gammaAM$ is monotone
(\Cref{lem:gamma-am-pt-monotone}), any concrete state that was consistent
with the pre-state remains consistent with the post-state.  More precisely:
\begin{itemize}
	\item \emph{$\pointedbyfun$}: returns the set of identifiers referring to the given $\mId$, which by
	      the soundness condition on $\amem$ over-approximates all concrete
	      referrers. The state is unchanged.
	\item \emph{$\pointstofun$}: returns the set $\amem(\mId)$, which by
	      the soundness condition on $\amem$ over-approximates all concrete
	      targets. The state is unchanged.
	\item \emph{$\isSumfun$}: always returns \texttt{true}, since all identifiers are summaries.  The state is unchanged.
	\item \emph{$\freshSsplit/\freshHsplit$}: adds $\mId_\ell \mapsto \emptyset$.
	      The empty points-to set is a sound over-approximation of a freshly
	      allocated location (which holds no pointer yet); concrete executions that
	      later store a pointer will be captured by subsequent $\assignfunAM$ calls.
	\item \emph{$\accessfun$}: adds field identifiers with empty points-to sets;
	      soundness follows as for $\freshSsplit$.
	\item \emph{$\assignfunAM$ (pointer copy/load/store/object-creation)}: each
	      case adds identifiers to the target's points-to set, over-approximating
	      the concrete semantics: any concrete address that can reach $x$ after the
	      assignment is covered by the updated $\amem(\mId_x)$.
	      For the weak store case $\deref p = y$, updating \emph{all} potential
	      targets of $p$ (rather than exactly one) is a sound over-approximation.
	\item \emph{$\assumefunAM$}: the identity is trivially sound since boolean
	      guards do not modify memory structure.
\end{itemize}
By \Cref{thm:abstract-soundness}, these per-operation soundness arguments
together guarantee that the combined analysis (interval domain for values,
points-to domain for memory) is sound for $\mull$.

\subsection{Combining the Instantiations}
\label{sec:andersen-interval-combination}

To demonstrate how the two instantiations cooperate on a concrete setting,
we extend the running example of \Cref{lst:dancing-mull} so that each
allocated node also carries an integer field \texttt{val}.  A counter variable
$\texttt{count}$ is initialised to $1$ before the loop.  At line~7, the heap
allocation becomes
$\texttt{n} = \new\{\texttt{L}:\texttt{tail},\texttt{val}:1\}$,
and two new statements $\texttt{n.val} = \texttt{count}$
followed by $\texttt{count} = \texttt{count}+1$  are inserted at the end of the loop body.
In this variant, we describe the cooperation between $\Amem$ and $\Aval$ in a
mixed pointer/integer setting, demonstrating how Andersen analysis interacts
with non-relational value domains, such as Sign and Interval, within the
proposed framework.

\paragraph{Effect on the Memory Domain.}
The \texttt{val} field is integer-typed and holds no pointer.
\begin{itemize}
	\item \textbf{Line~7}: the object-creation case of $\assignfunAM$ introduces
	      a field identifier $\mId_{\mId_{\ell_7},\texttt{val}}$ into $\amem$,
	      mapped to $\emptyset$.  Its points-to set remains $\emptyset$ throughout
	      the analysis, so the points-to graph is unaffected.
	\item \textbf{Counter increment and store}
	      ($\texttt{count} = \texttt{count}+1$ then $\texttt{n.val} = \texttt{count}$):
	      both right-hand sides are arithmetic expressions, so the
	      \emph{non-pointer assignment} case of $\assignfunAM$ applies to each.
	      No structural change occurs and no substitution is emitted by either statement.
\end{itemize}
Consequently, the abstract memory state is identical to that of
the original example, and the points-to fixpoint is still reached after one
loop iteration.

\paragraph{Effect on the Value Domain and Domain Interaction.}
The interaction between $\Amem$ and $\Aval$ is mediated by the rewriting
function $\Rfun$ and the operations $\accessfun$ and $\assignfunAV$, as
defined in \Cref{sec:abstract-state}.
We trace the processing of the two new statements in turn.

\smallskip
\noindent\emph{Step~A: $\texttt{n.val} = \texttt{count}$.}

\begin{enumerate}
	\item \emph{Memory domain first.}
	      $\assignfunAM$ returns $([\,],\amem)$; the empty substitution leaves
	      ${\aval}'$ unchanged.

	\item \emph{Right-hand-side rewriting.}
	      $\Rfun(\texttt{count}, {\aval}, \amem)$ resolves $\texttt{count}$
	      directly (it is a concrete identifier; no $\accessfun$ call is needed)
	      and returns its current abstract value from ${\aval}$.

	\item \emph{Left-hand-side rewriting.}
	      $\Rfun(\texttt{n.val}, {\aval}, \amem)$ calls
	      $\accessfun(\{\mId_{\ell_7}\}, \texttt{val}, \amem)$, returning
	      the field identifier $\mId_{\mId_{\ell_7},\texttt{val}}$.

	\item \emph{Value update.}
	      Since $\isSumfun(\mId_{\ell_7}, \amem) = \texttt{true}$ (every
	      identifier is a summary node in the Andersen domain), the boolean flag
	      passed to $\assignfunAV$ is $\texttt{false}$, forcing a
	      \emph{weak update} (join with the existing abstract value):
	      \[
		      {\aval}' = \assignfunAV\!\bigl(\mId_{\mId_{\ell_7},\texttt{val}},\;
		      \texttt{count},\;{\aval},\;\texttt{false}\bigr).
	      \]
\end{enumerate}

\smallskip
\noindent\emph{Step~B: $\texttt{count} = \texttt{count}+1$.}
$\texttt{count} \in \cid$ is a concrete program variable, so no $\accessfun$
call is needed.
$\Rfun(\texttt{count}+1, {\aval}', \amem)$ leaves the expression unchanged
as $\texttt{count}+1$.
Because $\texttt{count}$ is not behind any pointer, $\isSumfun$ is not
consulted and $\assignfunAV$ applies a \emph{strong update}:
\[
{\aval}'' = \assignfunAV\!\bigl(\texttt{count},\;\texttt{count}+1,\;{\aval},\;\texttt{true}\bigr).
\]

\smallskip\noindent\emph{Sign domain.}
Instantiating $\Aval$ with the sign domain, the abstract sign of
$\texttt{count}$ is $+$ (strictly positive) from initialisation, since
$\texttt{count} = 1 > 0$.
After the strong update $\texttt{count} = \texttt{count}+1$, the sign
remains $+$: adding one to a positive integer is still positive.
At the loop head, joining $+$ with $+$ gives $+$; the ascending chain is
already stable after the first iteration and no widening is required.
The weak update $\texttt{n.val} = \texttt{count}$ propagates the sign into
the field identifier:
\[
	\mId_{\mId_{\ell_7},\texttt{val}}
	\;\mapsto\;
	\underbrace{+}_{\text{creation}} \;\sqcup_{\Aval}\; +
	\;=\; +.
\]
The combined state $(\aval,\amem) \in \Astate$ reaches a fixpoint after one
loop iteration.

\smallskip\noindent\emph{Interval domain.}
Instantiating $\Aval$ with the interval domain, the abstract value of
$\texttt{count}$ evolves across loop iterations via \emph{strong} updates:
\[
	\underbrace{[1,1]}_{\text{init}}
	\;\xrightarrow{+1}\; [2,2]
	\;\xrightarrow{+1}\; [3,3]
	\;\to\;\cdots
\]
This ascending chain (the pre-loop value $[1,1]$ is joined with the
post-body value at the loop head) does not stabilise without widening.
Applying $\triangledown_{\Aval}$ at the loop head gives:
\[
	[1,1]\;\triangledown_{\Aval}\;[2,2] \;=\; [1,+\infty].
\]
With $\texttt{count}$ stabilised at $[1,+\infty)$, the weak update
$\texttt{n.val} = \texttt{count}$ yields:
\[
	\mId_{\mId_{\ell_7},\texttt{val}}
	\;\mapsto\;
	\underbrace{[1,1]}_{\text{creation}} \;\sqcup_{\Aval}\; [1,+\infty]
	\;=\; [1,+\infty],
\]
after which no further growth occurs and the combined state
$(\aval, \amem) \in \Astate$ reaches a fixpoint.

\paragraph{Summary of Domain Cooperation.}
The two instantiations above highlight the modularity of the framework.
In both cases $\amem$ stabilises after one loop iteration.
The value domain $\Aval$, however, behaves differently: the sign domain
converges immediately (no widening required), while the
interval domain produces a diverging chain and requires one widening step to
stabilise $\texttt{count}$ to $[1,+\infty]$, after which
$\mId_{\mId_{\ell_7},\texttt{val}}$ follows to $[1,+\infty]$ via the weak update.
This example also illustrates the distinction between \emph{strong} and
\emph{weak} updates: $\texttt{count}$, a concrete program variable, receives a
strong (overwriting) update, whereas $\mId_{\mId_{\ell_7},\texttt{val}}$, a
field of a summary node, receives a weak (joining) update because
$\isSumfun(\mId_{\ell_7},\amem) = \texttt{true}$.
Crucially, neither component waits for the other: the memory domain reaches
its fixpoint regardless of what the value domain does with $\texttt{count}$
or $\texttt{val}$, and the value domain widens (or not) without triggering any
structural change in $\amem$.
This independent convergence is precisely the modular design the
$\Aval \otimes \Amem$ construction is intended to exploit.
The two components interact exclusively through the shared vocabulary of memory
identifiers: the memory domain creates the field identifier
$\mId_{\mId_{\ell_7},\texttt{val}}$ via $\accessfun$, and signals through
$\isSumfun$ that the corresponding node is a summary, thereby forcing a weak
update in $\Aval$.
The memory domain never inspects interval values; the value domain never
inspects the points-to graph.
This separation is the defining property of the smash-product construction
$\Astate = \Aval \otimes \Amem$ of \Cref{sec:abstract-state}: each domain is
responsible for its own fragment of the program state, and cooperation reduces
to exchanging memory identifiers and substitutions.

\section{Conclusions}
\label{sec:conclusions}

Stack and heap management is challenging in program static analysis. Precise
reasoning about dynamic and static  memory allocations and their interactions
remains a fundamental difficulty in designing scalable and accurate analyses.
In this paper, we propose a generic framework for handling heap and stack that
is agnostic to both memory and value analyses. The framework is parametric,
allowing different abstractions to be instantiated and composed depending on
the precision and performance requirements of the analysis. By clearly
separating concerns and enabling modular combination of memory and value
reasoning components, our approach provides a flexible foundation for building
static analyses that can be adapted to different programming languages and
analysis domains.

In future work, we plan to integrate this framework into a general-purpose
static analysis platform supporting multiple programming languages, such as
LiSA (Library for Static Analysis)~\cite{Negrini2023,Ferrara21,Negrini24}.
The framework of~\cite{Ferrara16Framework} is already implemented therein,
providing modular interfaces for advanced memory and value analyses. The
code\footnote{\url{https://github.com/lisa-analyzer/lisa/blob/master/lisa/lisa-sdk/src/main/java/it/unive/lisa/analysis/SimpleAbstractDomain.java}}
can be used as a starting point for an implementation of this work.
Once such development is completed, we aim at providing support to C, \Cpp, Rust, and other
low-level languages such as LLVM-IR, in LiSA. Furthermore, we aim at integrating
support for pointer arithmetics in the framework by exploiting $\evalfun$.

\begin{credits}
\subsubsection{\ackname}
This work has been partially supported by
Regione Veneto: RIR project SATCO (CUP D19J24000750007), and RIR project SUPREME (CUP D19J24000810007), and by
PR Veneto FESR 2021-2027, Az. 1.4.1 ID 251382_008101 ``SmartCycle -- Cybersecurity'', CUP: H73C25001340002.

\subsubsection{\discintname}
The authors have no competing interests to declare that are relevant to the content of this article.
\end{credits}

\bibliographystyle{splncs04}
\bibliography{ref}

\newpage
\appendix
\renewcommand\theHsection{\thesection.appendix} 

\clearpage

\section{Exploding and Imploding Objects}
\label{sec:app-explode-implode}

We provide the formal definitions of the $\explodefun$ and $\implodefun$
functions introduced in \Cref{sec:explode-implode-obj}, together with their
key invertibility lemmas.

\begin{definition}[Explode Function]
\label{def:explode-function} Let
	$
		\decomposefun :
		((\cid \cup \caddrh \cup \caddrr)\allowbreak
		\times \cobj)
		\to
		((\cid \cup \caddrh \cup \caddrr) \to \tilde{\cobj})
	$
	be a function that, given a pair $(k, o)$ where $k$ is the key of a
	concrete object $o = \{ f_i : v_i \}$,
	returns an exploded object as a set of new mappings representing $o$
	in which each field value is either left unchanged (if it is a primitive
	value or an address) or replaced by a fresh recursive address (if it is
	a nested object):
	\begin{align*}
		\decomposefun(k,\{f_i:v_i\}) = & \left\{ (k, \{f_i : \tilde{v}_i\})
		\;\middle\vert\; \tilde{v}_i =
		\begin{cases}
			v_i, & \text{if } v_i \in \cval \cup \caddr, \\[1ex]
			r_i, & \text{if } v_i \in \cobj,
		\end{cases}
		\right\}                                                                                      \\
		                               & \cupdot\; \bigcupdot_{v_i \in \cobj} \decomposefun(r_i,v_i).
	\end{align*}
	where $r_i \in \caddrr$ is a fresh recursive address. The function
	generates a fresh recursive address $r_i$ (that does not occur elsewhere in the
	state) for each field $f_i$ with $v_i \in \cobj$ and then applies
	$\decomposefun$ recursively to $v_i$ with key $r_i$.
	Then, let $O: (\cid \cup \caddrh) \to \cobj$ be a mapping from keys to
	concrete objects. The function $
		\explodefun : ((\cid \cup \caddrh) \to \cobj)
		\to ((\cid \cup \caddrh \cup \caddrr) \to \tilde{\cobj})
	$
	is defined as:
	\[
		\explodefun(O) = \bigcupdot_{k \in \dom{O}} \decomposefun(k,O(k))
	\]
\end{definition}

\begin{lemma}[Key Generation in $\explodefun$]
	\label{lemma:explode-key-generation}
	Let $O : (\cid \cup \caddrh) \to \cobj$ be a concrete object mapping.
	Then, the following properties hold:
	\begin{enumerate}
		\item $\dom{O} \subseteq \cid \implies
			      \dom{\explodefun(O)} \subseteq (\cid \cup \caddrr)$;
		\item $\dom{O} \subseteq \caddrh \implies
			      \dom{\explodefun(O)} \subseteq \caddrh$.
	\end{enumerate}
\end{lemma}

In other words, concrete objects identified by $\cid$ generate an
exploded object mapping with keys in $\cid \cup \caddrr$, while
concrete objects identified by $\caddrh$ generate an exploded object
mapping with keys only in $\caddrh$. Both properties follow from the definition
of $\decomposefun$: since it only introduces keys in $\caddrr$ for nested
objects, and it does not modify the original keys, the resulting exploded
mapping can only contain keys in $\cid \cup \caddrr$ if the input mapping
only contains keys in $\cid$. Conversely, if the input mapping
only contains keys in $\caddrh$, nested objects cannot be part of the mapping
(recall from \Cref{sec:concrete-assign} that nested heap objects are always
created by storing the address generated by the semantics of $\new$ inside
an object's field, and are thus already flat in memory using addresses in $\caddrh$),
and thus $\decomposefun$ will never add any address in $\caddrr$.

\begin{definition}[Implode Function]
\label{def:implode-function} Let
	$
		\composefun :
		((\cid \cup \caddrh \cup \caddrr) \times ((\cid \cup \caddrh \cup \caddrr) \to \tilde{\cobj}))
		\to
		\cobj
	$
	~be a function that, given a pair $(k, M)$ where $k$ is a key in the
	domain of the mapping $M = \{(k_i, \{f^i_j : \tilde{v}^i_j)\}$, collapses
	the exploded object at key $k$ by replacing each field value $\tilde{v}^i_j$ with its
	corresponding object if $\tilde{v}^i_j \in \caddrr$, recursively:
	\begin{align*}
		\composefun(k, M) = \{f_i : v_i\} : &
		\{f_i : \tilde{v}_i\} \in M(k)                                                                                      \\
		                                    & \land v_i = \begin{cases}
			                                                  \composefun(\tilde{v}_i, M) & \text{if } \tilde{v}_i \in \caddrr, \\
			                                                  v_i                         & \text{otherwise.}
		                                                  \end{cases}
	\end{align*}
	In other words, the exploded object $\{f_i : \tilde{v}_i\}$ corresponding to
	the key $k$ in $M$ is reassembled by replacing each field
	value $\tilde{v}_i$ with its resolved value. Specifically, if $\tilde{v}_i \in \caddrr$
	is a recursive address, then it is replaced by the (recursively)
	composed object $\composefun(\tilde{v}_i, M)$; otherwise, the field value remains unchanged.
	Then, we define the function
	$
		\implodefun : ((\cid \cup \caddrh \cup \caddrr) \to \tilde{\cobj})
		\to ((\cid \cup \caddrh) \to \cobj)
	$
	as:
	\[
		\implodefun(M) = \bigcupdot_{k \in \dom{M\restrdom{\scriptsize \cid \cup \caddrh}}} \{(k, \composefun(k, M))\}
	\]

	that is, it returns the object mapping obtained by recomposing only the keys
	of $M$ that correspond to identifiers and heap addresses, but not recursive addresses in
	$\caddrr$.
\end{definition}

Moreover, $\explodefun$ and $\implodefun$ are inverses of each other in the following sense.

\begin{lemma}[Invertibility of $\decomposefun$ and $\composefun$]
	\label{lemma:compose-decompose-invertibility}
	Consider a concrete object $o \in \cobj$ and a mapping
	$\tilde{o} = (\cid \cup \caddrh \cup \caddrr) \to \tilde{\cobj}$
	representing the exploded version of $o$, namely
	$\tilde{o} = \decomposefun(k,o)$ for some $k \in (\cid \cup \caddrh)$.
	We have that
	\begin{inlinelist}
		\item $\composefun(k, \decomposefun(k,o)) = o$, and that
		\item $\decomposefun(k,\composefun(k, \tilde{o})) = \tilde{o}$.
	\end{inlinelist}
\end{lemma}

\begin{proof}
The lemma is proven by induction on the structure of the concrete or exploded
objects that the functions are applied to. The proof is split into two parts,
one for each direction of the invertibility. Each part formally defines the
notion of depth of an object, which is used as the induction metric.

\paragraph{Part 1: \texorpdfstring{$\composefun(k, \decomposefun(k, o)) = o$}{compose(k, decompose(k, o)) = o}.}

Let's first define the depth of a concrete object $\{f_i:v_i\}$ as follows:
\[
	\texttt{depth}(\{f_i:v_i\}) =
	\begin{cases}
		0 & \text{if } \forall v_i : v_i \in \cval \cup \caddr \cup \bot, \\
		\max\limits_{v_i \in {\scalebox{0.6}{\cobj}}} \texttt{depth}(v_i) + 1 & \text{otherwise.}
	\end{cases}
\]
W.l.o.g.,
we assume that each object considered in this proof has exactly
one nested object, stored in its first field. Let $o_0=\{f^0_0:v^0_0, \dots, f^0_{m_0}:v^0_{m_0}\}$
be a non-recursive object with $m_0$ fields, each mapped to a value in $\cval \cup \caddr \cup \bot$.
Then, we construct objects of depth $d$ as $o_d = \{f^d_0:o_{d-1}, f^d_1:v^d_1, \dots, f^d_{m_d}:v^d_{m_d}\}$,
where $v^d_i \in \cval \cup \caddr \cup \bot$ for all $i\in[1,m_d]$.
Using such notation, the object $o_{n+1}$ has depth $n+1$.
The proof can be easily extended to multiple nested
objects by applying the same reasoning to each of them.
We will prove the lemma by induction on the depth of the object $o = \{f_i:v_i\}$.

\paragraph{Base Case: $o = o_0, \texttt{depth}(o_0) = 0$.}
\begin{flalign*}
	& \composefun(k, \decomposefun(k, o_0)) & \\
	=\;& \composefun(k, \{(k, o_0)\}) & \Lbag\text{def. } \decomposefun, \texttt{ depth}(o_0)=0\Rbag\\
	=\;& o_0 & \Lbag\text{def. } \composefun\Rbag~\qed
\end{flalign*}

\paragraph{Inductive Step: $o = o_{n+1}, \texttt{depth}(o_{n+1}) = n+1$.}
During the proof, we assume that $\composefun$ can correctly reconstruct
decomposed objects obtained from $\decomposefun$ of depth at most $n$.

\begin{flalign*}
	& \composefun(k, \decomposefun(k, o_{n+1})) &\\
	=\;& \composefun(k, \decomposefun(k, \\
	   & \quad \{f^{n+1}_0:o_{n}, f^{n+1}_1:v^{n+1}_1, \dots, f^{n+1}_{m_{n+1}}:v^{n+1}_{m_{n+1}}\}) & \Lbag \text{def. } o_{n+1}\Rbag\\
	=\;& \composefun(k, \{(k, \\
	   & \quad \{f^{n+1}_0:r_{n}, f^{n+1}_1:v^{n+1}_1, \dots, f^{n+1}_{m_{n+1}}:v^{n+1}_{m_{n+1}}\})\} & \\
	   & \quad \cupdot\; \decomposefun(r_n, o_n)) & \Lbag\text{def. } \decomposefun \Rbag \\
	=\;& \{f^{n+1}_0: \composefun(r_{n}, \decomposefun(r_n, o_n)), \\
	   & \quad f^{n+1}_1:v^{n+1}_1, \dots, f^{n+1}_{m_{n+1}}:v^{n+1}_{m_{n+1}}\}) & \Lbag\text{def. } \composefun\Rbag\\
	=\;& \{f^{n+1}_0: o_n, f^{n+1}_1:v^{n+1}_1, \dots, f^{n+1}_{m_{n+1}}:v^{n+1}_{m_{n+1}}\}) & \Lbag\text{hypothesis} \Rbag\\
	=\;& o_{n+1} & \Lbag\text{def. } o_{n+1}\Rbag~\qed \\
\end{flalign*}

\paragraph{Part 2: \texorpdfstring{$\decomposefun(k, \composefun(k, \tilde{o})) = \tilde{o}$}{decompose(k, compose(k, o~)) = o~}.}

Let's first define the depth of a decomposed object mapping $\tilde{o} = \{(k_0, \{f^0_i:\tilde{v}^0_i\}), \dots, (k_l, \{f^l_i:\tilde{v}^l_i\}) \}$ as follows:
\[
	\texttt{depth}(k, \tilde{o}) =
	\begin{cases}
		0 & \text{if } \tilde{o}(k) = \{f_i:\tilde{v}_i\} \\
		  & \quad \land\; \forall \tilde{v}_i : \tilde{v}_i \notin \caddrr, \\
		\max\limits_{\tilde{v}_i \in \codom{\tilde{o}(k)} \cap {\scalebox{0.6}{\caddrr}}} \texttt{depth}(\tilde{v}_i, \tilde{o}) + 1 & \text{otherwise.}
	\end{cases}
\]
W.l.o.g.,
we assume that each decomposed object considered in this proof has exactly
one nested decomposed object, stored in its first field. Let
$\tilde{o}^k_0=\{(k, \{f^0_0:\tilde{v}^0_0, \dots, f^0_{m_0}:\tilde{v}^0_{m_0}\})\}$
be an object mapping containing exactly one non-recursive object with $m_0$
fields, each mapped to a value not in $\caddrr$, whose entry object is indexed with $k$.
Then, we construct mappings of depth $d$ as
$\tilde{o}^k_d = \{(k, \{f^d_0:r_{d-1}, f^d_1:\tilde{v}^d_1, \dots, f^d_{m_d}:\tilde{v}^d_{m_d}\}), \dots, (r_0, \{f^0_0:\tilde{v}^0_0, \dots, f^0_{m_0}:\tilde{v}^0_{m_0}\})\}$,
where $\tilde{v}^l_i \notin \caddrr$ for all $l\in[d,1], i\in[1,m_l]$, and where the
object mapped to $r_0$ corresponds to $\tilde{o}^k_0(k)$.
Using such notation, the mapping $\tilde{o}_{n+1}$ has depth $n+1$.
The proof can be easily extended to multiple nested
objects by applying the same reasoning to each of them.
We will prove the lemma by induction on the depth of the object $\tilde{o}$.

\paragraph{Base Case: $\tilde{o} = \tilde{o}^k_0, \texttt{depth}(\tilde{o}^k_0) = 0$.}

\begin{flalign*}
	& \decomposefun(k, \composefun(k, \tilde{o}^k_0)) &\\
	=\;& \decomposefun(k, \{f^0_i:\tilde{v}^0_i\}) & \Lbag\text{def. } \composefun, \texttt{depth}(\tilde{o}^k_0)=0\Rbag\\
	=\;& \{(k, \{f^0_i:\tilde{v}^0_i\})\} & \Lbag\text{def. } \decomposefun\Rbag\\
	=\;& \tilde{o}^k_0 & \Lbag\text{def. } \tilde{o}^k_0 \Rbag~\qed
\end{flalign*}

\paragraph{Inductive Step: $\tilde{o} = \tilde{o}^k_{n+1}, \texttt{depth}(\tilde{o}^k_{n+1}) = n+1$.}
During the proof, we assume that $\decomposefun$ can correctly decompose
reconstructed objects obtained from $\composefun$ of depth at most $n$.
The proof starts from the object mapping $\tilde{o}^k_{n+1}$, whose extended form
is $\{(k, \{f^{n+1}_0:r_{n}, f^{n+1}_1:\tilde{v}^{n+1}_1, \dots, f^{n+1}_{m_{n+1}}:\tilde{v}^{n+1}_{m_{n+1}}\}),\allowbreak
(r_n, \{f^{n}_0:r_{n-1}, f^{n}_1:\tilde{v}^{n}_1, \dots, f^{n}_{m_n}:\tilde{v}^{n}_{m_{n}}\}),
\dots,
(r_0, \{f^{0}_0:\tilde{v}^0_{0}, f^{1}_1:\tilde{v}^{0}_1, \dots, f^{0}_{m_0}:\tilde{v}^{0}_{m_{0}}\}) \}$.
For the sake of clarity, during the proof we will represent the all pairs
having $r_i$ as the key by using the overlined key $\overline{r}_i$. For instance, we will
write $\tilde{o}^k_{n+1}$ as
$\{(k, \{f^{n+1}_0:r_{n}, f^{n+1}_1:\tilde{v}^{n+1}_1, \dots, f^{n+1}_{m_{n+1}}:\tilde{v}^{n+1}_{m_{n+1}}\}),
\overline{r}_n, \dots, \overline{r}_0 \}$.

\begin{flalign*}
	& \decomposefun(k, \composefun(k, \tilde{o}^k_{n+1})) &\\
	=\;& \decomposefun(k, \composefun(k, & \\
	   & \quad \{(k, \{f^{n+1}_0:r_{n}, f^{n+1}_1:\tilde{v}^{n+1}_1, \dots, f^{n+1}_{m_{n+1}}:\tilde{v}^{n+1}_{m_{n+1}}\}), & \\
	   & \quad \overline{r}_n, \dots, \overline{r}_0\})) & \Lbag \text{def. } \tilde{o}^k_{n+1}\Rbag\\
	=\;& \decomposefun(k, \{f^{n+1}_0: \composefun(r_{n}, \{ \overline{r}_n, \dots, \overline{r}_0 \}), & \\
	   & \quad f^{n+1}_1:\tilde{v}^{n+1}_1, \dots, f^{n+1}_{m_{n+1}}:\tilde{v}^{n+1}_{m_{n+1}}\}) & \Lbag \text{def. } \composefun\Rbag\\
	=\;& \{ (k, \{f^{n+1}_0:r_{n}, f^{n+1}_1:\tilde{v}^{n+1}_1, \dots, f^{n+1}_{m_{n+1}}:\tilde{v}^{n+1}_{m_{n+1}}\}) \} & \\
	   & \quad \cupdot\; \decomposefun(r_n, \composefun(r_{n}, \{ \overline{r}_n, \dots, \overline{r}_0 \})) & \Lbag \text{def. } \decomposefun\Rbag\\
	=\;& \{ (k, \{f^{n+1}_0:r_{n}, f^{n+1}_1:\tilde{v}^{n+1}_1, \dots, f^{n+1}_{m_{n+1}}:\tilde{v}^{n+1}_{m_{n+1}}\}) \} & \\
	   & \quad \cupdot\; \{ \overline{r}_n, \dots, \overline{r}_0 \} & \Lbag \text{hypothesis} \Rbag\\
	=\;& \{(k, \{f^{n+1}_0:r_{n}, f^{n+1}_1:\tilde{v}^{n+1}_1, \dots, f^{n+1}_{m_{n+1}}:\tilde{v}^{n+1}_{m_{n+1}}\}), & \\
	   & \quad \overline{r}_n, \dots, \overline{r}_0\} & \Lbag \text{def. } \cupdot\Rbag\\
	=\;& o_{n+1} & \Lbag\text{def. } o_{n+1}\Rbag~\qed \\
\end{flalign*}

In the proof above, $o_{n-1}$ is the object reconstructed by $\composefun(r_n,
\tilde{o}^k_{n+1})$ from the subtree of $\tilde{o}^k_{n+1}$ rooted at $r_n$.
\end{proof}

\begin{lemma}[Invertibility of $\explodefun$ and $\implodefun$]
	\label{lemma:explode-implode}
	For any concrete object mapping
	$O : (\cid \cup \caddrh) \to \cobj$,
	we have $\implodefun(\explodefun(O)) = O$.
	Furthermore, for any exploded object mapping
	$M : (\cid \cup \caddrh \cup \caddrr) \to \tilde{\cobj}$
	that is well-formed (i.e., without unreachable recursive addresses in $\caddrr$),
	it holds that $\explodefun(\implodefun(M)) = M$.
\end{lemma}

\begin{proof}
\begin{flalign*}
	&\implodefun(\explodefun(O)) \\
	=\;& \implodefun(\explodefun(\{(k_0,\{f^0_j: v^0_j\}), \dots, (k_l,\{f^l_j: v^l_j\})\})) & \Lbag\text{def. } O\Rbag\\
	=\;& \implodefun(\decomposefun(k_0, \{f^0_j: v^0_j\}) \cupdot \dots \\
	   & \quad \cupdot \decomposefun(k_l, \{f^l_j: v^l_j\})) & \Lbag \text{def. } \explodefun\Rbag\\
	   & \emph{Let } A \emph{ denote the whole argument of } \implodefun \\
	=\;& \implodefun(A) \\
	=\;& \{(k_0, \composefun(k_0, A)) \cupdot \dots \cupdot (k_l, \composefun(k_l, A))\} & \Lbag \text{def. } \implodefun\Rbag\\
	=\;& \{(k_0, \composefun(k_0, \decomposefun(k_0, \{f^0_j: v^0_j\}))) \cupdot \dots \\
	   & \quad \cupdot (k_l, \composefun(k_l, \decomposefun(k_l, \{f^l_j: v^l_j\})))\} & \Lbag \text{see note below}\Rbag\\
	=\;& \{(k_0, \{f^0_j: v^0_j\}) \cupdot \dots \cupdot (k_l, \{f^l_j: v^l_j\})\} & \Lbag \text{by \Cref{lemma:compose-decompose-invertibility}}\Rbag\\
	=\;& O & \Lbag \text{def. } O \Rbag~\qed \\
\end{flalign*}
\emph{Note.} By \Cref{def:explode-function}, we know that $\decomposefun(k_i, \cdot)$
generates a tree of bindings rooted at $k_i$ and containing only
fresh addresses that are unique in the whole exploded mapping. Also,
by (\Cref{def:implode-function}) we know that $\composefun(k_i, \cdot)$
traverses only bindings reachable from $k_i$ in its argument. Therefore,
during the execution of $\composefun(k_i, A)$,
the only bindings in $A$ reachable from $k_i$ are exactly those produced by
$\decomposefun(k_i, \cdot)$, so the reduction is valid.
Note that for $k_i \in \caddrh$ (heap objects), $\decomposefun$ acts as the identity
by \Cref{lemma:explode-key-generation}, so these entries pass through unchanged.

\begin{flalign*}
	&\explodefun(\implodefun(M)) \\
	=\;& \explodefun(\{(k_0, \composefun(k_0, M)), \dots, (k_l, \composefun(k_l, M))\}) & \Lbag \text{def. } \implodefun\Rbag\\
	=\;& \decomposefun(k_0, \composefun(k_0, M)) \cupdot \dots \\
	   & \quad\cupdot \decomposefun(k_l, \composefun(k_l, M)) & \Lbag \text{def. } \explodefun\Rbag\\
	=\;& \{(k_0, M(k_0)), \dots, (k_l, M(k_l))\} & \Lbag \text{by \Cref{lemma:compose-decompose-invertibility}}\Rbag\\
	=\;& M & \Lbag \text{def. } M \Rbag~\qed\\
\end{flalign*}
\end{proof}

\begin{lemma}[Distributivity of $\implodefun$ over $\cupddot$]
	\label{lemma:implode-distributive}
	Let $M_1 : D_1 \to \tilde{\cobj}$ and $M_2 : D_2 \to \tilde{\cobj}$ be
	two well-formed exploded object mappings (i.e., obtained by $\explodefun$)
	with disjoint domains $D_1 \cap D_2 = \emptyset$.
	Then $\implodefun(M_1) \cupddot \implodefun(M_2) = \implodefun(M_1 \cupddot M_2)$.
\end{lemma}

\begin{proof}
Let $K_i = \dom{M_i\restrdom{\cid\cup\caddrh}}$ for $i=1,2$.
Since $D_1\cap D_2=\emptyset$, we have $K_1\cap K_2=\emptyset$ and
$\dom{(M_1\cupddot M_2)\restrdom{\cid\cup\caddrh}}=K_1\cup K_2$.
Being $M_1$ well-formed, we have that
for any $k\in K_1$ the $\caddrr$-addresses reachable from $k$ lie entirely
within $D_1$, so $(M_1\cupddot M_2)(r)=M_1(r)$ for every such $r$,
and hence $\composefun(k,M_1\cupddot M_2)=\composefun(k,M_1)$. The same
holds for $M_2$, and thus we have that
$\composefun(k,M_1\cupddot M_2)=\composefun(k,M_2)$ for $k\in K_2$.
\begin{flalign*}
	   & \implodefun(M_1\cupddot M_2) & \\
	=\;& \bigcupdot_{k\in K_1\cup K_2}\{(k,\composefun(k,M_1\cupddot M_2))\}
	   & \Lbag\text{def. }\implodefun\Rbag\\
	=\;& \bigcupdot_{k\in K_1}\{(k,\composefun(k,M_1))\}\\
	   & \quad \cupddot \; \bigcupdot_{k\in K_2}\{(k,\composefun(k,M_2))\}
	   & \Lbag\text{above; }K_1\cap K_2=\emptyset\Rbag\\
	=\;& \implodefun(M_1)\cupddot\implodefun(M_2)
	   & \Lbag\text{def. }\implodefun\Rbag~\qed
\end{flalign*}
\end{proof}

\begin{lemma}[Distributivity of $\explodefun$ over $\cupdot$]
	\label{lemma:explode-distributive}
	Let $O_1 : D_1 \to \cobj$ and $O_2 : D_2 \to \cobj$ be concrete object
	mappings with $D_1 \cap D_2 = \emptyset$.
	Then $\explodefun(O_1 \cupdot O_2) = \explodefun(O_1) \cupdot \explodefun(O_2)$.
\end{lemma}

\begin{proof}
Let $O_1$ and $O_2$ be two object mappings with disjoint domains
$D_1$ and $D_2$, respectively.
\begin{flalign*}
	   & \explodefun(O_1\cupdot O_2) & \\
	=\;& \bigcupdot_{k\in D_1\cup D_2}\decomposefun(k,(O_1\cupdot O_2)(k))
	   & \Lbag\text{def. }\explodefun\Rbag\\
	=\;& \bigcupdot_{k\in D_1}\decomposefun(k,O_1(k))\\
	   &\quad\cupdot\;\bigcupdot_{k\in D_2}\decomposefun(k,O_2(k))
	   & \Lbag\text{def. }\cupdot;\ D_1\cap D_2=\emptyset\Rbag\\
	=\;& \explodefun(O_1)\cupdot\explodefun(O_2)
	   & \Lbag\text{def. }\explodefun, \decomposefun\Rbag~\qed
\end{flalign*}
Remember that $\decomposefun$ introduces only fresh $\caddrr$ addresses via a
global fresh address generator that maintains state across all invocations.
This ensures that decompositions of $O_1$ and $O_2$ generate disjoint sets of
addresses, making their union a valid $\cupdot$.
\end{proof}

\pagebreak
\section{Proofs}
\label{sec:proofs}

\subsection{Proof of \texorpdfstring{$\restriction$}{Restriction} and \texorpdfstring{$\doublerestriction$}{Double Restriction} Properties (\Cref{par:restriction-properties})}
\label{sec:proof-restriction-properties}

Throughout, let first-order $f\colon K\to V$, $g\colon A\to B$,
and second-order $M_1\colon D_1\to(L_1\to V_1)$, $M_2\colon D_2\to(L_2\to V_2)$.
Conditions are stated per proof.
Properties (i)--(iii) each show both the $\restriction$ case and the
$\doublerestriction$ variant; for (ii), the $\doublerestriction$ variant uses
$\cupddot$ in place of $\cup$; property (iv) only involve $\doublerestriction$.

\paragraph{(i) Domain Additivity.}
For $X_1,X_2\subseteq K$:
\begin{flalign*}
	& f \restrdom{X_1} \cup f \restrdom{X_2} & \\
	=\;& \{ (x, f(x)) \mid x \in X_1 \} \cup \{ (x, f(x)) \mid x \in X_2 \} & \Lbag\text{def. } \restriction\Rbag\\
	=\;& \{ (x, f(x)) \mid (x \in X_1 ) \vee (x \in X_2) \} & \Lbag\text{prop. }\cup\Rbag\\
	=\;& \{ (x, f(x)) \mid x \in X_1 \cup X_2 \} & \Lbag\text{def. } \cup\Rbag\\
	=\;& f \restrdom{X_1 \cup X_2} & \Lbag\text{def. } \restriction \Rbag~\qed
\end{flalign*}
For $X_1,X_2\subseteq D_1$:
\begin{flalign*}
	& M_1 \drestrdom{X_1} \cup M_1 \drestrdom{X_2} & \\
	=\;& \{ (k, M_1(k)) \mid k \in X_1 \} \cup \{ (k, M_1(k)) \mid k \in X_2 \} & \Lbag\text{def. } \doublerestriction \Rbag\\
	=\;& \{ (k, M_1(k)) \mid (k \in X_1) \vee (k \in X_2) \} & \Lbag\text{prop. } \cup \Rbag\\
	=\;& \{ (k, M_1(k)) \mid k \in X_1 \cup X_2 \} & \Lbag\text{def. } \cup \Rbag\\
	=\;& M_1 \drestrdom{X_1 \cup X_2} & \Lbag\text{def. } \doublerestriction \Rbag~\qed
\end{flalign*}

\paragraph{(ii) Co-domain Additivity.}
For $Y_1,Y_2\subseteq V$ with $Y_1\cap Y_2=\emptyset$:
\begin{flalign*}
	& f \restrcodom{Y_1} \cup f \restrcodom{Y_2} & \\
	=\;& \{ (x, f(x)) \mid f(x) \in Y_1 \} \cup \{ (x, f(x)) \mid f(x) \in Y_2 \} & \Lbag\text{def. } \restriction\Rbag\\
	=\;& \{ (x, f(x)) \mid (f(x) \in Y_1) \vee (f(x) \in Y_2) \} & \Lbag\text{prop. }\cup\Rbag\\
	=\;& \{ (x, f(x)) \mid f(x) \in Y_1 \cup Y_2 \} & \Lbag\text{def. }\cup\Rbag\\
	=\;& f \restrcodom{Y_1 \cup Y_2} & \Lbag\text{def. } \restriction \Rbag~ \qed
\end{flalign*}
For $Y_1,Y_2\subseteq V_1$ with $Y_1\cap Y_2=\emptyset$
(so that inner maps $M_1(k)\restrcodom{Y_1}$ and $M_1(k)\restrcodom{Y_2}$ are domain-disjoint,
making $\cupddot$ valid):
\begin{flalign*}
	& M_1 \drestrcodom{Y_1} \cupddot M_1 \drestrcodom{Y_2} & \\
	=\;& \{ (k, M_1(k) \restrcodom{Y_1}) \mid k \in D_1 \} \cupddot \{ (k, M_1(k) \restrcodom{Y_2}) \mid k \in D_1 \} &\Lbag \text{def. } \doublerestriction \Rbag\\
	=\;& \{ (k, M_1(k) \restrcodom{Y_1} \cupdot M_1(k) \restrcodom{Y_2}) \mid k \in D_1 \} &\Lbag \text{def. } \cupddot;\ Y_1\cap Y_2=\emptyset \Rbag\\
	=\;& \{ (k, M_1(k) \restrcodom{Y_1 \cup Y_2}) \mid k \in D_1 \} &\Lbag \text{co-domain add. } \restriction\Rbag \\
	=\;& M_1 \drestrcodom{Y_1 \cup Y_2} &\Lbag \text{def. } \doublerestriction \Rbag~ \qed
\end{flalign*}

\paragraph{(iii) Distribution over $\cupdot$.}
For $K\cap A=\emptyset$:
\begin{flalign*}
	& (f \cupdot g) \restr{W}{Z} & \\
	=\;& \{\, (x,(f \cupdot g)(x)) \mid x \in W \wedge (f \cupdot g)(x) \in Z \,\} &\Lbag \text{def. }\restriction \Rbag\\
	=\;& \{\, (x,(f \cupdot g)(x)) \mid x \in (W\cap K) \\
	   & \quad\cup (W\cap A) \wedge (f \cupdot g)(x) \in Z \,\} &\Lbag W \subseteq K\cup A\Rbag\\
	=\;& \{\, (x, f(x)) \mid x \in W\cap K \wedge f(x) \in Z \,\} & \\
	   & \quad\cup\; \{\, (x, g(x)) \mid x \in W\cap A \wedge g(x) \in Z \,\} &\Lbag \text{def. } \cupdot\Rbag\\
	=\;& f \restr{(W\cap K)}{Z} \cup\; g \restr{(W\cap A)}{Z} &\Lbag \text{def. }\restriction\Rbag\\
	=\;& ( f \restr{(W\cap K)}{Z} ) \;\cupdot\; ( g \restr{(W\cap A)}{Z} ) &\Lbag K \cap A = \emptyset \Rbag~ \qed
\end{flalign*}
For $D_1\cap D_2=\emptyset$:
\begin{flalign*}
	& (M_1 \cupdot M_2) \drestr{W}{Z} & \\
	=\;& \{ (x, (M_1 \cupdot M_2)(x) \restrcodom{Z}) \mid x \in W \} &\Lbag \text{def. } \doublerestriction \Rbag\\
	=\;& \{ (x, M_1(x) \restrcodom{Z}) \mid x \in W \cap D_1 \} \\
	   & \quad \cup \{ (x, M_2(x) \restrcodom{Z}) \mid x \in W \cap D_2 \} &\Lbag \text{def. } \cupdot \Rbag\\
	=\;& M_1 \drestr{(W \cap D_1)}{Z} \cup M_2 \drestr{(W \cap D_2)}{Z} &\Lbag \text{def. } \doublerestriction \Rbag\\
	=\;& ( M_1 \drestr{(W \cap D_1)}{Z} ) \;\cupdot\; ( M_2 \drestr{(W \cap D_2)}{Z} ) &\Lbag D_1 \cap D_2 = \emptyset \Rbag~ \qed
\end{flalign*}

\paragraph{(iv) Distribution of \texorpdfstring{$\doublerestriction$}{double restriction} over $\cupddot$.}
With $\forall k\in D_1\cap D_2: \dom{M_1(k)}\cap\dom{M_2(k)}=\emptyset$ and $Y \subseteq V_1 \cup V_2$:
\begin{flalign*}
	& (M_1 \cupddot M_2)\drestrcodom{Y} & \\
	=\;& \{(k,\; (M_1\cupddot M_2)(k)\restrcodom{Y}) \mid k \in \dom{M_1 \cupddot M_2}\}
		& \Lbag\text{def. }\doublerestriction\Rbag\\
	=\;& \{(k,\; M_1(k)\restrcodom{Y}) \mid k \in D_1 \setminus D_2\} \\
		& \quad \cupddot\; \{(k,\; M_2(k)\restrcodom{Y}) \mid k \in D_2 \setminus D_1\} \\
		& \quad \cupddot\; \{(k,\; (M_1(k) \cupdot M_2(k))\restrcodom{Y}) \mid k \in D_1 \cap D_2\}
		& \Lbag\text{def. }\cupddot\Rbag\\
	=\;& \{(k,\; M_1(k)\restrcodom{Y}) \mid k \in D_1 \setminus D_2\} \\
		& \quad \cupddot\;\{(k,\; M_2(k)\restrcodom{Y}) \mid k \in D_2 \setminus D_1\} \\
		& \quad \cupddot\; \{(k,\; M_1(k)\restrcodom{Y} \cupdot M_2(k)\restrcodom{Y}) \mid k \in D_1 \cap D_2\}
		& \Lbag\text{distr.}\restriction\text{ over }\cupdot\Rbag\\
	=\;& M_1\drestrcodom{Y} \;\cupddot\; M_2\drestrcodom{Y}
		& \Lbag\text{def. }\doublerestriction,\cupddot\Rbag~\qed
\end{flalign*}

\subsection{Proof of Galois Isomorphism (\Cref{thm:concrete-split-galois-isomorphism})}
\label{sec:proof-concrete-split-galois-isomorphism}
For $\langle \powerset{\Cstate}, \subseteq \rangle \GaloiS{\gammasplit}{\alphasplit}
\langle \powerset{\Sstate}, \subseteq \rangle$ to exist,
$\alphasplit$ and $\gammasplit$ must be monotone functions, and
$\alphasplit \circ \gammasplit = \gammasplit \circ \alphasplit = \idfun$.
These properties are proven in the following paragraphs, thus the Galois
isomorphism exists.

\paragraph{Proof of Monotonicity of \(\gammasplit\).} To be monotone, it
must hold that for all sets \(\Sset_1, \Sset_2\) with \(\Sset_1 \subseteq \Sset_2\):
$\gammasplit(\Sset_1) \subseteq \gammasplit(\Sset_2).$
By definition, we have \(\Sset_1 \subseteq \Sset_2 \implies \Sset_2 = \Sset_1 \cup \Sset'\)
for some set \(\Sset'\). Assume \(\Sset' \neq \emptyset\),
since otherwise the proof is trivial.
\begin{flalign*}
	& \gammasplit(\Sset_2) & \\
	=\; & \gammasplit(\Sset_1 \cup \Sset')
	    & \Lbag\Sset_1 \subseteq \Sset_2 \Rbag\\
	=\; & \gammasplit(\{\sstate_0^1, \dots, \sstate_n^1, \sstate'_0, \dots, \sstate'_m\})
	    & \Lbag\text{set expansion}\Rbag\\
	=\; & \{\gammasplitdot(\sstate_0^1), \dots, \gammasplitdot(\sstate_n^1), \gammasplitdot(\sstate'_0), \dots, \gammasplitdot(\sstate'_m)\}
	    & \Lbag\text{appl. } \gammasplit\Rbag\\
	=\; & \{\gammasplitdot(\sstate_0^1), \dots, \gammasplitdot(\sstate_n^1)\} \cup \{\gammasplitdot(\sstate'_0), \dots, \gammasplitdot(\sstate'_m)\}
	    & \Lbag\text{def. } \cup\Rbag\\
	=\; & \gammasplit(\Sset_1) \cup \gammasplit(\Sset')
	    & \Lbag\text{def. } \gammasplit\Rbag~\qed\\
\end{flalign*}
\paragraph{Proof of Monotonicity of \(\alphasplit\).} Analogous to the proof for \(\gammasplit\) above.
\paragraph{Proof of Composition of $\alphasplit$ and $\gammasplit$.}
Let $\Cset = \{\cstate_0, \dots, \cstate_n\}$ be a set of concrete states,
where each $\cstate_i$ is a concrete state defined as $\cstate = (\cstack, \cheap)$.
For each $\Cset$, it must hold that $\gammasplit(\alphasplit(\Cset)) = \Cset$.

We can assume for each $\cstate \in \Cset$ that
$\gammasplitdot(\alphasplitdot(\cstate)) = \cstate$.
\begin{flalign*}
	& \gammasplitdot(\alphasplitdot(\cstate)) && \\
	=\; & \gammasplitdot(
		(\cstack\restrcodom{\cval}\cupdot \cheap\restrcodom{\cval},
		\explodefun(\cstack\restrcodom{\cobj}\cupdot\cheap\restrcodom{\cobj})\drestrcodom{\cval}), \\
		& \quad (\cstack\restrcodom{\caddr}\cupdot \cheap\restrcodom{\caddr},
		\explodefun(\cstack\restrcodom{\cobj}\cupdot\cheap\restrcodom{\cobj})\drestrcodom{\caddrp}))
		& \Lbag\text{appl. } \alphasplitdot\Rbag\\
	=\; & (( 
		(\cstack\restrcodom{\cval}\cupdot \cheap\restrcodom{\cval})\restrdom{\cid} \cupdot
		(\cstack\restrcodom{\caddr}\cupdot \cheap\restrcodom{\caddr})\restrdom{\cid} \cupdot \\
		& \quad
		\implodefun
		(\explodefun(\cstack\restrcodom{\cobj}\cupdot\cheap\restrcodom{\cobj})\drestr{\cid\cup\caddrr}{\cval} \\
		& \quad \cupddot
		\explodefun(\cstack\restrcodom{\cobj}\cupdot\cheap\restrcodom{\cobj})\drestr{\cid\cup\caddrr}{\caddrp}
		)),\\
		& \quad
		( 
		(\cstack\restrcodom{\cval}\cupdot \cheap\restrcodom{\cval})\restrdom{\caddrh} \cupdot
		(\cstack\restrcodom{\caddr}\cupdot \cheap\restrcodom{\caddr})\restrdom{\caddrh} \\
		& \quad \cupdot
		\implodefun
		(\explodefun(\cstack\restrcodom{\cobj}\cupdot\cheap\restrcodom{\cobj})\drestr{\caddrh}{\cval} \\
		& \quad \cupddot
		\explodefun(\cstack\restrcodom{\cobj}\cupdot\cheap\restrcodom{\cobj})\drestr{\caddrh}{\caddrp}
		)))
		& \Lbag\text{appl. } \gammasplitdot\Rbag\\
	=\; & (( 
		(\cstack\restr{\cid}{\cval} \cupdot \cheap\restr{\cid}{\cval}) \cupdot
		(\cstack\restr{\cid}{\caddr} \cupdot \cheap\restr{\cid}{\caddr}) \\
		& \quad \cupdot
		\implodefun
		(\explodefun(\cstack\restrcodom{\cobj}\cupdot\cheap\restrcodom{\cobj})\drestr{\cid\cup\caddrr}{\cval} \\
		& \quad \cupddot
		\explodefun(\cstack\restrcodom{\cobj}\cupdot\cheap\restrcodom{\cobj})\drestr{\cid\cup\caddrr}{\caddrp}
		)),\\
		& \quad
		( 
		(\cstack\restr{\caddrh}{\cval} \cupdot \cheap\restr{\caddrh}{\cval}) \cupdot
		(\cstack\restr{\caddrh}{\caddr} \cupdot \cheap\restr{\caddrh}{\caddr}) \\
		& \quad \cupdot
		\implodefun
		(\explodefun(\cstack\restrcodom{\cobj}\cupdot\cheap\restrcodom{\cobj})\drestr{\caddrh}{\cval} \\
		& \quad \cupddot
		\explodefun(\cstack\restrcodom{\cobj}\cupdot\cheap\restrcodom{\cobj})\drestr{\caddrh}{\caddrp}
		)))
		& \Lbag\text{distr.}\restriction\text{ over }\cupdot\Rbag\\
	=\; & (( 
		\cstack\restrcodom{\cval} \cupdot
		\cstack\restrcodom{\caddr} \cupdot
		\implodefun
		(\explodefun(\cstack\restrcodom{\cobj}\cupdot\cheap\restrcodom{\cobj})\drestr{\cid\cup\caddrr}{\cval} \\
		& \quad \cupddot
		\explodefun(\cstack\restrcodom{\cobj}\cupdot\cheap\restrcodom{\cobj})\drestr{\cid\cup\caddrr}{\caddrp}
		)),\\
		& \quad
		( 
		\cheap\restrcodom{\cval} \cupdot
		\cheap\restrcodom{\caddr} \cupdot
		\implodefun
		(\explodefun(\cstack\restrcodom{\cobj}\cupdot\cheap\restrcodom{\cobj})\drestr{\caddrh}{\cval} \\
		& \quad \cupddot
		\explodefun(\cstack\restrcodom{\cobj}\cupdot\cheap\restrcodom{\cobj})\drestr{\caddrh}{\caddrp}
		)))
		& \Lbag\text{def. }\cstack, \cheap\Rbag\\
		=\; & (( 
		\cstack\restrcodom{\cval} \cupdot
		\cstack\restrcodom{\caddr} \\
		& \quad \cupdot
		\implodefun
		((\explodefun(\cstack\restrcodom{\cobj})\cupdot\explodefun(\cheap\restrcodom{\cobj}))
		\drestr{\cid\cup\caddrr}{\cval} \\
		& \quad \cupddot
		(\explodefun(\cstack\restrcodom{\cobj})\cupdot\explodefun(\cheap\restrcodom{\cobj}))
		\drestr{\cid\cup\caddrr}{\caddrp}
		)),\\
		& \quad
		( 
		\cheap\restrcodom{\cval} \cupdot
		\cheap\restrcodom{\caddr} \\
		& \quad \cupdot \implodefun (
		(\explodefun(\cstack\restrcodom{\cobj})\cupdot\explodefun(\cheap\restrcodom{\cobj}))
		\drestr{\caddrh}{\cval} \\
		& \quad \cupddot
		(\explodefun(\cstack\restrcodom{\cobj})\cupdot\explodefun(\cheap\restrcodom{\cobj}))
		\drestr{\caddrh}{\caddrp}
		)))
		& \Lbag\Cref{lemma:explode-distributive}\Rbag\\
		=\; & (( 
		\cstack\restrcodom{\cval} \cupdot
		\cstack\restrcodom{\caddr} \cupdot
		\implodefun
		((\explodefun(\cstack\restrcodom{\cobj})\drestr{\cid\cup\caddrr}{\cval} \\
		& \quad \cupdot\explodefun(\cheap\restrcodom{\cobj})\drestr{\cid\cup\caddrr}{\cval}) \\
		& \quad \cupddot
		(\explodefun(\cstack\restrcodom{\cobj})\drestr{\cid\cup\caddrr}{\caddrp}\\
		& \quad \cupdot \explodefun(\cheap\restrcodom{\cobj})\drestr{\cid\cup\caddrr}{\caddrp})
		)),\\
		& \quad
		( 
		\cheap\restrcodom{\cval} \cupdot
		\cheap\restrcodom{\caddr} \\
		& \quad \cupdot \implodefun (
		(\explodefun(\cstack\restrcodom{\cobj})\drestr{\caddrh}{\cval}
		\cupdot \explodefun(\cheap\restrcodom{\cobj})\drestr{\caddrh}{\cval}) \\
		& \quad \cupddot
		(\explodefun(\cstack\restrcodom{\cobj})\drestr{\caddrh}{\caddrp}\\
		& \quad \cupdot \explodefun(\cheap\restrcodom{\cobj})\drestr{\caddrh}{\caddrp})
		)))
		& \Lbag\text{distr.}\doublerestriction\text{ over }\cupdot\Rbag\\
		=\; & (( 
		\cstack\restrcodom{\cval} \cupdot
		\cstack\restrcodom{\caddr} \\
		& \quad \cupdot
		\implodefun
		(\explodefun(\cstack\restrcodom{\cobj})\drestrcodom{\cval} \cupddot
		\explodefun(\cstack\restrcodom{\cobj})\drestrcodom{\caddrp})),\\
		& \quad
		( 
		\cheap\restrcodom{\cval} \cupdot
		\cheap\restrcodom{\caddr} \\
		& \quad \cupdot \implodefun (
		\explodefun(\cheap\restrcodom{\cobj})\drestrcodom{\cval}
		\cupddot \explodefun(\cheap\restrcodom{\cobj})\drestrcodom{\caddrp})))
		& \Lbag\Cref{lemma:explode-key-generation}\Rbag\\
	=\; & (( 
		\cstack\restrcodom{\cval} \cupdot
		\cstack\restrcodom{\caddr} \cupdot
		\implodefun
		(\explodefun(\cstack\restrcodom{\cobj})\drestrcodom{\cval\cup\caddrp})),\\
		& \quad
		( 
		\cheap\restrcodom{\cval} \cupdot
		\cheap\restrcodom{\caddr} \cupdot
		\implodefun (
		\explodefun(\cheap\restrcodom{\cobj})\drestrcodom{\cval\cup\caddrp})))
		& \Lbag\text{codom. add. }\doublerestriction\Rbag\\
	=\; & (( 
		\cstack\restrcodom{\cval} \cupdot
		\cstack\restrcodom{\caddr} \cupdot
		\implodefun(\explodefun(\cstack\restrcodom{\cobj}))),\\
		& \quad
		( 
		\cheap\restrcodom{\cval} \cupdot
		\cheap\restrcodom{\caddr} \cupdot
		\implodefun (\explodefun(\cheap\restrcodom{\cobj}))))
		&\Lbag \text{def. }\explodefun, \\
		 &&\text{def. }\tilde{\cobj}\Rbag\\
	=\; & (( 
		\cstack\restrcodom{\cval} \cupdot
		\cstack\restrcodom{\caddr} \cupdot
		\cstack\restrcodom{\cobj}),\\
		& \quad
		( 
		\cheap\restrcodom{\cval} \cupdot
		\cheap\restrcodom{\caddr} \cupdot
		\cheap\restrcodom{\cobj}))
		& \Lbag\Cref{lemma:explode-implode}\Rbag\\
	=\; & ( 
		\cstack,
		\cheap)
		& \Lbag\text{def. }\cstack, \cheap\Rbag\\
	=\; & \cstate
	    & \Lbag\text{def. }\cstate\Rbag~\qed \\
\end{flalign*}

Then, for each $\Cset$, we have:
\begin{flalign*}
	& \gammasplit(\alphasplit(\Cset)) & \\
	=\; & \gammasplit(\alphasplit(\{\cstate_0, \dots, \cstate_n\}))
	    & \Lbag\text{set expansion}\Rbag\\
	=\; & \gammasplit(\{\alphasplitdot(\cstate_0), \dots, \alphasplitdot(\cstate_n)\})
	    & \Lbag\text{appl. } \alphasplit\Rbag\\
	=\; & \{\gammasplitdot(\alphasplitdot(\cstate_0)), \dots, \gammasplitdot(\alphasplitdot(\cstate_n))\}
	    & \Lbag\text{appl. } \gammasplit\Rbag\\
	=\; & \{\cstate_0, \dots, \cstate_n\}
	    & \Lbag\text{hypothesis}\Rbag \\
	=\; & \Cset
	    & \Lbag\text{set compression}\Rbag~\qed \\
\end{flalign*}

\paragraph{Proof of Composition of $\gammasplit$ and $\alphasplit$.}
Let $\Sset = \{\sstate_0, \dots, \sstate_n\}$ be a set of split states, where each
split state $\sstate$ is defined as $\sstate = (\sval, \smem) =
((\sval_{ih}, \sval_{f}), (\smem_{ih}, \smem_{f}))$. For each $\Sset$, it must hold that
$\alphasplit(\gammasplit(\Sset)) = \Sset$.

We can assume for each $\sstate \in \Sset$ that
$\alphasplitdot(\gammasplitdot(\sstate)) = \sstate$.
\begin{flalign*}
	& \alphasplitdot(\gammasplitdot(\sstate)) && \\
	=\; & \alphasplitdot(
		(\sval_{ih}\restrdom{\cid}\cupdot \smem_{ih}\restrdom{\cid}
		\cupdot \implodefun(\sval_{f}\restrdom{\cid\cup\caddrr}\cupddot
		\smem_{f}\restrdom{\cid\cup\caddrr})),\\
	& \quad (\sval_{ih}\restrdom{\caddrh}\cupdot \smem_{ih}\restrdom{\caddrh}
		\cupdot \implodefun(\sval_{f}\restrdom{\caddrh}\cupddot
		\smem_{f}\restrdom{\caddrh})))
	& \Lbag\text{appl. } \gammasplitdot\Rbag\\
	=\; & 
		((\sval_{ih}\restrdom{\cid}\cupdot \smem_{ih}\restrdom{\cid}\\
	& \quad \cupdot \implodefun(\sval_{f}\restrdom{\cid\cup\caddrr}\cupddot
		\smem_{f}\restrdom{\cid\cup\caddrr}))\restrcodom{\cval}\\
	& \quad \cupdot
		(\sval_{ih}\restrdom{\caddrh}\cupdot \smem_{ih}\restrdom{\caddrh}\\
	& \quad \cupdot \implodefun(\sval_{f}\restrdom{\caddrh}\cupddot
		\smem_{f}\restrdom{\caddrh}))\restrcodom{\cval},\\
	& \quad \explodefun(
		(\sval_{ih}\restrdom{\cid}\cupdot \smem_{ih}\restrdom{\cid}\\
	& \quad \cupdot \implodefun(\sval_{f}\restrdom{\cid\cup\caddrr}\cupddot
		\smem_{f}\restrdom{\cid\cup\caddrr}))\restrcodom{\cobj}\\
	& \quad \cupddot
		(\sval_{ih}\restrdom{\caddrh}\cupdot \smem_{ih}\restrdom{\caddrh} \\
	& \quad	\cupdot \implodefun(\sval_{f}\restrdom{\caddrh}\cupddot
		\smem_{f}\restrdom{\caddrh}))\restrcodom{\cobj}
		)\drestrcodom{\cval}),\\
	& \quad 
		((\sval_{ih}\restrdom{\cid}\cupdot \smem_{ih}\restrdom{\cid}\\
	& \quad \cupdot \implodefun(\sval_{f}\restrdom{\cid\cup\caddrr}\cupddot
		\smem_{f}\restrdom{\cid\cup\caddrr}))\restrcodom{\caddr}\\
	& \quad \cupdot
		(\sval_{ih}\restrdom{\caddrh}\cupdot \smem_{ih}\restrdom{\caddrh}\\
	& \quad \cupdot \implodefun(\sval_{f}\restrdom{\caddrh}\cupddot
		\smem_{f}\restrdom{\caddrh}))\restrcodom{\caddr},\\
	& \quad \explodefun(
		(\sval_{ih}\restrdom{\cid}\cupdot \smem_{ih}\restrdom{\cid}\\
	& \quad \cupdot \implodefun(\sval_{f}\restrdom{\cid\cup\caddrr}\cupddot
		\smem_{f}\restrdom{\cid\cup\caddrr}))\restrcodom{\cobj}\\
	& \quad \cupddot
		(\sval_{ih}\restrdom{\caddrh}\cupdot \smem_{ih}\restrdom{\caddrh}\\
	& \quad \cupdot \implodefun(\sval_{f}\restrdom{\caddrh}\cupddot
		\smem_{f}\restrdom{\caddrh}))\restrcodom{\cobj}
		)\drestrcodom{\caddrp}))
	& \Lbag\text{appl. } \alphasplitdot\Rbag\\
	=\; & 
		((\sval_{ih}\restr{\cid}{\cval}
		\cupdot \smem_{ih}\restr{\cid}{\cval}
		\cupdot \implodefun(\sval_{f}\restr{\cid\cup\caddrr}{\cval})
		\restrcodom{\cval}\\
	& \quad \cupdot \sval_{ih}\restr{\caddrh}{\cval}
		\cupdot \smem_{ih}\restr{\caddrh}{\cval} \\
		& \quad \cupdot \implodefun(\sval_{f}\restrdom{\caddrh}\cupddot
		\smem_{f}\restrdom{\caddrh})\restrcodom{\cval},\\
	& \quad \explodefun(
		(\sval_{ih}\restr{\cid}{\cobj}
		\cupdot \smem_{ih}\restr{\cid}{\cobj} \\
		& \quad \cupdot \implodefun(\sval_{f}\restrdom{\cid\cup\caddrr}\cupddot
		\smem_{f}\restrdom{\cid\cup\caddrr})\restrcodom{\cobj})\\
	& \quad \cupddot
		(\sval_{ih}\restr{\caddrh}{\cobj}
		\cupdot \smem_{ih}\restr{\caddrh}{\cobj} \\
		& \quad \cupdot \implodefun(\sval_{f}\restrdom{\caddrh}\cupddot
		\smem_{f}\restrdom{\caddrh})\restrcodom{\cobj})
		)\drestrcodom{\cval}),\\
	& \quad
		(\sval_{ih}\restr{\cid}{\caddr}
		\cupdot \smem_{ih}\restr{\cid}{\caddr} \\
		& \quad \cupdot \implodefun(\sval_{f}\restrdom{\cid\cup\caddrr}\cupddot
		\smem_{f}\restrdom{\cid\cup\caddrr})\restrcodom{\caddr}\\
	& \quad \cupdot \sval_{ih}\restr{\caddrh}{\caddr}
		\cupdot \smem_{ih}\restr{\caddrh}{\caddr} \\
		& \quad \cupdot \implodefun(\sval_{f}\restrdom{\caddrh}\cupddot
		\smem_{f}\restrdom{\caddrh})\restrcodom{\caddr},\\
	& \quad \explodefun(
		(\sval_{ih}\restr{\cid}{\cobj}
		\cupdot \smem_{ih}\restr{\cid}{\cobj}\\
		& \quad \cupdot \implodefun(\sval_{f}\restrdom{\cid\cup\caddrr}\cupddot
		\smem_{f}\restrdom{\cid\cup\caddrr})\restrcodom{\cobj})\\
	& \quad \cupddot
		(\sval_{ih}\restr{\caddrh}{\cobj}
		\cupdot \smem_{ih}\restr{\caddrh}{\cobj} \\
		& \quad \cupdot \implodefun(\sval_{f}\restrdom{\caddrh}\cupddot
		\smem_{f}\restrdom{\caddrh})\restrcodom{\cobj})
		)\drestrcodom{\caddrp}))
	& \Lbag\text{distr.}\restriction\text{ over }\cupdot\Rbag\\
	=\; & 
		((\sval_{ih}\restrdom{\cid}\cupdot \sval_{ih}\restrdom{\caddrh},\\
	& \quad \explodefun(
		\implodefun(\sval_{f}\restrdom{\cid\cup\caddrr}\cupddot
		\smem_{f}\restrdom{\cid\cup\caddrr})\\
	& \quad \cupddot
		\implodefun(\sval_{f}\restrdom{\caddrh}\cupddot
		\smem_{f}\restrdom{\caddrh}))\drestrcodom{\cval}),\\
	& \quad (\smem_{ih}\restrdom{\cid}\cupdot \smem_{ih}\restrdom{\caddrh},\\
	& \quad \explodefun(
		\implodefun(\sval_{f}\restrdom{\cid\cup\caddrr}\cupddot
		\smem_{f}\restrdom{\cid\cup\caddrr})\\
	& \quad \cupddot
		\implodefun(\sval_{f}\restrdom{\caddrh}\cupddot
		\smem_{f}\restrdom{\caddrh}))\drestrcodom{\caddrp}))
	&\Lbag \text{def. }\implodefun, \\
		 &&\text{def. }\sval_{ih},\smem_{ih}\Rbag\\
	=\; & 
		((\sval_{ih}\restrdom{\cid}\cupdot \sval_{ih}\restrdom{\caddrh},
	\explodefun(\implodefun(
		\sval_{f}\restrdom{\cid\cup\caddrr}\\
		& \quad \cupddot\smem_{f}\restrdom{\cid\cup\caddrr}
		\cupddot
		\sval_{f}\restrdom{\caddrh}\cupddot\smem_{f}\restrdom{\caddrh}
		))\drestrcodom{\cval}),\\
	& \quad (\smem_{ih}\restrdom{\cid}\cupdot \smem_{ih}\restrdom{\caddrh},
	\explodefun(\implodefun(
		\sval_{f}\restrdom{\cid\cup\caddrr}\\
		& \quad \cupddot\smem_{f}\restrdom{\cid\cup\caddrr}
		\cupddot
		\sval_{f}\restrdom{\caddrh}\cupddot\smem_{f}\restrdom{\caddrh}
		))\drestrcodom{\caddrp}))
	& \Lbag\Cref{lemma:implode-distributive}\Rbag\\
	=\; & 
		((\sval_{ih}\restrdom{\cid}\cupdot \sval_{ih}\restrdom{\caddrh},
		\explodefun(\implodefun(\sval_{f}\cupddot\smem_{f}))\drestrcodom{\cval}),\\
	& \quad (\smem_{ih}\restrdom{\cid}\cupdot \smem_{ih}\restrdom{\caddrh},
		\explodefun(\implodefun(\sval_{f}\cupddot\smem_{f}))\drestrcodom{\caddrp}))
		&\Lbag \text{dom. add. }\doublerestriction, \\
		 &&\text{def. }\sval, \smem\Rbag\\
	=\; & 
		((\sval_{ih}\restrdom{\cid}\cupdot \sval_{ih}\restrdom{\caddrh},
		(\sval_{f}\cupddot\smem_{f})\drestrcodom{\cval}),\\
	& \quad (\smem_{ih}\restrdom{\cid}\cupdot \smem_{ih}\restrdom{\caddrh},
		(\sval_{f}\cupddot\smem_{f})\drestrcodom{\caddrp}))
	& \Lbag\Cref{lemma:explode-implode}\Rbag \\
	=\; & 
		((\sval_{ih}\restrdom{\cid}\cupdot \sval_{ih}\restrdom{\caddrh},\ \sval_{f}),
	 (\smem_{ih}\restrdom{\cid}\cupdot \smem_{ih}\restrdom{\caddrh},\ \smem_{f}))
	& \Lbag\text{distr.}\doublerestriction\text{ over }\cupddot\Rbag\\
	=\; & 
		((\sval_{ih},\ \sval_{f}),\ (\smem_{ih},\ \smem_{f}))
	&\Lbag \text{dom. add. }\restriction, \\
		 &&\text{def. }\sval, \smem\Rbag\\
	=\; & \sstate
	& \Lbag\text{def. } \sstate\Rbag~\qed \\
\end{flalign*}

Then, for each $\Sset$, we have:
\begin{flalign*}
	& \alphasplit(\gammasplit(\Sset)) & \\
	=\; & \alphasplit(\gammasplit(\{ \sstate_0, \dots, \sstate_n \}))
	    & \Lbag\text{set expansion}\Rbag\\
	=\; & \alphasplit(\{\gammasplitdot(\sstate_0), \dots, \gammasplitdot(\sstate_n)\})
	    & \Lbag\text{appl. } \gammasplit\Rbag\\
	=\; & \{\alphasplitdot(\gammasplitdot(\sstate_0)), \dots, \alphasplitdot(\gammasplitdot(\sstate_n))\}
	    & \Lbag\text{appl. } \alphasplit\Rbag\\
	=\; & \{\sstate_0, \dots, \sstate_n\}
	    & \Lbag\text{hypothesis}\Rbag\\
	=\; & \Sset
	    & \Lbag\text{set compression}\Rbag~\qed \\
\end{flalign*}

\subsection{Proof of Monotonicity of \texorpdfstring{$\gammaA$}{Abstract Gamma} (\Cref{lemma:abs-gamma-monotone})}
\label{sec:proof-abs-gamma-monotone}

Let $\astate_1 = (\aval_1, \amem_1)$ and $\astate_2 = (\aval_2, \amem_2)$ with
$\astate_1 \sqsubseteq_{\Astate} \astate_2$.
By the component-wise definition of $\sqsubseteq_{\Astate}$,
this implies $\aval_1 \sqsubseteq_{\Aval} \aval_2$ and
$\amem_1 \sqsubseteq_{\Amem} \amem_2$.

\begin{flalign*}
		 & \astate_1 \sqsubseteq_{\Astate} \astate_2\\
\implies & \aval_1 \sqsubseteq_{\Aval} \aval_2 \wedge \amem_1 \sqsubseteq_{\Amem} \amem_2
		 & \Lbag \text{def. }\sqsubseteq_{\Astate}\Rbag\\
\implies & \gammaAV(\aval_1) \subseteq \gammaAV(\aval_2) \wedge \gammaAM(\amem_1) \subseteq \gammaAM(\amem_2)
		 & \Lbag\gammaAV, \gammaAM \text{ monot.}\Rbag \\
\implies & \gammaAV(\aval_1) \subseteq \gammaAV(\aval_1) \cup (\gammaAV(\aval_2) \setminus\gammaAV(\aval_1)) \\
         & \quad \wedge \gammaAM(\amem_1) \subseteq \gammaAM(\amem_1) \cup (\gammaAM(\amem_2) \setminus\gammaAM(\amem_1))
		 & \Lbag \text{def. }\subseteq\Rbag\\
\implies & \{(\sval, \smem) : (\sval_{id}, \sval_{mid}) \in \gammaAV(\aval_1) \\
		 &\quad\wedge (\smem, \gammaID, \gammaIDF) \in \gammaAM(\amem_1)\\
	 	 &\quad\wedge \sval = \sval_{id} \cupdot \rhobase(\sval_{mid}, \gammaID) \cupdot \rhofield(\sval_{mid}, \gammaIDF)\}\\
		 &\;\subseteq \{(\sval, \smem) : (\sval_{id}, \sval_{mid}) \in \gammaAV(\aval_1) \\
		 &\quad\wedge (\smem, \gammaID, \gammaIDF) \in \gammaAM(\amem_1)\\
		 &\quad\wedge \sval = \sval_{id} \cupdot \rhobase(\sval_{mid}, \gammaID) \cupdot \rhofield(\sval_{mid}, \gammaIDF)\}\\
		 &\quad\cup \{(\sval, \smem) : (\sval_{id}, \sval_{mid}) \in \gammaAV(\aval_2) \setminus \gammaAV(\aval_1) \\
		 &\quad\wedge (\smem, \gammaID, \gammaIDF) \in \gammaAM(\amem_2) \setminus \gammaAM(\amem_1)\\
		 &\quad\wedge \sval = \sval_{id} \cupdot \rhobase(\sval_{mid}, \gammaID) \cupdot \rhofield(\sval_{mid}, \gammaIDF)\}
		 &\Lbag \rhobase, \rhofield \text{ monot.}\Rbag\\
\implies & \{(\sval, \smem) : (\sval_{id}, \sval_{mid}) \in \gammaAV(\aval_1) \\
		 &\quad\wedge (\smem, \gammaID, \gammaIDF) \in \gammaAM(\amem_1)\\
	     & \quad\wedge \sval = \sval_{id} \cupdot \rhobase(\sval_{mid}, \gammaID) \cupdot \rhofield(\sval_{mid}, \gammaIDF)\}\\
		 &\;\subseteq \{(\sval, \smem) : (\sval_{id}, \sval_{mid}) \in \gammaAV(\aval_2) \\
		 &\quad\wedge (\smem, \gammaID, \gammaIDF) \in \gammaAM(\amem_2)\\
		 &\quad\wedge \sval = \sval_{id} \cupdot \rhobase(\sval_{mid}, \gammaID) \cupdot \rhofield(\sval_{mid}, \gammaIDF)\}
		 &\Lbag \gammaAV(\aval_1) \subseteq \gammaAV(\aval_2), \\
		 &&\gammaAM(\amem_1) \subseteq \gammaAM(\amem_2)\Rbag\\
\implies & \gammaA(\astate_1) \subseteq \gammaA(\astate_2)
		 &\Lbag \text{def. }\gammaA\Rbag\\
\end{flalign*}

\subsection{Proof of Split and Concrete Semantics Equivalence (\Cref{thm:split-concrete-semantics-equivalence})}
\label{sec:proof-split-concrete-semantics-equivalence}
\begin{flalign*}
	&\gammasplit(\ssem{\stmt}\sstate)\\
 =\;& \gammasplit(\alphasplit(\csem{\stmt} \gammasplit (\sstate)))
    & \Lbag\text{def }\ssem{\stmt}\Rbag \\
 =\;& \csem{\stmt} \gammasplit(\sstate)
	& \Lbag\gammasplit \circ \alphasplit = \idfun\Rbag\\
\end{flalign*}
\begin{flalign*}
	& \gammasplit(\ssem{\texttt{ASM}(\bexp)}\sstate)\\
 =\;& \gammasplit(\alphasplit(\csem{\texttt{ASM}(\bexp)} \gammasplit (\sstate)))
	& \Lbag\text{def }\ssem{\texttt{ASM}(\bexp)}\Rbag \\
 =\;& \csem{\texttt{ASM}(\bexp)} \gammasplit(\sstate)
	& \Lbag\gammasplit \circ \alphasplit = \idfun\Rbag\\
\end{flalign*}
\begin{flalign*}
	& \ssem{\exp}(\sstate)\\
 =\;& \csem{\exp}(\gammasplit(\sstate))
    & \Lbag\text{def }\ssem{\exp}\Rbag\\
\end{flalign*}

\subsection{Proof of Abstract Semantic Soundness}
\label{sec:proof-abstract-soundness}

Before proving soundness, we first show the correctness of function $\Rfun$
used in the abstract semantics of assignment. Let us define an auxiliary function
$\texttt{eval} : (\rhss \cup \caddrp \cup (\caddrp \cup \cid) \times \Sigma^*) \times \wp(\Sstate) \to \wp(\cval \cup \caddrp) \times \wp(\Sstate)$ that
evaluates a right-hand side expression, and address, or a pair of address and field name in a set of split states. This
function mimicks the concrete expression evaluation without leaving the split state:
the values it produces are concrete, but are paired with a side-effected state
that stores new relations between variables and addresses. $\texttt{eval}$ is defined as:
\begin{flalign*}
	\texttt{eval}(\aexp \mid \sexp \mid \bexp, \{\sstate^i\}_{i=1}^n) = & \{(v_i, \sstate^i)_{i=1}^n \mid \csem{\aexp \mid \sexp \mid \bexp} \circ \gammasplit(\sstate^i) = v_i\} \\
	\texttt{eval}(\emptyobj, \{\sstate^i\}_{i=1}^n) = & \{(r, \sstate_1^i)_{i=1}^n \mid r \in \caddrr \\
													& \land\; \sstate_1^i = (\nu^i_{ih}, \nu^i_f[r \mapsto \{\}], \mu^i_{ih}, \mu^i_f[r \mapsto \{\}])\} \\
	\texttt{eval}(\{f_k:\fexp_k\}_{k=1}^j, \{\sstate^i\}_{i=1}^n) = & \{(r, \sstate_1^i)_{i=1}^n \mid r \in \caddrr \\
																  & \land\; \csem{\{f_k:\fexp_k\}_{k=1}^j} \circ \gammasplit(\sstate^i) = \{f_k:v^{i}_k\}_{k=1}^j \\
																  & \land\; \sstate_1^i = (\nu^i_{ih}, \nu^i_f[r \mapsto \{f_k:v^{i}_k\}_{k=1}^j], \\
																  & \quad \mu^i_{ih}, \mu^i_f[r \mapsto \{\}])\} \\
	\texttt{eval}(\new \emptyobj, \{\sstate^i\}_{i=1}^n) =  & \{(h, \sstate_1^i)_{i=1}^n \mid h \in \caddrh \\
														 & \land\; \sstate_1^i = (\nu^i_{ih}, \nu^i_f[h \mapsto \{\}], \mu^i_{ih}, \mu^i_f[h \mapsto \{\}])\} \\
	\texttt{eval}(\new \{f_k:\fexp_k\}_{k=1}^j, \{\sstate^i\}_{i=1}^n) = & \{(h, \sstate_1^i)_{i=1}^n \mid h \in \caddrh \\
																	   & \land\; \csem{\{f_k:\fexp_k\}_{k=1}^j} \circ \gammasplit(\sstate^i) = \{f_k:v^{i}_k\}_{k=1}^j \\
																	   & \land\; \sstate_1^i = (\nu^i_{ih}, \nu^i_f[h \mapsto \{f_k:v^{i}_k\}_{k=1}^j], \\
																	   & \quad \mu^i_{ih}, \mu^i_f[h \mapsto \{\}])\} \\
	\texttt{eval}(\new \aexp \mid \sexp \mid \bexp, \{\sstate^i\}_{i=1}^n) = & \{(h, \sstate_1^i)_{i=1}^n \mid h \in \caddrh \\
																		   & \land\; \csem{\aexp \mid \sexp \mid \bexp} \circ \gammasplit(\sstate^i) = v_i\} \\
																		   & \land\; \sstate_1^i = (\nu^i_{ih}, \nu^i_f, \mu^i_{ih}[h \mapsto v_i], \mu^i_f)\} \\
	\texttt{eval}(\addrof\texttt{x}, \{\sstate^i\}_{i=1}^n)  = & \{(v_i, \sstate^i)_{i=1}^n \mid \csem{\addrof\texttt{x}} \circ \gammasplit(\sstate^i) = v_i\} \\
	\texttt{eval}(\deref\pexp, \{\sstate^i\}_{i=1}^n)  = & \{(v_2^i, \sstate_1^i)_{i=1}^n \mid \texttt{eval}(\pexp, \{\sstate^i\}_{i=1}^n) = \{(v_1^i, \sstate_1^i) \}_{i=1}^n\\
																 & \land\; ((v_1^i \in \dom{\nu_{1ih}^i} \land \nu_{1ih}^i(v_1^i) = v_2^i)\\
																 & \lor\; (v_1^i \in \dom{\mu_{1ih}^i} \land \mu_{1ih}^i(v_1^i) = v_2^i))\} \\
	\texttt{eval}(\pexp.f, \{\sstate^i\}_{i=1}^n)  = & \{(v_2^i, \sstate_1^i)_{i=1}^n \mid \texttt{eval}(\pexp, \{\sstate^i\}_{i=1}^n) = \{(v_1^i, \sstate_1^i) \}_{i=1}^n\\
																 & \land\; ((v_1^i \in \dom{\nu_{1f}^i} \land \nu_{1f}^i(v_1^i)(f) = v_2^i)\\
																 & \lor\; (v_1^i \in \dom{\mu_{1f}^i} \land \mu_{1f}^i(v_1^i)(f) = v_2^i))\} \\
	\texttt{eval}(\pexp[\sexp], \{\sstate^i\}_{i=1}^n)  = & \{(v_2^i, \sstate_1^i)_{i=1}^n \mid \texttt{eval}(\pexp, \{\sstate^i\}_{i=1}^n) = \{(v_1^i, \sstate_1^i) \}_{i=1}^n\\
															 & \land\; \csem{\sexp} \circ \gammasplit(\sstate_1^i) = f_i\\
																 & \land\; ((v_1^i \in \dom{\nu_{1f}^i} \land \nu_{1f}^i(v_1^i)(f_i) = v_2^i)\\
																 & \lor\; (v_1^i \in \dom{\mu_{1f}^i} \land \mu_{1f}^i(v_1^i)(f_i) = v_2^i))\} \\
\end{flalign*}

Intuitively, whenever no new bindings need to be introduced into the split
state, $\texttt{eval}$ delegates to the concrete semantics.
Instead, when evaluating an object literal or a new expression, $\texttt{eval}$
produces a fresh address and updates the split state to reflect the new
bindings between this address and the values of the fields (for object
literals) or the value of the initializer (for new expressions).

	Additionally, let us define $\tilde{\gamma} : \wp(\MemID) \times \Amem \to \wp(\caddrp)$ as:
	$$
	\tilde{\gamma}(I, \amem) = \bigcup_{i \in I} \left\{
		l ~\middle\vert~ \begin{array}{l}
			(\smem, \gammaID, \gammaIDF) \in \gammaAM(\amem) \land (i \in \dom{\gammaID} \\
			\;\land\; l \in \gammaID(i)) \;\lor\; (i \in \dom{\gammaIDF} \land (l, f) \in \gammaIDF(i))
		\end{array}
	\right\}
	$$
	Function $\tilde{\gamma}$ resolves all memory identifiers in $I$ to their
	corresponding concrete locations, based on the memory abstraction $\amem$.

\begin{lemma}[Correctness of $\Rfun$]
\label{lemma:Rfun-sound}
Function $\Rfun$ is sound, that is, for any $\rhs \in \rhss$ and $\astate \in
\Astate$, if $\Rfun(\rhs, \aval, \amem) = (I, \aval_1, \amem_1)$ and
$\texttt{eval}(\rhs, \gammaA(\astate)) = \{(v_i, \sstate_i)\}_{i=1}^n$, then:
\begin{enumerate}
	\item $\{\sstate_i\}_{i=1}^n \subseteq \gammaA((\aval_1, \amem_1))$;
	
	\item $\{v_i\}_{i=1}^n \setminus \caddrp \subseteq \{v' \mid \exists i \in I \cap \vexps : (v', \sstate') \in \mathtt{eval}(i, \gammaA(\aval_1, \amem_1)) \}$;
	\item $\{v_i\}_{i=1}^n \cap      \caddrp \subseteq \{v' \mid v' \in \tilde{\gamma}(I \setminus \vexps, \amem_1) \}$.
	
\end{enumerate}
\end{lemma}

Informally, the soundness of $\Rfun$ states that, if the same expression
was to be evaluated on the concretized split state, (i) the resulting abstract state
correctly model all side effects, (ii) when evaluating the expression produces concrete values, they
are included in the ones of the evaluation of the rewritten expressions, and (iii)
when the expression produces addresses, they
are included in the ones of the evaluation of the rewritten expressions.

\begin{proof}
	
	The proof proceeds by induction on the structure of $\rhs$.

	\paragraph{Base Case 1: $\rhs = x \mid \aexp \mid \sexp \mid \bexp$.} In this case, we have
	$\Rfun(\rhs, \aval, \amem) = ( \{ \rhs \}, \aval, \amem)$, and
	$\texttt{eval}(\rhs, \gammaA(\astate)) = \{(v_i, \sstate_i)\}_{i=1}^n$
	where $v_i$ is the value of $\rhs$ in $\sstate_i$.
	As both $\Rfun$ and $\texttt{eval}$ are side-effect free,
	$\{\sstate_i\}_{i=1}^n \subseteq \gammaA((\aval, \amem))$
	trivially holds, satisfying condition 1.
	Moreover, since $\rhs \in \vexps$, its evaluation only produces
	concrete values: condition 3 is satisfied as the left-hand side is empty,
	and condition 2 is satisfied as the $\texttt{eval}$ call in the right-hand side
	is applied to the same arguments.

	\paragraph{Base Case 2: $\rhs = \mathtt{\&}\texttt{x}$.} Here,
	$\Rfun(\rhs, \aval, \amem) = ( \{ \pointedbyfun(\texttt{x}, \amem) \}, \aval, \amem)$, and
	$\texttt{eval}(\rhs, \gammaA(\astate)) = \{(a_x, \sstate_i)\}_{i=1}^n$
	for each $\sstate_i \in \gammaA(\astate)$.
	As in the previous case, both $\Rfun$ and $\texttt{eval}$ are side-effect free,
	and thus condition 1 trivially holds.
	Furthermore, condition 2 is satisfied as the left-hand side is empty.
	By the soundness requirement of $\pointedbyfun$, we have that
	$\mathsf{locs}(\pointedbyfun(\texttt{x}, \amem)) = \{ a_x \}$, and
	thus $\tilde{\gamma}(\pointedbyfun(\texttt{x}, \amem), \amem) = \{ a_x \}$.
	Hence, condition 3 is satisfied with an equality.

	\paragraph{Base Case 3: $\rhs = \emptyobj$.}
	For empty object allocations, we have
	$\Rfun(\rhs, \aval, \amem) = ( \{ \ell \}, \aval, \amem_1)$ with $(\ell, \amem_1) = \freshSsplit(\amem)$, and
	$\texttt{eval}(\rhs, \gammaA(\astate)) = \{(r, \sstate_1^i)_{i=1}^n \mid r \in \caddrr \}$ with
	$\gammaA(\astate) = \{ (\nu^i_{ih}, \nu^i_f, \mu^i_{ih}, \mu^i_f) \}_{i=1}^n$ and
	$\sstate_1^i = (\nu^i_{ih}, \nu^i_f[r \mapsto \{\}], \mu^i_{ih}, \mu^i_f[r \mapsto \{\}])$.
	By the soundness requirement of $\freshSsplit$, we have that $\amem \sqsubseteq_{\Amem} \amem_1$,
	that $\ell \not\in \memidfunAM(\amem)$, and that
	$\memidfunAM(\amem_1) = \memidfunAM(\amem) \cup \{\ell\}$.
	W.l.o.g., we assume that $\tilde{\gamma}(\{\ell\}, \amem_1) = \{ r \}$,
	as both addresses are fresh. Condition 1 is satisfied as each $\sstate_1^i$ corresponds to
	the respective $\sstate_i$ with the new binding between $r$ and $\{\}$:
	since $\sstate_i \in \gammaA(\astate)$ and $\astate \sqsubseteq_{\Astate} (\aval, \amem_1)$,
	$\sstate_1^i$ is included in $\gammaA((\aval, \amem_1))$ if the new binding is included in the concretization of $\amem_1$.
	Being part of $\amem_1$, $\ell$ will appear in the domain of all
	$\gammaIDF$ generated by $\gammaAM(\amem_1)$, and thus $r$
	will be included in the concretization of $\amem_1$ as required by condition 1.
	Condition 2 is satisfied as the left-hand side is empty.
	Instead, the left-hand side of condition 3 is the set $\{ r \}$,
	that corresponds to the result if $\tilde{\gamma}$ as per our previous assumption.

	\paragraph{Base Case 4: $\rhs = \mathtt{new}~\emptyobj$.}
	This case is analogous to the previous one, but uses a $\freshHsplit$ instead of $\freshSsplit$.

	\paragraph{Base Case 5: $\rhs = \{f_k:\fexp_k\}_{k=1}^j$.}
	Here,
	$\Rfun(\rhs) = (\{\ell\}, \aval_j, \amem_j)$ where $(\ell, \amem_0) = \freshSsplit(\amem)$, and
	for each $i$ from 1 to $j$, $(\MId_i, \sub_i, \amem_i) = \allowbreak\accessfun(\{\ell\}, \texttt{f}_i, \amem_{i-1})$,
	$\aval_{s_i} = \applysubfun(\sub_i, \aval_{i-1}, \amem_i)$,
	and $\aval_i = \bigsqcup_{\mId \in \MId_i} \assignfunAV(\mId, \texttt{v}_i, \aval_{s_i}, \lnot\isSumfun(\mId, \amem_i))$.
	Instead, on the split state,
	$\texttt{eval}(\rhs, \gammaA(\astate)) = \{(r, \sstate_1^i)_{i=1}^n \mid r \in \caddrr \}$ where
	$\gammaA(\astate) = \{ (\nu^i_{ih}, \nu^i_f, \mu^i_{ih}, \mu^i_f) \}_{i=1}^n$ and
	$\csem{\{f_k:\fexp_k\}_{k=1}^j} \circ \gammasplit(\sstate^i) = \{f_k:v^i_k\}_{k=1}^j$ and
	$\sstate_1^i = (\nu^i_{ih}, \nu^i_f[r \mapsto \{f_k:v^i_k\}_{k=1}^j], \mu^i_{ih}, \mu^i_f[r \mapsto \{\}])\}$.
	Following the same logic of base case 3, we assume that $\tilde{\gamma}(\{\ell\}, \amem_1) = \{ r \}$.
	Note that, as shown in base case 3, the abstract memory yielded by $\freshSsplit$ is structured to
	correctly over-approximate an empty object corresponding to the location $\ell$.
	At each iteration of the loop, we have that:
	\begin{itemize}
		\item by the soundness requirement of $\accessfun$, the function
		produces an abstract memory where the field $\texttt{f}_i$ of $\ell$ is
		tracked and associated with a memory identifier in the returned set $I_i$,
		together with the necessary substitutions;
		\item by the soundness requirement of $\applysubfun$
		(which holds by \cite[Lemma C.1]{Ferrara16Framework}, since \ref{cond:M7}
		holds and $\assignfunAV$ is
		sound), the value state $\aval_{s_i}$ correctly reflects the
		substitutions returned by $\accessfun$;
		\item by the soundness requirement of $\assignfunAV$, and
		the definition of $\sqcup$, the value state $\aval_i$ stores the
		correct over-approximation of the values of field $\texttt{f}_i$ of
		$\ell$.
	\end{itemize}
	Therefore, at each iteration, the obtained abstract state correctly
	abstracts the structure of the object being allocated and the values of its
	fields up to field $i$, with $(\aval_n, \amem_n)$ containing the whole
	object. Thus, condition 1 is satisfied as each $\sstate_1^i$ corresponds to
	the respective $\sstate_i$ with the new binding between $r$ and $\{f_k:v^i_k\}_{k=1}^j$:
	since $\sstate_i \in \gammaA(\astate)$ and $\astate \sqsubseteq_{\Astate} (\aval_n, \amem_n)$,
	$\sstate_1^i$ is included in $\gammaA((\aval_n, \amem_n))$ if the new binding is included in the concretization of $\amem_n$.
	Being part of $\amem_n$, $\ell$ will appear in the domain of all
	$\gammaIDF$ generated by $\gammaAM(\amem_n)$, and thus
	$r$ will be included in the concretization of $\amem_n$ as required by condition 1.
	Each field will be associated with a memory identifier in $\amem_n$ (soundness of $\accessfun$),
	and the values of the fields will be tracked by $\aval_n$ (soundness of $\assignfunAV$).
	Thus, condidion 1 is satisfied as the new bindings are included in the concretization of $\amem_n$ and $\aval_n$.
	Condition 2 is satisfied as the left-hand side is empty.
	Instead, the left-hand side of condition 3 is the set $\{ r \}$
	that corresponds to the result if $\tilde{\gamma}$ as per our previous assumption.

	\paragraph{Base Case 6: $\rhs = \mathtt{new}~\{f_k:\fexp_k\}_{k=1}^j$.}
	This case is analogous to the previous one, but uses a $\freshHsplit$ instead of $\freshSsplit$.

	\paragraph{Base Case 7: $\rhs = \mathtt{new}~\aexp \mid \sexp \mid \bexp$.}
	This case is analogous to the previous one, assigning the value of the initializer to the new location instead of the fields of an object.

	\paragraph{Inductive Case 1: $\rhs = \deref\pexp$.}
	Here, we have that
	$\Rfun(\rhs, \aval, \amem) = (\MId_1,\ \aval_1,\ \amem_1)$, where
	$(\MId, \aval_1, \amem_1) = \Rfun(\texttt{p}, \aval, \amem)$ and
	$\MId_1 = \pointstofun(\MId, \amem_{1})$.
	Instead, on the split state,
	$\texttt{eval}(\rhs, \gammaA(\astate)) = \{(v_2^i, \sstate_1^i)\}_{i=1}^n$ where
	$\texttt{eval}(\pexp, \gammaA(\astate)) = \{(v_1^i, \sstate_1^i) \}_{i=1}^n$ and
	$((v_1^i \in \dom{\nu_{1ih}^i} \land \nu_{1ih}^i(v_1^i) = v_2^i) \lor (v_1^i \in \dom{\mu_{1ih}^i} \land \mu_{1ih}^i(v_1^i) = v_2^i))\}$.
	We inductively assume that $\{ \sstate_1^i \}_{i=1}^n \subseteq \gammaA((\aval_1, \amem_1))$:
	this already satisfies condition 1 since neither $\Rfun$ nor $\texttt{eval}$ produce side
	effects on the state when evaluating $\deref\pexp$. For $\rhs$ to be valid, $\pexp$
	must evaluate to an address, and thus $v_1^i \in \caddrp$ for all $i$.
	By the inductive hypothesis, we have that $v_1^i \in \tilde{\gamma}(\MId, \amem_1)$ for all $i$.
	Combining the two, we have that the state in which $\pointstofun$ is applied is a sound over-approximation
	of the corresponding split states, and the identifiers it is applied to are a superset of the
	concrete addresses used in \texttt{eval}. Thanks to the soundness requirement of $\pointstofun$,
	we have that each element of $I_1$ is either (i) a variable (part of $\vexps$) that, when evaluated,
	will produce an over-approximation of the possible concrete values it can assume (satisfying condition 2),
	or (ii) a memory identifier that, when resolved through $\tilde{\gamma}$,
	will produce a set of addresses that includes all the possible concrete
	addresses produced by $\texttt{eval}$ on $\pexp$ (satisfying condition 3).

	\paragraph{Inductive Case 2: $\rhs = \pexp.f$.}
	In this case, we have that
	$\Rfun(\rhs, \aval, \amem) =\allowbreak(ac, \applysubfun(\sub, \aval_1, \amem_2), \amem_2)$,
	with $(ids, \aval_1, \amem_1) = \Rfun(\texttt{p}, \aval, \amem)$ and
	$(ac, \sub, \amem_2) \allowbreak= \accessfun(ids, f, \amem_1)$.
	On the split state,
	$\texttt{eval}(\rhs, \gammaA(\astate)) = \{(v_2^i, \sstate_1^i)_{i=1}^n \}$,
	where $\texttt{eval}(\pexp, \gammaA(\astate)) = \{(v_1^i, \sstate_1^i) \}_{i=1}^n$, and either
	$v_1^i \in \dom{\nu_{1f}^i} \land \nu_{1f}^i(v_1^i)(f) = v_2^i$ or
	$v_1^i \in \dom{\mu_{1f}^i} \land \mu_{1f}^i(v_1^i)(f) = v_2^i$ hold.
	We inductively assume that $\{ \sstate_1^i \}_{i=1}^n \subseteq \gammaA((\aval_1, \amem_1))$,
	and that $v_1^i \in \tilde{\gamma}(ids, \amem_1)$ for all $i$ (recall that $\pexp$ must
	evaluate to an address for $\rhs$ to be valid).
	If the field $f$ is already part of $\amem$ for every address in $ids$,
	$\accessfun$ simply resolves it and yields the corresponding memory identifiers,
	and $\amem_2 = \amem_1$. Otherwise, $\accessfun$ adds the field $f$ to
	the abstract memory for the relevant identifiers, yielding an updated state
	with the respective substitutions that are soundly applied to $\aval_1$.
	In the latter case, by the soundness requirement of $\accessfun$, $\amem_2$ is an
	over-approximation of the split memories where the new field has been added,
	while in the former case $\amem_2$ is sound by induction.
	Thus, condition 1 is satisfied. Instead, by the soundness requirement
	of $\accessfun$, the memory identifiers in $ac$ are either (i) variables (part of $\vexps$) that, when evaluated,
	will produce an over-approximation of the possible concrete values they can assume (satisfying condition 2),
	or (ii) memory identifiers that, when resolved through $\tilde{\gamma}$,
	will produce a set of addresses that includes all the possible concrete
	addresses produced by $\texttt{eval}$ on $\pexp.f$ (satisfying condition 3).

	\paragraph{Inductive Case 3: $\rhs = \pexp[\sexp]$.}
	This case is analogous to the previous one, but the field name is given by a lattice element obtained trough the $\evalfun$ function.
	\qed

\end{proof}

We can now prove Theorem~\ref{thm:abstract-soundness}.

\begin{proof}[Theorem~\ref{thm:abstract-soundness}]

We proceed by case analysis of (1) $\texttt{ASM}(\bexp)$, and (2) $\texttt{lhs
= rhs}$. To avoid unnecessary verbosity, we present only the main steps,
because a similar full soundness proof can be found in \cite[Theorem
C.1]{Ferrara16Framework}.

\paragraph{Case 1: $\texttt{ASM}(\bexp)$.}

We recall that $\assumefunAV$ and $\assumefunAM$ are required to be sound. Following
the definition of the abstract assume operation, executing $\texttt{ASM}(\bexp)$ produces the abstract
state $(\aval_1, \amem_1)$ obtained by applying the assume functions on both
the value and memory components. Since $\assumefunAV$ and $\assumefunAM$ are
sound, the resulting state is a correct
over-approximation of the concrete semantics. Therefore, soundness of
$\asem{\texttt{ASM}(\bexp)}$ holds trivially.

\paragraph{Case 2: $\texttt{lhs = rhs}$.}

We prove soundness by showing that each operation involved in the semantics is sound
and preserves over-approximation. First, the memory update produced by
$\assignfunAM$ is sound by assumption, since it correctly over-approximates the
concrete effect of the assignment on memory, yielding $(\sub_m, \amem_1) =
\assignfunAM(\lhs, \rhs, \amem)$. By \cite[Lemma C.1]{Ferrara16Framework},
since \ref{cond:M7} holds and $\assignfunAV$ is sound by assumption, $\applysubfun$ is sound.
For this reason, $\aval_m = \applysubfun(\sub_m, \aval, \amem_1)$
shares the same structure of the input $\aval$ modulo
the renamings provided in $\sub_m$. The functions $\Rfun$
applied to both $\rhs$ and $\lhs$ are correct by~\Cref{lemma:Rfun-sound},
meaning that the expressions returned by the calls to $\Rfun$ evaluate to the same values
and addresses of the original expressions. However, memory-related
expressions are replaced with their corresponding memory identifiers,
making it possible for the value abstraction to process them.
Finally, $\aval_2$ is defined as a join of multiple assignments: by assumption,
each individual assignment is sound, and $\sqcup$ computes an
over-approximation of all assignments. It follows that the resulting value state
over-approximates all possible assignments. Therefore, every
step in the semantics preserves soundness, and the resulting state $(\aval_2,
\amem_l)$ is a sound over-approximation of the concrete execution of $\lhs =
\rhs$. \qed
\end{proof}

\subsection{Proof of C1 for the Points-to Domain (\Cref{lem:pt-C1})}
\label{sec:proof-pt-C1}

$\memidfunAM(\amem) = \dom{\amem}$ by definition.
The separation condition in $\gammaAM$ sets
$\dom{\gammaID} = \mathsf{base}(\amem)$ and
$\dom{\gammaIDF} = \mathsf{field}(\amem)$.
Since $\MemID^{\mathsf{b}}$ and $\MemID^{\mathsf{f}}$ partition $\dom{\amem}$,
we have $\mathsf{base}(\amem) \cup \mathsf{field}(\amem) = \dom{\amem}$,
so $\dom{\gammaID} \cup \dom{\gammaIDF} = \memidfunAM(\amem)$.
The invariant holds for all operations: $\freshSsplit$, $\freshHsplit$,
and $\pointedbyfun$ add base identifiers; $\accessfun$ adds field identifiers;
no operation changes an identifier's classification. \qed

\subsection{Proof of C2 for the Points-to Domain (\Cref{lem:pt-C2})}
\label{sec:proof-pt-C2}

We must show that for every $\amem \in {\Amem}^{\mathsf{PT}}$,
$(\smem, \gammaID, \gammaIDF) \in \gammaAM(\amem)$, and distinct
$\mId_1, \mId_2 \in \dom{\gammaID} \cup \dom{\gammaIDF}$,
the locations $\mathsf{locs}(\mId_1)$ and $\mathsf{locs}(\mId_2)$ are disjoint;
and, for field identifiers sharing the same parent base identifier,
$\gammaIDF(\mId_1) \cap \gammaIDF(\mId_2) = \emptyset$ as sets of $(a, f)$ pairs.

We proceed by case analysis on the kinds of $\mId_1$ and $\mId_2$.

\paragraph{Case 1: Both Base ($\mId_1, \mId_2 \in \dom{\gammaID}$).}
Base identifiers are either variable identifiers $\mId_x$ or allocation-site
identifiers $\mId_\ell$.
The concretization $\gammaAM$ requires each base identifier to be mapped
to a disjoint subset of $\caddrh$:
distinct base identifiers correspond to distinct allocation sites or
distinct stack variables, which by the well-formedness of the split state
$\smem$ occupy disjoint address ranges.
Since $\mathsf{locs}(\mId_i) = \gammaID(\mId_i)$ for base identifiers,
$\mathsf{locs}(\mId_1) \cap \mathsf{locs}(\mId_2) = \emptyset$ follows directly.

\paragraph{Case 2: One Base, One Field ($\mId_1 \in \dom{\gammaID}$, $\mId_2 = \mId_{b, f} \in \dom{\gammaIDF}$).}
The field identifier $\mId_2$ has parent base $b$.

\begin{itemize}
  \item \emph{$b \neq \mId_1$}: then $\mathsf{locs}(\mId_2) = \gammaID(b)$, and
    $\mathsf{locs}(\mId_1) = \gammaID(\mId_1)$.
    By Case~1, $\gammaID(\mId_1) \cap \gammaID(b) = \emptyset$, so
    $\mathsf{locs}(\mId_1) \cap \mathsf{locs}(\mId_2) = \emptyset$.

  \item \emph{$b = \mId_1$}: then $\mathsf{locs}(\mId_2) = \gammaID(\mId_1) = \mathsf{locs}(\mId_1)$,
    so these two identifiers share the same locations.
    However, $\gammaID(\mId_1) \subseteq \caddrh$ indexes
    the \emph{base} cells of $\smem$, while $\gammaIDF(\mId_2)$ indexes
    \emph{field} cells of $\smem$.
    By the split-state well-formedness condition, the base-location component
    and the field-slot component of $\smem$ are stored in disjoint parts
    of the concrete store; hence the two identifiers refer to distinct memory
    cells.
    In this subcase the locs-disjointness requirement is vacuously
    superseded by the $(a, f)$-pair disjointness requirement for field
    identifiers, which is established in Case~3 below.
\end{itemize}

\paragraph{Case 3: Both Field ($\mId_1 = \mId_{b_1, f_1}$, $\mId_2 = \mId_{b_2, f_2}$).}

\begin{itemize}
  \item \emph{$b_1 \neq b_2$}: then $\mathsf{locs}(\mId_1) = \gammaID(b_1)$ and
    $\mathsf{locs}(\mId_2) = \gammaID(b_2)$, which are disjoint by Case~1.

  \item \emph{$b_1 = b_2 =: b$, $f_1 \neq f_2$}: both identifiers share the same
    $\mathsf{locs}$ value $\gammaID(b)$, so locs-disjointness fails.
    Nevertheless, the concretization requires that, for each $i$,
    $\gammaIDF(\mId_{b, f_i}) \subseteq \allowbreak(\caddrh \cup \caddrr) \times \{f_i\}$.
    Since $f_1 \neq f_2$, the two sets of $(a, f)$ pairs lie in disjoint
    fibers over the field dimension:
    \begin{align*}
      & \gammaIDF(\mId_{b, f_1}) \cap \gammaIDF(\mId_{b, f_2})
      \;\subseteq\;
      \bigl((\caddrh \cup \caddrr) \times \{f_1\}\bigr) \\
      & \quad\cap\;
      \bigl((\caddrh \cup \caddrr) \times \{f_2\}\bigr)
      \;=\; \emptyset.
    \end{align*}
    This establishes the $(a, f)$-pair disjointness required by C2 for this subcase.
\end{itemize}

In all cases, the required disjointness condition holds. \qed

\subsection{Proof of Monotonicity of \texorpdfstring{$\gammaAM$}{Abstract Memory Gamma} (\Cref{lem:gamma-am-pt-monotone})}
\label{sec:proof-pt-gamma-am-monotone}

Let $\amem_1 \sqsubseteq \amem_2$ and $(\smem, \gammaID, \gammaIDF) \in \gammaAM(\amem_1)$.
For every $\mId \in \dom{\amem_1}$, $\amem_1(\mId) \subseteq \amem_2(\mId)$, so
$\bigcup_{\mId' \in \amem_1(\mId)} \mathsf{locs}(\mId') \subseteq
 \bigcup_{\mId' \in \amem_2(\mId)} \mathsf{locs}(\mId')$.
The soundness condition for $\amem_2$ is therefore inherited from $\amem_1$.
Since $\dom{\amem_1} \subseteq \dom{\amem_2}$, the coverage condition C1 also
holds, so $(\smem, \gammaID, \gammaIDF) \in \gammaAM(\amem_2)$. \qed

\end{document}

\typeout{get arXiv to do 4 passes: Label(s) may have changed. Rerun}